\documentclass[reprint,superscriptaddress, nofootinbib, amsmath,amssymb,aps,pra]{revtex4-2}

\pdfoutput=1

\usepackage[utf8]{inputenc}
\usepackage{amsthm}
\usepackage{mathrsfs}
\usepackage{physics}
\usepackage{bbm}
\usepackage{bm}
\usepackage{quantikz}
\usepackage[ruled,vlined]{algorithm2e}
\usepackage{dcolumn}
\usepackage{hyperref}
\usepackage{xcolor,soul}
\usepackage[english]{babel}
\usepackage{romannum}
\usepackage{natbib}
\usepackage{enumitem}
\usepackage{overpic}
\usepackage{graphicx,epsfig}
\usepackage{amsfonts}
\usepackage{braket}
 \usepackage{booktabs}
\usepackage{tabularx}
\usepackage{array}
\RestyleAlgo{ruled}

\makeatletter
\newcommand*{\rom}[1]{\expandafter\@slowromancap\romannumeral #1@}
\makeatother

\hypersetup{hidelinks}

\newcommand{\ph}{\overleftarrow{h}}

\DeclareMathOperator*{\argmax}{\arg\!\max}
\DeclareMathOperator{\pr}{Pr}

\usepackage{ragged2e}

\newtheorem{result}{Result}

\newtheorem{theorem}{Theorem}
\newtheorem{lemma}[theorem]{Lemma}
\newtheorem{proposition}[theorem]{Proposition}

\theoremstyle{definition}
\newtheorem{definition}[theorem]{Definition}

\newtheorem{corollary}[theorem]{Corollary}

\def\<{\langle}
\def\>{\rangle}
\definecolor{MGcolor}{RGB}{0,105,105}

\newcommand{\piopt}[0]{\pi^{\mathrm{opt}}}

\newcommand{\M}[0]{\mathsf M}

\begin{document}
\pagenumbering{arabic}

%%%%%%%%%%%%%%%%%%%%%%%%%%%%%%%%%%%%%%

\title{An Irreducible Quantum Advantage in Aligning World Models with Reality}

\author{Josep Lumbreras} 
\email{josep.lz@ntu.edu.sg}
\affiliation{Centre for Quantum  Technologies, Nanyang Technological University, Singapore}
\affiliation{Nanyang Quantum Hub, School of Physical and Mathematical Sciences, Nanyang Technological University, Singapore}
\author{Hailan Ma}
\email{hailanma0413@gmail.com}
\affiliation{Centre for Quantum  Technologies, Nanyang Technological University, Singapore}
\affiliation{Nanyang Quantum Hub, School of Physical and Mathematical Sciences, Nanyang Technological University, Singapore}
\author{Jayne Thompson}
\email{thompson.jayne2@gmail.com}
\affiliation{College of Computing and Data Science, Nanyang Technological University, Singapore}
\affiliation{Centre for Quantum  Technologies, Nanyang Technological University, Singapore}
\affiliation{Nanyang Quantum Hub, School of Physical and Mathematical Sciences, Nanyang Technological University, Singapore}

\author{Mile Gu}
\email{mgu@quantumcomplexity.org}
\affiliation{Centre for Quantum Technologies, Nanyang Technological University, Singapore}
\affiliation{Nanyang Quantum Hub, School of Physical and Mathematical Sciences, Nanyang Technological University, Singapore}

\begin{abstract}
World models provide digital simulacra of the true world, allowing agents to be trained and tested before costly real-world deployment. At each time step, they receive an action and generate an observation and reward matching the statistics of the true world. In complex environments where present outcomes depend on events far in the past, this requires memory. One might expect that, by increasing memory, we can always build a model accurately enough to align the optimal agent policies of the real and virtual worlds. We show that this is false for classical world models, even when the true world itself is classical. We construct true worlds for which every finite classical model fails along the same possible trajectory: it either loses the ability to distinguish actions when the true world clearly prefers one, or repeatedly assigns the highest expected reward to suboptimal actions. Its expected-reward estimates also retain a nonvanishing average error. In contrast, each such true world admits a quantum world model using a single qutrit that reproduces it exactly: its reward estimates and preferred actions always match those of the true world, ensuring that the optimal policies of the real and virtual worlds remain perfectly aligned.
\end{abstract}

\maketitle

From a robot navigating crowded supermarkets to an autonomous vehicle operating on busy streets, a central ambition of reinforcement learning is to train agents to navigate ever more complex environments~\cite{Sutton1998,mnih2015human,schrittwieser2020mastering,hafner2025mastering}. Such training, however, depends on data that, in physical environments, may be scarce, slow to collect, or economically costly. Meanwhile, a successful agent must be able to act in all potential scenarios, many of which are difficult to test because doing so could endanger people or equipment, or incur significant economic loss~\cite{chua2018deep,m2023model}. World models address this bottleneck by providing surrogate environments in which agents can instead be trained, stress-tested, and benchmarked~\cite{sutton1991dyna,ha2018world,hafner2020dream}. At their operational core, they accept the same actions as the environment where the agent will ultimately be deployed - the \emph{true world} - and generate corresponding observations and rewards over successive time steps. But even setting aside how such models are developed, how do physical resources fundamentally constrain what worlds they can represent?

Memory is a key constraint. Complex true worlds can be highly non-Markovian: previous actions, outcomes, and rewards can change their response to the same action in the distant future. Reproducing this contextual behavior requires memory (see Fig.~\ref{fig:wmodel}). Whenever the model receives an action and generates an outcome, it updates this memory, enabling the consequences of a future action to depend on what came before~\cite{rabiner1989tutorial,KAELBLING199899,littman2001predictive}. In complex environments, relevant context may extend over long horizons, and memory requirements can grow in tandem. Consequently, memory is a central bottleneck for building world models that remain accurate over extended timescales~\cite{robine2023transformer,deng2023s4wm,laird2025memory}.

\begin{figure}[t!]
    \centering
    \includegraphics[width=1.0\linewidth]{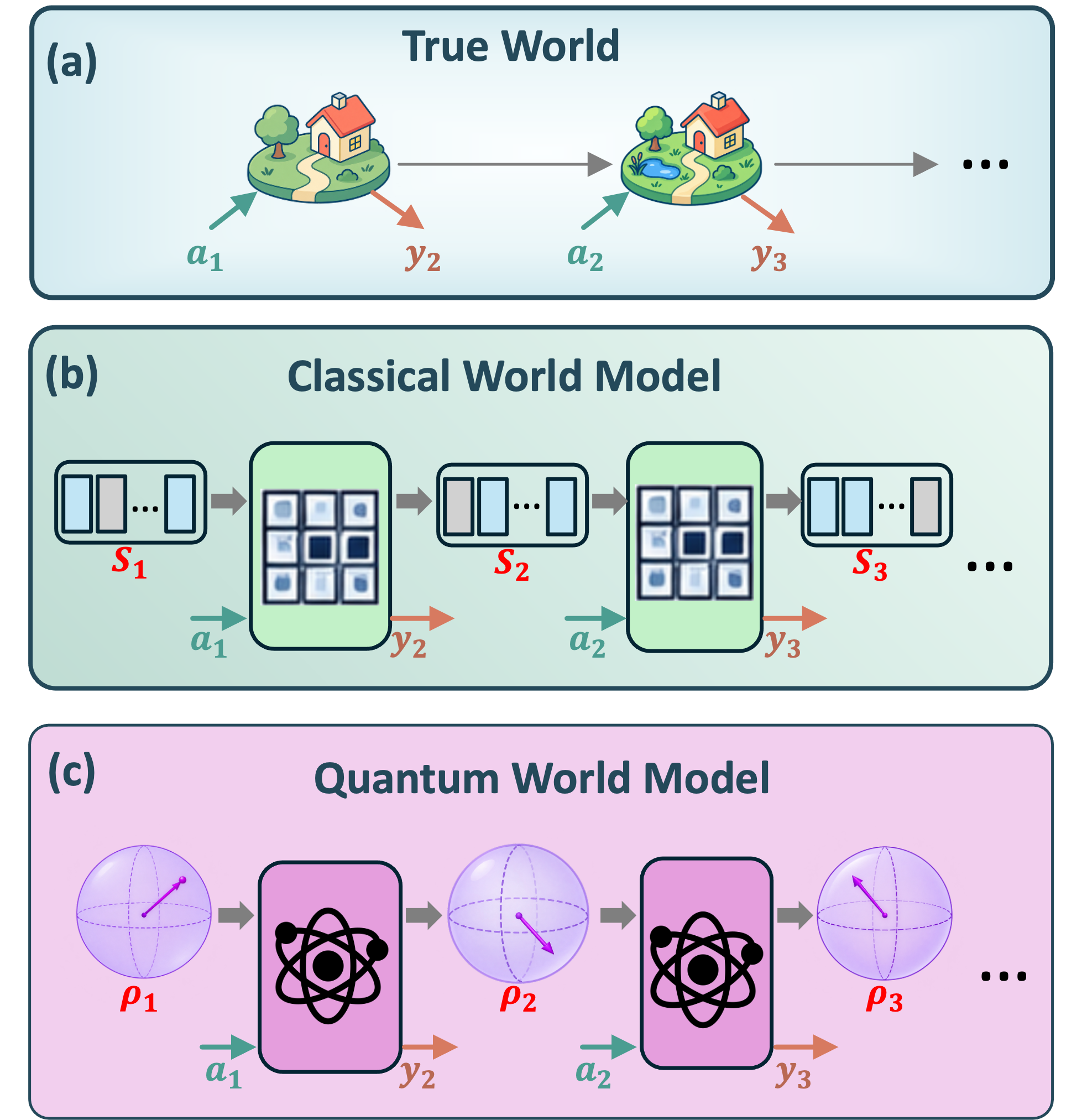}
    \caption{\textbf{Classical and quantum world models}. The true world represents reality (a). At each time step, it receives an action $a_t$ and returns an outcome $y_{t+1}$ that consists of an observation and reward. A world model reproduces this input–output behavior through recurrent internal dynamics, updating its memory from one simulated step to the next. In a classical world model (b), a classical register with a finite number of bits carries the memory, and it evolves under classical stochastic updates. In a quantum world model (c), the memory is encoded in a finite number of qubits and evolves through quantum operations.}
    \label{fig:wmodel}
\end{figure}

This memory burden has operational consequences. A world model implicitly assigns an action-value to each candidate action, representing its expected cumulative reward if taken. If memory limitations distort a world model, the resulting deviancy can erase the separation between a clearly preferable action and its competitors, or reverse their ordering altogether~\cite{grimm2020value,farahmand2011action,bellemare2016increasing}. Even when such failures occur only for specific action-outcome sequences, the effects could be disastrous. An autonomous vehicle may seldom encounter a pedestrian in its path, yet correctly ranking the merits of braking and accelerating is crucial~\cite{o2018scalable,liu2024curse}. Training or benchmarking agents on a world model that reverses this ranking can induce misalignment that adversaries could exploit to cause highly undesirable outcomes.

\emph{Could quantum world models, where context is enabled by repeated interactions with a quantum memory, offer unique advantages in removing this misalignment?} We answer in the affirmative. We introduce a family of true worlds where any finite-memory classical model exhibits an alignment gap that cannot be removed. Along certain fixed future action-outcome trajectories, the expected future rewards assigned to potential actions retain a non-vanishing average error $\Delta$ bounded away from $0$. Any classical world model either loses the ability to resolve the true world's preferred action or nominates suboptimal actions as optimal actions at least half of the time - incurring an average reward loss of at least $\Delta$. This alignment gap of $\Delta$ cannot be reduced by increasing the size of the classical memory, as long as it remains finite. We then introduce a class of quantum models with a single qutrit (see Fig.~\ref{fig:wmodel} (c)) that can avoid these limitations. Every action it judges optimal is also optimal in the true world in every reachable future trajectory. Our work thus demonstrates an irreducible quantum advantage in aligning world models with certain true worlds.

\section{Framework}

\textbf{Modeling True Worlds}. The heart of our problem is to build concise
models of true worlds---ones that reproduce their statistical black-box
behavior from the perspective of any agent interacting with them. We must
therefore first define this behavior operationally. We regard a true world (see Fig.~\ref{fig:wmodel} (a)) as
an environment with which an agent interacts over discrete time steps
$t=0,1,\ldots$. At time $t$, the world receives an action
$A_t\in\mathcal A$ and returns an outcome
$Y_{t+1}=(O_{t+1},R_{t+1})\in\mathcal Y\subseteq
\mathcal O\times\mathbb R$, consisting of an observation
$O_{t+1}\in\mathcal O$ and a reward $R_{t+1}\in\mathbb R$. Larger rewards
represent more desirable outcomes, while costs are represented by negative
rewards. Here, we take the action and outcome alphabets
$\mathcal A$ and $\mathcal Y$ to be finite.

Let $h_t = (a_{t-1},y_t)$ denote the action-outcome pair generated from the $t^{th}$ interaction. Each instance of a true world at time $t$ then has a history $\ph_t = h_1h_2\ldots h_t$. The operational behavior of a true world can then be entirely encapsulated by $\Pr_\star(y\mid \ph, a)$, the probability that it emits outcome $y$ when receiving action $a$.

A faithful world model is a machine that reproduces this conditional action-outcome behavior. A standard finite-state classical realization is a
controlled hidden Markov model~\cite{rabiner1989tutorial,KAELBLING199899,littman2001predictive} - 
a framework that underpins applications ranging from speech recognition to
biological sequence analysis
~\cite{baum1966statistical,rabiner1989tutorial,ephraim2002hidden,krogh1994hidden}. These machines contain a memory $M$ with a finite set $\mathcal{S} = \{1,\cdots , N\}$ of $N$ distinct physical states. When receiving an action $a$, their dynamics can then be completely described by transition elements
\begin{align}
\bigl(D_y^{(a)}\bigr)_{ji}
&:=
\Pr \left(Y_{t+1} = y,S_{t+1} = j|S_t = i,A_t = a\right)
\label{eq:classical-branch-matrix-main}
\end{align}
representing the probability that a machine in memory state $S_t = i$ transitions to $j$ while emitting the outcome $y$ after receiving action $a$. Given an initial probability distribution  $z^{\mathrm{in}}_k = \Pr(S_0 = k)$ over memory states and action sequence $a_0a_1\ldots$, the dynamics $\bigl(D_y^{(a)}\bigr)_{ji}$ then completely determines the model's probability of emitting outcome $\Pr_{\mathsf M}(y\mid \ph, a)$ on seeing action $a$ at time $t$ for any history $\ph$ (see Appendix~\ref{sec:classical_realization} for details). A world model is then faithful if it is operationally indistinguishable from the true world $\mathrm{Pr}_\star$, such that $\Pr_\M(y\mid \ph, a)= \Pr_\star(y\mid \ph, a)$ for all $\ph$ and $a$. Each world model is thus specified by the tuple $\mathsf M
=
\left(
\mathcal A,
\mathcal Y,
\mathcal S,
z^{\mathrm{in}},
\mathbf D
\right)$, where $\mathbf D = \{
D_y^{(a)}\}$ is the collection of $N \times N$ transition matrices describing transition dynamics on the model memory for each action-outcome pair $(a,y)$.

Physically, such a model operates as a sequence of stochastic interactions on its memory $M$ (see Fig.~\ref{fig:wmodel} (b)). Thus, $M$ enables a world model to exhibit complex non-Markovian behavior: without it, a world model's outcome behavior cannot depend on history. Thus, the number $N$ of distinct configurations available to the memory system $M$ provides a measure of the complexity required to model a true world.~\footnote{Specifically, this measure is sometimes referred to as its generative complexity for stochastic models~\cite{marzen2017nearly}. Note that it is distinct from the statistical complexity used to measure the minimal memory needed to store an agent's belief state of the environment~\cite{shalizi2001computational}. This is because environmental HMMs need not be unifilar, whereas belief state updates are~\cite{ruebeck2018}.} Operationally, $N$ is often referred to as the physical \emph{memory dimension} of $M$, defined as the largest number of states that can be
perfectly distinguished in a single use
~\cite{gu2012quantum,Gallego_2010,brunner2014dimension,ghafari2019,
hardy2001quantumtheoryreasonableaxioms}. Our description operates at the level of recurrent physical memory, whose distinguishable configurations quantify the memory resources underlying the implementation of different reinforcement-learning architectures, including neural world models.

\textbf{Rewards and Value Functions}. In practice, a perfectly faithful model can have immense generative complexity. Any practical world model is often an approximation -  identifying a model where $\mathrm{Pr}(y\mid \ph, a)$  is a sufficiently good approximation of $\mathrm{Pr}_\star(y\mid \ph, a)$. But what constitutes being sufficiently good? To answer this, we first need to review the main impetus for world models: true-world replacements for training and benchmarking agents to obtain policies that yield higher rewards.

Specifically, world models are designed to house agents. An agent's behavior is dictated by their policy: a probability distribution $\pi(a|\ph_t)$ governing what actions $a$ the agent would take on seeing history $\ph_t$. When acting on a world model $\M$ with action-outcome response $\mathrm{Pr}_M(y|\ph, a)$, each policy induces a sequence of rewards governed by random variables  $R_t|_{\pi, \M}$ governing the reward at each time-step $t$. Let $0 \leq \gamma < 1$ denote the \emph{discount factor}, which captures the preference for immediate over future rewards: when $\gamma=0$, only the next reward is considered, whereas as $\gamma \to 1$, rewards are weighted increasingly equally over time. The resulting discounted cumulative reward, or simply the \emph{return}, is represented by the random variable
\begin{equation}
G_{\pi,\M}
:=
\sum_{k=0}^{\infty}\gamma^k R_{k+1}|_{\pi,\M},
\label{eq:true-world-policy-return}
\end{equation}
which has played a dominant role in benchmarking the efficacy of a policy in reinforcement learning~\cite{Sutton1998,puterman2014markov}. We can then introduce
$\piopt_{\M}$ as the theoretical optimal policy in a given world model $\M$, one that achieves the maximum expected return $\langle G_{\pi,\M} \rangle$. Here and throughout, policies are understood as history-dependent decision
rules: the optimal policy specifies which action to take after every possible
history encountered during interaction with the world. For true worlds with input-output response $\mathrm{Pr}_\star(y|\ph,a)$, we obtain return $G_{\pi, \star}$ and true optimal policy $\piopt_{\star}$. Indeed, since the true world is operationally indistinguishable from a perfectly faithful model $\M = \star$, our subsequent exposition holds both for world models and true worlds. 

Given a particular history $\ph$, the reward potential of taking different actions can then be captured by the \emph{action-value} function. Consider the policy $\piopt_{\M}|_a$ that involves taking immediate action $a_t = a$ at history $\ph$, and thereafter making decisions according to the optimal policy $\piopt_\M$. The action-value function 
\begin{equation}
Q_{\M,a}^{\mathrm{opt}}(\ph) = \left\langle
\sum_{k=0}^{\infty}
\gamma^k R_{t+k+1}|_{\piopt_M|_a,\M,\ph} \right\rangle
\end{equation}
then captures the maximum expected reward we can potentially extract after taking a possibly non-optimal action $a$ according to world model $\M$. Taking the supremum of this quantity over all $a$ then gives us the value function 
\begin{equation}
V_{\M}^{\mathrm{opt}}(\ph) = \left\langle
\sum_{k=0}^{\infty}
\gamma^k R_{t+k+1}|_{\piopt_\M,\M,\ph} \right\rangle
\end{equation}
that represents the reward potential of that particular history $\ph$ under world model $\M$. If given for the true world $\star$, these value functions would immediately allow an agent to identify the best action at each time to maximize expected reward. %Meanwhile, the efficacy of any candidate policy or action can be quantified by how close it achieves such an optimal state. 

\textbf{Benchmarking World Models}. These quantities immediately provide operationally meaningful methods to assess candidate models. Consider a candidate model $\M$ of a true world $\star$; an obvious measure is to look at the differences in their action-value and value functions
\begin{align}\nonumber
e_{\mathrm{Q}}^{\M}(\ph)
&= \max_{a \in \mathcal{A}}
\left|
Q_{\star,a}^{\mathrm{opt}} (\ph )
-
Q_{\M,a}^{\mathrm{opt}} (\ph )
\right|, \\
e_{V}^{\M}(\ph) 
&=
\left|
V_\star^{\mathrm{opt}} (\ph )
-
V_\M^{\mathrm{opt}} (\ph )
\right| 
\label{eq:classical-optimal-value-error-main}
\end{align}
Given a particular history $\ph$, the first captures the maximum disagreement between $\M$ and $\star$ on the reward potential of various possible actions. The second captures their disagreement on the reward potential of $\ph$ itself. For such disagreements to be significant, they should not be isolated to single points in time. To formalize this, let the reachable action-outcome trajectory $\mathscr F = \ph_1,\ph_2,\ldots$ be an infinite sequence of action-outcome histories that have strictly positive conditional probability of occurrence in the true world. That is, (i) $\ph_t=h_1h_2\cdots h_t$ for some fixed sequence of
action--outcome pairs $h_t=(a_{t-1},y_t)$, and
(ii) $\Pr_\star(y_t\mid\ph_{t-1},a_{t-1})>0$ for every $t$. The quantities

\begin{align}
\overline{e}_{\mathrm{Q}}^{\M}|_\mathscr F
&:=
\liminf_{T\to\infty}
\frac{1}{T}
\sum_{t=1}^{T}
e_{\mathrm{Q}}^{\M}(\overleftarrow{h}_t),
\label{eq:mean-classical-action-value-error-main}
\\
\overline{e}_{V}^{\M}|_\mathscr F
&:=
\liminf_{T\to\infty}
\frac{1}{T}
\sum_{t=1}^{T}
e_{V}^{\M}(\overleftarrow{h}_t),
\label{eq:mean-classical-optimal-value-error-main}
\end{align}
then capture the deviation of assigned action-values and values, averaged across time, for different trajectories that an agent in the true world can experience. Thus, deploying an $\M$ where they strongly deviate to benchmark various policies could lead us to very different - and likely erroneous - conclusions about its efficacy. This motivates us to define the following measure of model deviancy:
\begin{definition} [Value Deviancy]\label{def:world-model-deviancy}
A world model $\M$ is $\Delta$-deviant if there exists a reachable action-outcome trajectory $\mathscr F$ such that $\overline{e}_{V}^{\M}|_\mathscr{F} \geq \Delta$, and $\overline{e}_{\mathrm{Q}}^{\M}|_\mathscr F \geq \Delta$. That is, its average disagreement with the true world in the reward potential along this trajectory is at least $\Delta$ for both values and action-values.
\end{definition}

A second approach to benchmarking candidate world models is to focus on how deviations can lead to different conclusions about optimal agent actions. Along these lines, we first introduce the decision margin, which captures how accurately we need to estimate action-rewards to decide the optimal action based on a world model $M$. Specifically, given a history $\ph$, define $a = a_{\M}^{opt}|_{\ph_t}$ as the action that leads to the optimal reward according to $M$, and $a' = a'_{\M}|_{\ph_t}$ as the second-best action - the one that attains maximum reward subject to the condition that $a' \neq a$. The decision margin,
\begin{equation}
g_\M(\overleftarrow{h}) = Q^{\mathrm{opt}}_{M,a}(\ph_t) -  Q^{\mathrm{opt}}_{\M,a'}(\ph_t)\label{eq:classical-decision-margin-main}
\end{equation}
represents the action-value gap between taking the best action vs its closest competitor according to the world model $\M$. When  $\M = \star$, the true action-value gap $g_\star(\overleftarrow{h})$ is considered a measure of how robust the optimal policy is in $\star$, as a perturbation of up to $g_\star(\overleftarrow{h})/2$ cannot change conclusions on optimal action~~\cite{farahmand2011action,bellemare2016increasing}. This then allows us to define another potential point of failure for a candidate classical model - \emph{loss of decision resolution}.

\begin{definition}[Loss of decision resolution]
\label{def:loss-decision-resolution}
Let $\mathscr F=h_1h_2\cdots$ be a reachable action--outcome
trajectory, with prefixes $\ph_t=h_1\cdots h_t$. For
$\Delta>0$, a world model $\M$ loses decision resolution
of magnitude $\Delta$ along $\mathscr F$ if, for every
$\delta>0$ and every $T\in\mathbb N$, there exists $t\geq T$
such that $g_\star(\ph_t)\geq \Delta$, $g_\M(\ph_t)<\delta$.
We say simply that $\M$ loses decision resolution along
$\mathscr F$ if this condition holds for some $\Delta>0$.
\end{definition}

Thus, along one possible continuing interaction, there is always improvement when taking the best action over its closest competitor in the true world. However, a model $\M$ suffering loss of decision resolution may rank them as progressively closer in value, such that agents being trained via $\M$ will find it increasingly difficult to make the correct decision. Specifically for agents that estimate action values from noisy model rollouts,
vanishing decision margins can require arbitrarily large sampling
budgets to identify the best action with fixed confidence, even at
histories where the true world retains a clear
preference~\cite{pmlr-v49-garivier16a}.

A complementary question is whether the model's preferred action is correct. Let $a_{\M} = a^{opt}_{\M}|_{\ph_t}$ be a selected action that maximize the action-value in a world model $\M$. At a given history $\ph_t$, an agent trained to perform optimally on $\M$, when deployed in the true world, would thus suffer a loss of
\begin{align}
\ell_\M (\ph_t)
&:=
V_\star^{\mathrm{opt}}(\ph_t) -
Q^{\mathrm{opt}}_{\star,a_\M} \!(\ph_t ),
\label{eq:classical-decision-loss-main}
\end{align}
even if all subsequent actions are chosen optimally. It is also commonly known as the regret of the current decision and naturally vanishes when $\M$ is faithful. Of course, an agent that performs optimally on $\M$ can continue to make non-optimal decisions, which is then captured by the
\emph{mean decision loss}:

\begin{definition}[Mean decision loss]
\label{def:mean-decision-loss}
For $\Delta>0$, a world model $\M$ exhibits mean decision loss of at least $\Delta$ along a reachable action-outcome trajectory
$\mathscr F = h_1,h_2,
\ldots$ if
\begin{align}
\overline{\ell}_{\M}(\mathscr F)
&:=
\liminf_{T\to\infty}
\frac{1}{T}
\sum_{t=1}^{T}
\ell_{\M}(\overleftarrow{h}_t) \geq \Delta
\label{eq:mean-classical-decision-loss-main}
\end{align}
\end{definition}

This quantity represents the asymptotic average loss incurred when decisions are optimized using a distorted model. Decision resolution and decision loss capture complementary requirements: a
world model should preserve a clear true-world preference, and the action it
prefers should be optimal when deployed. 

We use unweighted time averages of value errors, action-value errors, and decision loss along reachable trajectories to give equal weight to successive decision points, reflecting the importance of accurate decisions even after rare histories with critical safety implications~\cite{liu2024curse,o2018scalable}.

\section{Results}

\textbf{Complex True Worlds}. Our first result is to establish that true worlds can be very complex. Recall that the models introduced above simply as world models, denoted by $\mathsf M$, are classical, with dynamics governed by the stochastic transition matrices $D_y^{(a)}$ in~\eqref{eq:classical-branch-matrix-main}. To make this distinction explicit, we henceforth refer to them as classical world models and denote them by the tuple $\mathsf C
=
\left(
\mathcal A,
\mathcal Y,
\mathcal S,
z^{\mathrm{in}},
\mathbf D
\right)$, where $\mathcal S$ is the set of physical states contained in its memory. We say that a classical world is finite if this set is finite - a requirement that must be true for any world model that can be physically realized. Ideally, we would then be able to simulate every true world arbitrarily closely in terms of model deviancy, loss of decision resolution, and mean decision loss. Our first result establishes that for model deviancy, this gap cannot be closed (see Appendix~\ref{subsec:fdrn-planning-value-errors}):

\begin{result}\label{res:world-model-deviancy}
There exist true worlds and a fixed $\Delta>0$ such that, for each
of these worlds, every finite classical world model $\mathsf C$ is
$\Delta$-deviant.
\end{result}

This implies that there will always be action-outcome trajectories where our model $ \mathsf {C} $ predicts action-values and values that disagree with the true world by an average of at least $\Delta$ through the trajectory. 

Next, we establish that this fallibility of classical models extends to loss of decision resolution and mean decision loss. To do so, we first introduce a definition of a \emph{treacherous true world} - one that cannot be replaced by a finite-memory classical model without serious degradations when used to optimise agent actions:

\begin{definition}[Classically treacherous true world]
\label{def:classically-treacherous-world}
For $\Delta>0$, a true world is classically
$\Delta$-treacherous if there exists a reachable
action--outcome trajectory $\mathscr F$ such that, for any
finite-memory classical model $\mathsf C$, one of the following holds:
\begin{enumerate}
\item[(a)] $\mathsf C$ loses decision resolution of magnitude at least $\Delta$ along $\mathscr F$;
\item[(b)] $\mathsf C$ exhibits a mean decision loss of at least
$\Delta$ and recommends a suboptimal action at least half the
time, asymptotically, along $\mathscr F$.
\end{enumerate}
\end{definition}

We emphasize that the constant $\Delta$ applies to \emph{every} classical
memory dimension and every classical world model - and does not approach $0$ in the limit of large classical memory. In a classically treacherous true world, no finite-dimensional classical
world model reproduces the decision-relevant statistics well enough to support
reliable action selection along $\mathscr F$. Its action-values either fail
to separate the candidate actions or lead the agent to deploy persistently
suboptimal actions in the true world. Our second result establishes that classical treacherous true worlds exist:

\begin{result}\label{res:classically-treacherous-world} There exist true worlds that are classically $\Delta$ treacherous - any finite-memory classical model that attempts to model such a true world will either suffer loss of decision resolution of magnitude $\Delta$ or recommend suboptimal decisions resulting in a mean decision loss of at least $\Delta$ along certain action-outcome trajectories.
\end{result}

For the construction of such true worlds, see
Appendix~\ref{sec:fdrn-classical-gap}. The decision separation is proved in Appendix~\ref{subsec:fdrn-model-selected-decisions}. Thus, replacing such treacherous environments with classical world models for benchmarking agents, agent training, or policy planning may therefore leave the resulting decisions vulnerable to a sophisticated adversary.

\textbf{Quantum world models.} 
Can quantum models alleviate these fundamental limitations? Recall that an $N$-dimensional classical memory supports $N$ physical states, which we labeled by integers $k = 1,2,\ldots N$. In contrast, a quantum $d$-level system can prepare
such states in quantum superposition, such that each pure state of the system can be described by any normalized superposition $\ket{\phi} = \sum^d_{k=1} c_k \ket{k}$ in a $d$-dimensional Hilbert space. The basis states $\ket{k}$ are perfectly distinguishable and play the same role as the $N$ distinct configurations of a classical
memory. A general probabilistic mixture in which state $\ket{\phi_j}$ is prepared with
probability $p_j$ is described by a positive, unit-trace density operator
$\rho=\sum_j p_j\ket{\phi_j}\bra{\phi_j}$, which generalizes a classical
probability distribution over physical states.

When the quantum memory receives an action $a$, its dynamics are completely
described by outcome-labelled quantum operations
$\{\mathcal E_y^{(a)}\}_{y\in\mathcal Y}$. Each
$\mathcal E_y^{(a)}$ is completely positive and trace nonincreasing, while $\sum_{y\in\mathcal Y}\mathcal E_y^{(a)}$ is trace preserving for every $a\in\mathcal A$. Thus, for each action $a$, the operations
$\{\mathcal E_y^{(a)}\}_{y\in\mathcal Y}$ form a quantum instrument
~\cite{monras2016,fanizza2024quantum}.

To make the correspondence with the classical transition elements explicit,
let $J_t=i$ denote that the incoming quantum memory is prepared in state $\ket{i}$, and let $J_{t+1}=j$ denote the result
of reading the outgoing memory in the same basis. The instrument elements then satisfy
\begin{align}\nonumber
\bra{j}
\mathcal E_y^{(a)}
(\ket{i} \!\bra{i})\ket{j}
&:=
\Pr\left(
Y_{t+1}=y,J_{t+1}=j
\mid
J_t=i,A_t=a
\right).
\end{align}
This is the joint probability that the model emits outcome $y$ and its memory
is read as $j$, conditioned on the incoming basis state $i$ and supplied
action $a$. It is the direct quantum counterpart of
$\bigl(D_y^{(a)}\bigr)_{ji}$. Unlike in the classical case, these
basis-resolved probabilities do not completely specify the dynamics, since
the operations $\mathcal E_y^{(a)}$ also describe the evolution of
superpositions and coherences.

Given an initial memory prepared in state $\rho^{\mathrm{in}}$ and a sequence of supplied
actions, the instrument operations completely determine the model's
conditional outcome probabilities $\Pr_{\mathsf Q}(y\mid\ph,a)$ for every
generated history $\ph$. A quantum world model is faithful when it is
operationally indistinguishable from the true world, such that
$\Pr_{\mathsf Q}(y\mid\ph,a)=\Pr_\star(y\mid\ph,a)$ for every reachable
history $\ph$, action $a$, and outcome $y$. Then a \emph{quantum world model} is specified by the tuple
$\mathsf Q =
\left(
\mathcal A,
\mathcal Y,
\mathcal H_Q,
\rho^{\mathrm{in}},
\boldsymbol{\mathcal E}
\right)$,
where
$\boldsymbol{\mathcal E}
=
\{\mathcal E_y^{(a)}\}_{(a,y)\in\mathcal A\times\mathcal Y}$
is the collection of instrument operations describing the memory dynamics
for every action--outcome pair $(a,y)$.

Physically, such a model operates through a sequence of quantum-instrument
interactions on its memory, in direct analogy with the stochastic
interactions of a classical world model. The number $d$ is its physical
memory dimension: the largest number of memory states that can be perfectly
distinguished in a single use.

\textbf{Quantum advantage for value estimation.}
We now combine the preceding classical limitations with an exact quantum realization and ask how accurately the model reproduces the values. The separation is robust: no finite classical memory, however large, can eliminate the dimension-independent gap in the action-values and optimal value, whereas a $3$-dimensional quantum world model reproduces them exactly. The advantage therefore cannot be overcome merely by allocating more finite classical storage; it arises from the physical encoding of the model's memory.

\begin{result}
\label{res:finite_quantum_planning_advantage}
For some fixed $\Delta>0$, there exists a family of true worlds such that, for each world, every finite-memory classical world model is $\Delta$-deviant, whereas the corresponding single-qutrit quantum world model $\mathsf Q$ is exact.
\end{result}

\textbf{Quantum advantage for model-based decisions.}
The previous result concerns the values and action-values reproduced by the models and establishes a strict separation between classical and quantum world models for the same true world. We now establish the corresponding separation for the action ultimately deployed by the agent. A world model generates simulated futures from which the agent evaluates its candidate actions; the action assigned the largest value is then selected for deployment in the true world. Our final result establishes a quantum advantage that no finite classical memory can overcome: a single qutrit preserves true world optimal decisions, whereas every finite classical world model fails to identify the optimal actions reliably, either by losing decision resolution or by persistently recommending suboptimal actions.

\begin{result}
\label{res:finite_quantum_decision_advantage}
For the same family of true worlds as in Result~\ref{res:finite_quantum_planning_advantage}, each world is classically $\Delta$-treacherous, yet admits an exact quantum world model $\mathsf Q$ with a single qutrit of memory. This separation holds for all discount factors $\gamma\in[0,1)$, with the same $\Delta$.
\end{result}

The last two results describe complementary consequences of the same representational limitation. The first concerns the numerical action-values used to compare the candidate actions, while the second establishes the corresponding separation for the actions selected using the model. Increasing the size of a finite classical memory may postpone the discrepancy to longer histories, but cannot remove its asymptotic average. In contrast, the same three-dimensional quantum memory reproduces the conditional dynamics and action-values exactly, enabling true world optimal model-based decisions.

The Methods Section~\ref{sec:methods} gives an overview of the true-world dynamics, why every finite classical model fails to reproduce them, and how we construct the corresponding exact quantum world model. The true
worlds family is defined formally in Appendix~\ref{subsec:fdrn-controlled-world}. The decision separation is
proved in Appendix
Subsection~\ref{subsec:fdrn-model-selected-decisions}, while the action-value
and optimal-value bounds are proved together in Appendix
Subsection~\ref{subsec:fdrn-planning-value-errors}. The exact qutrit
world model is constructed and verified in Appendix
Section~\ref{sec:fdrn-quantum-realization}.

\section{Discussion}

Here, we showed how a world model processes information fundamentally changes its capacity to align with the true world - both in estimating reward potential and identifying the actions needed to realize it. Every finite-dimensional classical model unavoidably assigns erroneous action-values along certain action--outcome trajectories, with a mean error bounded away from zero independently of memory dimension. Along these trajectories, each model either makes vanishing distinctions between candidate actions when the true world has a clear preference or selects a suboptimal action at least half the time. It therefore cannot reliably align model-optimal decisions with reality, and increasing its memory cannot remove this misalignment. A quantum model with a single qutrit, by contrast, reproduces every action-conditioned future after every reachable history, yielding exact values and true-world-optimal actions. To the best of our knowledge, this is the first work to establish such an irreducible quantum advantage in world-model alignment, separating a fixed finite quantum system from all finite-dimensional classical counterparts. The comparison applies at the level of recurrent physical memory in general controlled stochastic input--output machines, encompassing finite hidden-state and partially observable Markov decision process-style generative models~\cite{rabiner1989tutorial,KAELBLING199899,cassandra_94} and connecting to recurrent methods for partial observability and finite-precision architectures with finite-automaton characterizations~\cite{hausknecht2015deep,zhang2019learning,korsky2019computational}.

The broader consequence of this separation is that the physical medium used to store a world model's internal state cannot always be treated as merely an implementation detail. Modern research has made major advances in how such states are learned and used for prediction and control, from Dyna-style model-based reinforcement learning to recurrent latent simulators~\cite{sutton1991dyna,ha2018world,pmlr-v97-hafner19a,hafner2020dream,hafner2025mastering}, with recent extensions to language-based agents~\cite{yu2026reinforcement}. A complementary line learns predictive representations directly in latent space~\cite{lecun2022path,assran2023self,bardes2024revisiting,maes2026leworldmodel}. Despite their different architectures and training objectives, these approaches generally assume that the learned state is carried by a conventional classical memory. Our results show that this assumption can impose a fundamental limit: there exist true worlds for which no representation supported by finite
classical memory can preserve alignment, regardless of how it is parameterized or learned, whereas the single-qutrit model preserves it exactly. Our separation therefore identifies a physical resource that complements advances in neural representation learning rather than competing with them. Learning determines what internal state a world model constructs; its physical encoding can determine what that state can faithfully represent. A natural next step is to harness this alignment advantage within learned neural world models and understand what our results imply under realistic noise, finite-shot estimation, and training constraints.

We also emphasize that our results concern the modeling of entirely classical true worlds: although their internal memories are quantum, our quantum world models interact with agents entirely through classical random variables representing actions, observations, and rewards. The resulting quantum advantage is therefore directly relevant to settings in which conventional world models are used. Beyond this, our alignment advantage may provide another building block towards quantum-enhanced reinforcement learning. Quantizing an agent's internal memory and processing can yield memory and energetic advantages that grow without bound for suitable families of strategies~\cite{elliot_qagent,thompson_energeticadvantage}.
Meanwhile, coherent quantum access to a world model or environment can enable different trajectories to be explored in superposition, leading to speed-ups in learning~\cite{zeng2023quantum,dunjko_enhanced16}. It would be exciting to determine how these advantages can be combined, enabling potential simultaneous speedups in learning, world-model alignment, and reduced memory and energy costs during inference.

\section{Methods}
\label{sec:methods}

\subsection{The resettable clock}\label{subsec:methods-renewal-clock-world}

Our results use the same true world, and we now state how it is constructed. The true world is built around a resettable clock. After every reset, it
independently draws a hidden lifetime $L\in\mathbb N_0$. If allowed to
continue, the clock produces $L$ consecutive Ticks followed by a Break, at
which point it resets and draws a new lifetime. Viewed first as an uncontrolled process, each reset starts a new independent run. A run consists of $L$ Ticks followed by a Break and therefore lasts $L+1$ steps. The Breaks are renewal events, with independent and identically
distributed inter-renewal times $L+1$. The underlying clock is thus a discrete-time renewal process. 

The renewal process above describes how the clock behaves when each run is
simply allowed to unfold. To obtain the true world used in our
results, we retain this renewal law but introduce actions that determine
whether the current run advances, is tested, or is deliberately reset.
Whenever a reset occurs, the clock independently draws a fresh lifetime
from the same distribution. The action set is
$\mathcal A=\{W,M,P\}$, and the observation set is
$\mathcal O=\{T,B\}$, where $T$ denotes a Tick and $B$ the end of
the current run. At each step, the clock receives one action and returns an
observation together with an action-dependent reward.

Immediately after a reset, the clock has age $t=0$. Wait ($W$) allows the
current run to continue for one further step. If the clock returns a Tick
($T$), the run survives and its age increases from $t$ to $t+1$. If it
returns a Break ($B$), the run ends and the age resets to zero. Thus, the
clock's age is the number of consecutive Wait--Tick pairs since its most
recent reset.

Maintain ($M$) represents preventive maintenance. It incurs a fixed cost,
ends the current run, and resets the age to zero without testing whether the
run would have survived another step. Probe ($P$) performs precisely this
test: at age $t$, it has the same Tick--Break probabilities as Wait, but the
clock resets after either observation.

The three actions therefore offer distinct ways of interacting with the
clock. Maintain accepts a certain cost in exchange for an immediate reset.
Probe acts as a one-step wager on the remaining lifetime: a Tick can produce
a positive reward, whereas a Break can incur a cost. Probe may therefore be preferable when another Tick is sufficiently likely, while Maintain may be preferable when a Break is more likely. Wait has a different role because a
Tick leaves the current run active and exposes the clock at the next age.
The complete outcome kernel and reward assignment are specified in
Appendix~\ref{subsec:fdrn-controlled-world}. See Fig.~\ref{fig:FRDN} for an illustrative summary of this true world.

Let $\mathscr F_{\mathrm{tick}}=h_1h_2\cdots$ denote the trajectory in
which every action--outcome pair $h_t$ records a Wait action followed by a
Tick. The history after $t$ interactions is
$\ph_t=h_1h_2\cdots h_t$, with $\ph_0=\emptyset$, and therefore
contains $t$ consecutive Wait--Tick interactions following a reset. This
is the trajectory used in all four results. Every finite history $\ph_t$ is
reachable, and the decision at clock age $t$ is evaluated conditional on
this history. We denote such histories as $\ph_t^{\mathrm{tick}}$.

Given $\ph_t^{\mathrm{tick}}$, the preceding Ticks imply that the hidden lifetime satisfies
$L\geq t$. If the next action is Wait or Probe, the conditional probability
of returning another Tick is therefore
\begin{align}
S(t)
&:=
\mathrm{Pr}_\star(T\mid\ph_t^{\mathrm{tick}},W)
=
\mathrm{Pr}_\star(T\mid\ph_t^{\mathrm{tick}},P)
\nonumber\\
&=
\mathrm{Pr}(L\geq t+1\mid L\geq t)
=
\frac{\Pr(L\geq t+1)}{\Pr(L\geq t)}.
\label{eq:fdrn-conditional-tick}
\end{align}
Thus, $S(t)$ is the age-dependent statistic that enters the evaluation of
the available actions at $\ph_t$. A world model need not store $S(t)$
explicitly, but its memory and readout must jointly reproduce this dependence
to predict the next observation and assign the actions their correct values.

This construction defines a family of true worlds, indexed by the
lifetime law and the rewards. We study a concrete member inspired by early examples of stochastic processes with finite-dimensional linear
representations~\cite{dharmadhikari1963sufficient,fox1968,dharmadhikari1970}. Fix $\lambda\in(0,1/2]$ and
$\alpha\in\mathbb R$ with $\alpha/\pi\notin\mathbb Q$, and set
\begin{align}
\Pr(L=\ell)
&=
\lambda^\ell
\sin^2\!\left(\frac{\ell\alpha}{2}\right),
\qquad \ell\geq1, 
\label{eq:sinusoidal-renewal-law-main}
\end{align}
and $\Pr(L=0)
=
1-\sum_{\ell=1}^{\infty}\Pr(L=\ell)$. Substituting the above lifetime law~\eqref{eq:sinusoidal-renewal-law-main} into the conditional survival
probability~\eqref{eq:fdrn-conditional-tick} gives $S(t)
=
f(t\alpha)$ for $ t\geq1$, where $f$ is a continuous, nonconstant, and $2\pi$-periodic function. Thus,
the sinusoidal dependence of the lifetime distribution is inherited by the conditional probability of observing another Tick, producing an oscillatory dependence on the clock age. The closed form of $f$ and its derivation are given in Appendix~\ref{subsec:fdrn-controlled-world}.

The irrational phase increment of discrete ages $t$ in $f(t\alpha)$ prevents the age dependence from synchronizing with any finite cycle. Even if the clock ages are divided into any finite collection of regularly repeating classes, the probabilities
within each class continue to explore the full profile of $f$. Operationally,
the prediction-relevant information carried by the clock therefore never
reduces to a finite periodic label. 

The rewards translate this predictive structure into decisions. For the reward assignment used in our main results, the true optimal action continues to switch between Probe and Maintain as the clock ages: Probe is preferred when $S(t)$ lies above a fixed threshold, and Maintain when it lies below. This switching persists for every discount factor $\gamma\in[0,1)$. It is not essential that these particular two actions compete. Other reward choices for the same renewal clock make the optimal action switch between Wait and Maintain instead, as discussed in Appendix~\ref{subsec:fdrn-model-selected-decisions}. Thus, the fact that Wait is dominated in our main construction is a convenient way of isolating a clean decision boundary, rather than a structural property of the world.

\begin{figure}
    \centering
    \includegraphics[width=1\linewidth]{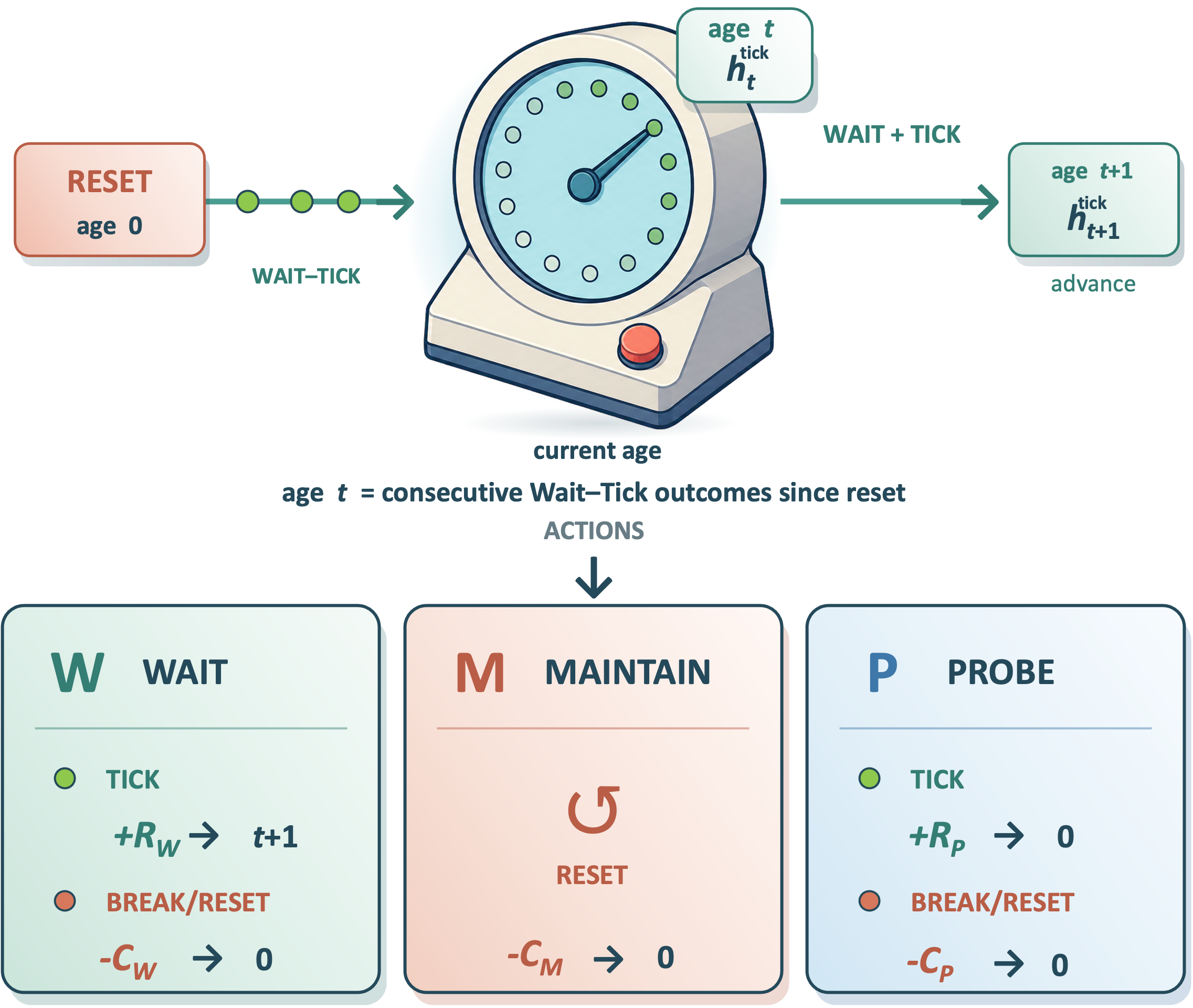}
    \caption{\textbf{The resettable FRDN clock world.}
    The displayed signs illustrate a representative payoff regime: Wait offers a reward $+R_W$ if the clock advances but risks a failure cost $-C_W$; Maintain accepts the known cost $-C_M$ of an immediate reset; and Probe also resets, with reward $+R_P$ or cost $-C_P$ determined by whether the clock would have advanced. The diagram thus contrasts risky continuation under Wait, a certain outcome under Maintain, and an outcome-dependent bet under Probe, creating a nontrivial decision problem.}
    \label{fig:FRDN}
\end{figure}

\subsection{Why finite classical memory fails and a qutrit succeeds}
\label{subsec:methods-classical-quantum-separation}

The relevant contrast is how the two types of memory respond to the
increasing clock age along the fixed reachable trajectory
$\mathscr F_{\mathrm{tick}}$. Every additional Wait--Tick interaction
applies the same update to a classical world model's memory.
Perron--Frobenius theory implies that repeated application of this update to
any finite classical memory eventually approaches a finite collection of
limiting behaviours
~\cite{seneta2006nonnegative,puterman2014markov}. Thus, for some finite
spacing $p$, the model's conditional predictions converge separately along
the interlaced sequences of ages $t=np+r$, with
$r\in\{0,\ldots,p-1\}$.

The true clock does not settle into such a finite pattern. As noted above, the conditional probability of one more Tick satisfies $S(t)=f(t\alpha)$, where $f$ is continuous, nonconstant, and $2\pi$-periodic. Since $\alpha/\pi\notin\mathbb Q$, for every finite spacing $p$, the phase increment $p\alpha$ remains irrational relative to $2\pi$. Weyl equidistribution~\cite{kuipers2012uniform} therefore implies that the phases $(np+r)\alpha$ explore the full circle within every interlaced sequence of ages. Consequently, $S(np+r)$ continues to sample
the same nonconstant profile of $f$ rather than converging to one limiting
probability.

The action-dependent rewards transfer this persistent age dependence to the
action-values and values. In our construction, Probe and Maintain both reset
the clock, so their future contributions are the same and their comparison
retains the oscillation of $S(t)$. The value obtained by selecting the
highest-valued action retains a corresponding age dependence. This remains
true for every discount factor $\gamma\in[0,1)$. Combining these facts with
Perron--Frobenius theory and Weyl equidistribution gives a positive lower
bound on the asymptotic mean action-value and value errors of every finite
classical world model. A common bound can be chosen independently of the
memory dimension, the particular classical model, and the discount factor.
Increasing the finite memory may postpone the discrepancy to later clock
ages, but cannot make either asymptotic mean error vanish. This is the
$\Delta$-deviancy of
Definition~\ref{def:world-model-deviancy}, established in
Result~\ref{res:world-model-deviancy}.

The same mismatch also affects which action is assigned the highest value.
The true preferred action switches between Probe and Maintain as the clock
ages. Every finite classical world model must therefore either become
arbitrarily indecisive at increasingly late histories where the true
preference remains separated by a fixed positive amount, or recommend a
truly suboptimal action at least half of the time along
$\mathscr F_{\mathrm{tick}}$ and incur a positive mean decision loss. This
is the $\Delta$-treachery of
Definition~\ref{def:classically-treacherous-world}, established in
Result~\ref{res:classically-treacherous-world}.

In contrast, the same clock is reproduced by a concrete world model whose
memory is a single qutrit. This model tracks the increasing clock age through
a phase in its quantum memory. After the first Tick following Wait, the
memory lies in the two-dimensional subspace spanned by
$\{|0\rangle,|1\rangle\}$. Let $\Pi_{01}:=
|0\rangle\langle0|
+
|1\rangle\langle1|$
be the projector onto this subspace, and let $X$ and $Z$ denote the
corresponding Pauli operators. The unitary phase rotation is
$U_\alpha=e^{i\alpha Z/2}$, while the complete instrument operation
associated with a Tick following Wait is
\begin{align}
\mathcal E_{\mathrm{tick}}^{(W)}(\rho)
&:=
\lambda
\Bigl(
e^{-rX}U_\alpha e^{rX}\Pi_{01}
\Bigr)
\rho
\Bigl(
e^{-rX}U_\alpha e^{rX}\Pi_{01}
\Bigr)^\dagger .
\label{eq:qutrit-clock-phase-update-main}
\end{align}
Here $U_\alpha$ advances the phase by $\alpha$. The surrounding factors
$e^{\pm rX}$ adjust the relative amplitudes so that the phase also determines
the Tick probability. Without them, the map
$\rho\mapsto\lambda U_\alpha\rho U_\alpha^\dagger$ would give the constant
Tick probability $\lambda$. Because these factors are inverses, adjacent
copies cancel when the operation is repeated, allowing the phase rotations
to accumulate. The parameter $r$ is chosen so that the resulting
probabilities reproduce~\eqref{eq:fdrn-conditional-tick}; its value is
derived in Appendix~\ref{subsec:fdrn-quantum-realization}. Since
$\alpha/\pi$ is irrational, the accumulated phase never closes into a finite
cycle. Different clock ages are instead associated with generally
nonorthogonal qutrit states along this phase orbit.

Along $\mathscr F_{\mathrm{tick}}$, this is precisely the operation applied
after every Wait--Tick interaction. The qutrit memory therefore follows the
same histories $\ph_t$ along which the finite classical models are assessed.
The remaining instrument operations are constructed in
Appendix~\ref{subsec:fdrn-quantum-realization}, and their agreement with the
true world after every reachable history is verified in
Appendix~\ref{sec:exact_qutrit_worldmodel}. The finite-classical
decision-loss and value-error bounds are proved in
Appendices~\ref{subsec:fdrn-model-selected-decisions}
and~\ref{subsec:fdrn-planning-value-errors}, respectively.

\subsection{Numerical illustration of the separation}

We illustrate the separation using the clock with $\lambda=0.4$ and
$\alpha=\pi/\sqrt{2}$. For each memory dimension $N$, we fit a classical
world model to the conditional Tick probabilities $S(t)$ defined in~\eqref{eq:fdrn-conditional-tick}. The fit uses clock ages
$t=0,\ldots,1000$. Since the optimization is nonconvex, we report the best models found numerically. Appendix~\ref{app:FRDN-classical-numerical} gives the fitting and evaluation details. Figure~\ref{fig:classical_benchmark-1} compares the fitted predictions
with $S(t)$. Larger classical memories track the true probabilities
to later clock ages. Eventually, however, the predictions of the
fitted models shown approach a single value $c_N$, while the true
probability continues to oscillate.

We evaluate decisions using the same reward assignment that makes
the clock classically treacherous
(Appendix~\ref{subsec:fdrn-model-selected-decisions}).
Under these rewards, Probe is optimal when $S(t)>\lambda$, and Maintain
when $S(t)<\lambda$. The true optimal action therefore continues to
switch with clock age, whereas a fitted model with $c_N\neq\lambda$
eventually recommends one fixed action.
Figure~\ref{fig:classical_benchmark-wrong action} shows the cumulative
fraction of queried ages at which the fitted models recommend a
suboptimal action. For every fitted model shown, this fraction
approaches $1/2$, illustrating the second form of classical treachery:
persistent recommendations of wrong actions.

The appendix also considers classical models that approximate the
clock's irrational phase advance by a rational rotation. These models
continue to switch between Probe and Maintain, but their predictions
and recommended actions repeat after each cycle. The true optimal
action follows no fixed repeating cycle. Larger memories allow longer
cycles and more accurate phase approximations, so the models track
the true action changes for longer. Eventually, however, accumulated
phase error leads to wrong recommendations.

Thus, increasing classical memory delays decision errors but does
not eliminate them. The qutrit world model reproduces $S(t)$ exactly
and recommends an optimal action at every queried age, giving a
wrong-action fraction of zero.

\begin{figure}
	\centering
\centerline{\includegraphics[width=1\linewidth]{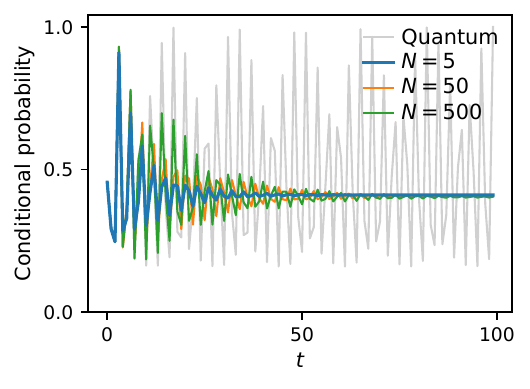}}
	\caption{\textbf{Conditional Tick probabilities.}
Probability of one additional Tick after $\ph_t$ for
$\lambda=0.4$ and $\alpha=\pi/\sqrt{2}$. The gray curve is the true-world
probability, which the qutrit world model reproduces exactly. The colored
curves show the predictions of the fitted classical models with memory
dimensions $N=5$ (blue), $N=50$ (orange), and $N=500$ (green). Increasing
$N$ allows the model to reproduce the nonrepeating dependence on the clock
age over a larger initial range. At later ages, the predictions approach the
finite limiting behavior imposed by finite classical memory.}

\label{fig:classical_benchmark-1}
\end{figure}

\begin{figure}
	\centering
\centerline{\includegraphics[width=1\linewidth]{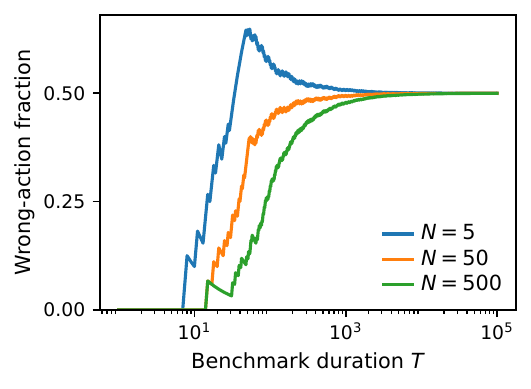}}
\caption{\textbf{Cumulative wrong-action fraction under the reward setting in Appendix~\ref{subsec:fdrn-model-selected-decisions}.} Blue, orange, and green correspond to
$N=5$ (blue), $N=50$ (orange), and $N=500$ (green), respectively. The qutrit world model has zero wrong actions for every $T$; its curve coincides with the horizontal axis and is omitted for clarity.
}
\label{fig:classical_benchmark-wrong action}
\end{figure}

\section*{Acknowledgements}

JL thanks Alessandro Luongo and Aditya Chidambaram for their help in revising the manuscript. This work is supported by the National Research Foundation of Singapore through the NRF Investigatorship Program (Award No. NRF-NRFI09-0010), the National Quantum Office, hosted in A*STAR, under its Centre for Quantum Technologies Funding Initiative (S24Q2d0009), the RIE 2025 AQAS projects S25Q9D001 and S25Q9D002, the Singapore Ministry of Education Tier 1 Grant RT4/23 and RG91/25, the RIE25 Japan-Singapore Joint Call on Quantum (Project ID H25-MRO3490) and the NRF AI-for-Science (AI4S) Catalytic Grant NRF-AI4SCT-2025-0005. HM gratefully acknowledges support from Schmidt Sciences, LLC through the Eric and Wendy Schmidt AI in Science Postdoctoral Fellowship.

\paragraph*{Use of generative AI.}
The authors derived the main technical results and produced main manuscript drafts. OpenAI's ChatGPT and Codex were used to assist with correctness checks and proofreading. All content was reviewed and verified by the authors.

\section*{Code Availability Statement} 

Code and data for reproducing the results in this work are available at \url{https://github.com/tuliplan/quantum-world-model-RL}.

\appendix

\section*{Appendix}

\noindent\textbf{Roadmap.}
Section~\ref{sec:world_models} develops the general framework for true worlds
and classical and quantum world models. In particular,
Theorem~\ref{thm:value_equivalence_exact_world_models} shows that exact world
models preserve optimal values, action values and optimal decisions.
Section~\ref{sec:fdrn-classical-gap} then constructs the true
world and the fixed reachable all-Tick trajectory
$\mathscr F_{\mathrm{tick}}$ used to establish all four results.
The proofs proceed as follows:
\begin{itemize}

\item \textbf{World-model deviancy.}
Theorems~\ref{thm:fdrn_q_gap} and~\ref{thm:fdrn_value_gap}, proved in
Subsection~\ref{subsec:fdrn-planning-value-errors}, show that every finite
classical world model has mean action-value and optimal-value errors bounded
below by a common positive constant. This proves
Result~\ref{res:world-model-deviancy}.

\item \textbf{Classically treacherous true worlds.}
Theorem~\ref{thm:fdrn-greedy-action-obstruction} and
Corollaries~\ref{cor:fdrn-loss-decision-resolution}
and~\ref{cor:fdrn-decision-regret-gap}, proved in
Subsection~\ref{subsec:fdrn-model-selected-decisions}, show that every finite
classical world model either loses decision resolution or incurs persistent
decision loss through suboptimal actions. This proves
Result~\ref{res:classically-treacherous-world}.

\item \textbf{Quantum advantage for model-based decisions.}
Section~\ref{sec:fdrn-quantum-realization} constructs a world model with a
single qutrit and proves in
Theorem~\ref{thm:qutrit_realization_exactplanning} that it reproduces the
FRDN true world exactly after every reachable history. Combining this
exact realization with the classical treachery result proves
Result~\ref{res:finite_quantum_decision_advantage}.

\item \textbf{Quantum advantage for value estimation.}
Combining the classical value-error bounds with the same exact qutrit
realization, whose action-value and optimal-value errors vanish after every
reachable history, proves
Result~\ref{res:finite_quantum_planning_advantage}.

\end{itemize}
Section~\ref{app:FRDN-classical-numerical} gives the numerical procedures
used for the illustrations.

\noindent\textbf{Notation.}
Throughout the appendix,
$\mathbb N=\{1,2,\ldots\}$,
$\mathbb N_0=\{0,1,2,\ldots\}$, and
$[N]=\{1,\ldots,N\}$.
For a finite set $\mathcal X$, let $\mathcal{P}(\mathcal X)$
denote the set of probability distributions on $\mathcal X$.
In particular, $\mathcal{P}([N])$ is the $(N-1)$-dimensional
probability simplex.
Finally, we write
\begin{align}
(x)_+
&:=
\max\{x,0\},
&
\mathbbm 1\{E\}
&:=
\begin{cases}
1, & \text{if $E$ holds},\\
0, & \text{otherwise},
\end{cases}
\end{align}
for the positive part of $x\in\mathbb R$ and the indicator of an event or
condition $E$, respectively.

\section{World models and value functions}
\label{sec:world_models}

For the definitions below, use the following fixed notation. Let
$\mathcal A$ be a finite action set, let
$\mathcal Y\subseteq\mathcal O\times\mathbb R$ be a finite alphabet of
observation--reward outcomes. For
$y=(o_y,r_y)\in\mathcal Y$, write $o_y$ for its observation component and
$r_y$ for its reward component. We assume that rewards are bounded:
$|r_y|\leq R_{\max}$ for every $y\in\mathcal Y$.

For $t\geq0$, define
\begin{align}\label{eq:history_set}
    \mathscr{H}_t
    &:=
    (\mathcal A\times\mathcal Y)^t,
    &
    \mathscr H
    &:=
    \bigcup_{t\geq0}\mathscr{H}_t .
\end{align}
A length-$t$ history prefix can be written as
\begin{align}
\ph_t
&=
h_1h_2\cdots h_t
=
(a_0,y_1,\ldots,a_{t-1},y_t),\\
h_k
&:=
(a_{k-1},y_k), \quad
|\ph_t|
=
t,
\end{align}
with $\ph_0=\emptyset$. When no time index is needed, we write
$h\in\mathscr H$ for a generic finite history.

For $h\in\mathscr H_s$ and $g\in\mathscr{H}_t$, write
$h\mathbin{\smallfrown}g\in\mathscr H_{s+t}$ for their concatenation.
For a one-step extension by $a\in\mathcal A$ and $y\in\mathcal Y$, we
continue to use the shorter notation $hay$.

A deterministic history string carries no true-world or world-model label.
Its generating source is instead indicated by the probability law or by the
corresponding random variable. Uppercase letters denote random variables and
lowercase letters their realizations.

\subsection{World models}

We first define the true world---the environment in reinforcement-learning
terminology---and the world model used as a
conditional simulator.

\begin{definition}\label{def:true-controlled-world}
A \emph{true world} is specified by the tuple
$\mathsf W_\star=(\mathcal A,\mathcal Y,\mathscr H,\mathrm{Pr}_\star)$, where
\begin{enumerate}[label=(\roman*)]
    \item $\mathcal A$ is the finite set of actions that the true world can receive;
    
    \item $\mathcal Y\subseteq\mathcal O\times\mathbb R$ is the finite set of
    possible observation--reward outcomes $y=(o,r)$;
    
    \item $\mathscr H$ is the common finite action--outcome history space defined above~\eqref{eq:history_set};
    
    \item
    $\mathrm{Pr}_\star:\mathscr H\times\mathcal A\to\mathcal{P}(\mathcal Y)$ is the
    conditional outcome law, so that
    $\mathrm{Pr}_\star(y\mid h,a)$ is the probability of outcome $y$ after history
    $h$ when the agent takes action $a$.
\end{enumerate}
Its reachable history tree $\mathscr H_\star\subseteq\mathscr H$ is the
smallest set containing $\emptyset$ such that $hay\in\mathscr H_\star$
whenever $h\in\mathscr H_\star$, $a\in\mathcal A$, $y\in\mathcal Y$, and
$\mathrm{Pr}_\star(y\mid h,a)>0$.
\end{definition}

\begin{definition}[World model]
\label{def:world_model}
A
\emph{world model} is specified by the tuple
\begin{align}
\mathsf M
=
\left(
\mathcal A,
\mathcal Y,
\mathcal M,
m^{\mathrm{in}},
\mathrm{Pr}_M,
T_M
\right),
\end{align}
where $\mathcal M$ is the state space of the model's recurrent memory,
$m^{\mathrm{in}}\in\mathcal M$ is its initial memory state, $\mathcal{A}$ is the set of actions, $\mathcal Y\subseteq\mathcal O\times\mathbb R$ is the set of observations-rewards,
\begin{align}
\mathrm{Pr}_M:
\mathcal M\times\mathcal A
\longrightarrow
\mathcal{P}(\mathcal Y)
\end{align}
is its conditional outcome law, and
\begin{align}
T_M
=
\left\{
T_{M,y}^{a}:\mathcal M\longrightarrow\mathcal M
\right\}_{(a,y)\in\mathcal A\times\mathcal Y}
\end{align}
is its family of outcome-conditioned memory updates. Thus, from memory state
$m$ and supplied action $a$, the model generates
$Y\sim \mathrm{Pr}_M(\cdot\mid m,a)$ and updates its memory to
$T_{M,Y}^{a}(m)$. The value of $T_{M,y}^{a}(m)$ on a branch for which
$\mathrm{Pr}_M(y\mid m,a)=0$ may be chosen arbitrarily.
\end{definition}

This abstract representation isolates the recurrent input--output interface
needed to define value functions. At this level, $\mathrm{Pr}_M$ and $T_M$ are
operational maps and need not be specified independently in a physical
realization. For the classical and quantum world models considered below,
they are induced, respectively, by the transition matrices $D_y^{(a)}$ and
the instrument operations $\mathcal E_y^{(a)}$, which directly describe the
physical dynamics of the model's memory.

An \emph{encoder} for $\mathsf M$ is a separate initialization interface
\begin{align}
E_M:
\mathscr H
\longrightarrow
\mathcal M,
\qquad
E_M(\emptyset)
=
m^{\mathrm{in}}.
\end{align}
For a world model query following history $h$, the encoder supplies the memory
state $E_M(h)$ from which the rollout is initialized. Once initialized, the rollout evolves
entirely through $\mathrm{Pr}_M$ and $T_M$.

The encoder belongs to the agent and provides the interface through which the
world model is queried. For a query history $h$, it initializes the model's
recurrent memory in the state $E_M(h)$, after which the model generates the
rollout through its internal dynamics. The memory dimension of a world model
refers exclusively to this recurrent physical memory: it is the largest
number of memory states that can be perfectly distinguished without error in
a single use~\cite{gu2012quantum,Gallego_2010,brunner2014dimension,ghafari2019}.

A world model is consistent and exact with respect to a true world if $\mathrm{Pr}_M$ reproduces the outcome probability of the true world and both
descriptions of the outcome-conditioned memory updates $T_M$ and the encoder $E_M$ remain synchronized on reachable true-world branches. We formalize this notion below.

\begin{definition}[Predictive consistency and exactness]
\label{def:predictive_consistency}
Let $\mathsf M$ be a world model and let $E_M$ be an encoder for it. The pair
$(\mathsf M,E_M)$ is \emph{predictively consistent} with
$\mathsf W_\star$ on $\mathscr G\subseteq\mathscr H$ if, for every
$h\in\mathscr G$, $a\in\mathcal A$, and $y\in\mathcal Y$,
\begin{align}
\mathrm{Pr}_M(y\mid E_M(h),a)
=
\mathrm{Pr}_\star(y\mid h,a).
\label{eq:predictive-consistency-probability}
\end{align}
The pair is \emph{exact} on $\mathscr G$ if it is predictively consistent
there and, whenever $\mathrm{Pr}_\star(y\mid h,a)>0$,
\begin{align}
T_{M,y}^{a}\!\left(E_M(h)\right)
=
E_M(hay).
\label{eq:predictive-consistency-update}
\end{align}
It is an \emph{exact realization} of the true world if it is exact on the
reachable history tree $\mathscr H_\star$. When the encoder is fixed by
context, we use the shorter statement that $\mathsf M$ is exact relative to
$E_M$.
\end{definition}

\subsection{Value functions in world models}
\label{sec:values_world_model}

Throughout this subsection, fix an abstract recurrent world model
$\mathsf M$ together with an encoder $E_M$. The history-indexed rollout laws
and value functions therefore depend on the pair $(\mathsf M,E_M)$. For
notational economy, we suppress the encoder dependence from their subscripts.

The world model generates simulated outcomes, whereas the agent converts
the resulting reward sequences into scores or value functions in order to evaluate its current policy. Now we formally define the concept of policy.

\begin{definition}
\label{def:history-policy}
A policy is a stochastic kernel from the
complete classical history space to the action set, equivalently a map
\begin{align}
    \pi:
    \mathscr H
    \longrightarrow
    \mathcal{P}(\mathcal A),
    \qquad
    g
    \longmapsto
    \pi(\cdot\mid g).
\end{align}
Thus, for every $g\in\mathscr H$,
\begin{align}
    \pi(a\mid g)
    &\geq
    0,
    &
    \sum_{a\in\mathcal A}\pi(a\mid g)
    &=
    1.
\end{align}
Let
\begin{align}
    \Pi_{\mathscr H}
    :=
    \left\{
        \pi:
        \mathscr H\longrightarrow\mathcal{P}(\mathcal A)
    \right\}
    \label{eq:history-policy-class}
\end{align}
denote the class of all such policies.
\end{definition}

The same policy $\pi\in\Pi_{\mathscr H}$ is used in the true-world and
world-model continuation laws below. During simulation, it is conditioned on
the complete query history extended by the generated continuation.

Fix a query history $h\in\mathscr H$ and a policy
$\pi\in\Pi_{\mathscr H}$. A model-generated rollout begins from
\begin{align}
    H_0^M
    &:=
    h,
    &
    M_0
    &:=
    E_M(h).
\end{align}
At rollout depth $k\geq0$,
\begin{align}
    A_k^M
    &\sim
    \pi(\cdot\mid H_k^M),
    \label{eq:first_recurrent}
    \\
    Y_{k+1}^M
    =
    (O_{k+1}^M,R_{k+1}^M)
    &\sim
    \mathrm{Pr}_M(\cdot\mid M_k,A_k^M),
    \\
    H_{k+1}^M
    &:=
    H_k^M A_k^M Y_{k+1}^M,
    \\
    M_{k+1}
    &:=
    T_{M,Y_{k+1}^M}^{A_k^M}(M_k).
    \label{eq:history-policy-model-process}
\end{align}
Thus, the policy conditions on the accumulated classical history $H_k^M$,
whereas the world model generates its next outcome using only its recurrent
memory $M_k$.

The conditional true-world benchmark begins from the same query history,
\begin{align}
    H_0^\star
    &:=
    h,
\end{align}
and evolves according to
\begin{align}
    A_k^\star
    &\sim
    \pi(\cdot\mid H_k^\star),
    \\
    Y_{k+1}^\star
    =
    (O_{k+1}^\star,R_{k+1}^\star)
    &\sim
    \mathrm{Pr}_\star(\cdot\mid H_k^\star,A_k^\star),
    \\
    H_{k+1}^\star
    &:=
    H_k^\star A_k^\star Y_{k+1}^\star .
    \label{eq:history-policy-true-process}
\end{align}

More explicitly, let
\begin{align}
    \zeta_n
    =
    (a_0,y_1,\ldots,a_{n-1},y_n)
\end{align}
be a deterministic continuation of length $n\in\mathbb{N}_0$ for a history $h$. For $1\leq k\leq n$, let
\begin{align}
    \zeta_k
    &:=
    (a_0,y_1,\ldots,a_{k-1},y_k),
    &
    \zeta_0
    &:=
    \emptyset
\end{align}
denote its length-$k$ prefix. The complete history at rollout depth $k$ is
then $h\mathbin{\smallfrown}\zeta_k$.

Define the model-memory state associated with this continuation recursively
\begin{align}
    m_0(h,\zeta_0)
    &:=
    E_M(h),
    \nonumber\\
    m_{k+1}(h,\zeta_{k+1})
    &:=
    T_{M,y_{k+1}}^{a_k}
    \bigl(m_k(h,\zeta_k)\bigr),
    \label{eq:recurrent_histories_memories}
\end{align}
for $0\leq k<n.$ The probability of $\zeta_n$ under the world model rollout is
\begin{align}
    \pr_{M,\pi}^{h}(\zeta_n)
    =
    \prod_{k=0}^{n-1}
    \pi\!\left(
        a_k
        \mid
        h\mathbin{\smallfrown}\zeta_k
    \right)
    \mathrm{Pr}_M\!\left(
        y_{k+1}
        \mid
        m_k(h,\zeta_k),a_k
    \right).
    \label{eq:history-policy-model-trajectory-law}
\end{align}
The corresponding true-world probability continuation law after $h$ is
\begin{align}
    \pr_{\star,\pi}^{h}(\zeta_n)
    =
    \prod_{k=0}^{n-1}
    \pi\!\left(
        a_k
        \mid
        h\mathbin{\smallfrown}\zeta_k
    \right)
    \mathrm{Pr}_\star\!\left(
        y_{k+1}
        \mid
        h\mathbin{\smallfrown}\zeta_k,a_k
    \right).
    \label{eq:history-policy-true-trajectory-law}
\end{align}
The continuation laws above determine the corresponding rollout
expectations, denoted by
$\mathbb E_{M,\pi}^{h}$ and $\mathbb E_{\star,\pi}^{h}$.

For $X\in\{M,\star\}$ and $a\in\mathcal A$, we write
$\mathbb E_{X,\pi}^{h,a}$ for expectation under the corresponding rollout
initialized at $h$, with only the first action-selection rule replaced by
\begin{align}
    A_0=a
    \qquad\text{almost surely}.
\end{align}
All outcome-generation and memory-update rules remain unchanged, and the
policy $\pi$ supplies the actions from rollout depth $k=1$ onward. In the fixed-first-action laws, $\pi(\cdot\mid h)$ is not used: the policy
controls only the actions at rollout depths $k\geq1$.

Now we define the policy values for the true world and world model. Unless stated otherwise, all quantities are defined for a fixed
discount factor $\gamma\in[0,1)$.

\begin{definition}[Policy values]
\label{def:model_policy_value_functions}
Let $\pi\in\Pi_{\mathscr H}$, $h\in\mathscr H$, $\gamma \in [0,1)$ and
$a\in\mathcal A$. For $X\in\{M,\star\}$, define
\begin{align}
    V_X^\pi(h)
    &:=
    \mathbb E_{X,\pi}^{h}
    \left[
        \sum_{k=0}^{\infty}\gamma^kR_{k+1}
    \right],
    \\
    Q_X^\pi(h,a)
    &:=
    \mathbb E_{X,\pi}^{h,a}
    \left[
        \sum_{k=0}^{\infty}\gamma^kR_{k+1}
    \right].
\end{align}
Thus, $V_X^\pi(h)$ follows $\pi$ from the first action, whereas
$Q_X^\pi(h,a)$ takes $a$ first and follows $\pi$ thereafter.
\end{definition}

For later use, define the expected one-step rewards of the model in memory
state $m$ and of the true world after history $h$ by
\begin{align}
    r_M(m,a)
    &:=
    \sum_{y\in\mathcal Y}
    \mathrm{Pr}_M(y\mid m,a)r_y,
    \\
    r_\star(h,a)
    &:=
    \sum_{y\in\mathcal Y}
    \mathrm{Pr}_\star(y\mid h,a)r_y.
    \label{eq:one-step-expected-rewards}
\end{align}
When $\gamma=0$, the action-values reduce to these expected one-step
rewards.

\subsubsection{Optimal values and Bellman equations}
\label{subsec:optimal-values-bellman}

The material in this subsection is standard in reinforcement learning and
dynamic programming. We briefly recall the definitions and results needed
below, including Bellman recursions, contraction and fixed-point
characterizations, and their connection with policy optimization. For
complete treatments and rigorous proofs, see~\cite{Sutton1998,puterman2014markov,
bertsekas2012dynamic,szepesvari2010algorithms}.

The true world is Markov when the complete history is used as its state,
whereas a rollout of the world model is Markov in its recurrent memory state.
Accordingly, for bounded functions
$u:\mathscr H\to\mathbb R$ and $v:\mathcal M\to\mathbb R$, define
\begin{align}
    (\mathcal T_\star u)(h)
    :=
    \max_{a\in\mathcal A}
    \Bigg[
        &r_\star(h,a)
        \nonumber\\
        &+
        \gamma
        \sum_{y\in\mathcal Y}
        \mathrm{Pr}_\star(y\mid h,a)u(hay)
    \Bigg],
    \label{eq:true-optimal-bellman-operator}
\end{align}
and
\begin{align}
    (\mathcal T_Mv)(m)
    :=
    \max_{a\in\mathcal A}
    \Bigg[
        &r_M(m,a)
        \nonumber\\
        &+
        \gamma
        \sum_{y\in\mathcal Y}
        \mathrm{Pr}_M(y\mid m,a)
        v(T_{M,y}^a(m))
    \Bigg].
    \label{eq:optimal-bellman-operator}
\end{align}
Since rewards are bounded and $\gamma<1$, both operators are
$\gamma$-contractions in the supremum norm. They therefore have unique bounded
fixed points. Denote these fixed points
temporarily by $U_\star$ and $u_M$,
\begin{align}
    U_\star
    &=
    \mathcal T_\star U_\star,
    &
    u_M
    &=
    \mathcal T_Mu_M.
    \label{eq:preoptimal-bellman-fixed-points}
\end{align}

We next identify these fixed points using the policy values defined above.
Let
\begin{align}
    G_n
    :=
    \sum_{k=0}^{n-1}\gamma^kR_{k+1}
\end{align}
be the $n$-step return and $G_\infty:=\sum_{k=0}^{\infty}\gamma^kR_{k+1}$. Backward induction gives
\begin{align}
    (\mathcal T_\star^n0)(h)
    &=
    \sup_{\pi\in\Pi_{\mathscr H}}
    \mathbb E_{\star,\pi}^{h}[G_n],
    \label{eq:true-finite-horizon-value-iteration}
    \\
    (\mathcal T_M^n0)(E_M(h))
    &=
    \sup_{\pi\in\Pi_{\mathscr H}}
    \mathbb E_{M,\pi}^{h}[G_n],
    \label{eq:model-finite-horizon-value-iteration}
\end{align}
where $0$ denotes the zero function. Indeed, conditioning on the first action
and outcome produces the Bellman recursions: the first action is optimized,
and after each outcome the later actions may be chosen separately on the
resulting history branch. Conversely, these choices define a single history
policy because each generated outcome is included in the history observed by
the policy.

If $|R_{k+1}|\leq R_{\max}$ for every $k\geq0$,
\begin{align}
    |G_\infty-G_n|
    \leq
    \frac{R_{\max}\gamma^n}{1-\gamma}
\end{align}
uniformly over policies and initial conditions. Note that since $\gamma\in[0,1)$ we have $G_\infty$ bounded. Hence the finite-horizon
suprema converge to the corresponding infinite-horizon suprema. Since value
iteration also converges to the unique fixed points,
\begin{align}
    U_\star(h)
    &=
    \sup_{\pi\in\Pi_{\mathscr H}} V_\star^\pi(h),
    \label{eq:true-fixed-point-policy-supremum}
    \\
    u_M(E_M(h))
    &=
    \sup_{\pi\in\Pi_{\mathscr H}} V_M^\pi(h).
    \label{eq:model-fixed-point-policy-supremum}
\end{align}
The same argument with the first action fixed gives the corresponding
action-value identities. We can therefore introduce the optimal quantities
without ambiguity.

\begin{definition}[Optimal true-world value functions]
\label{def:true-optimal-planning-values}
For $h\in\mathscr H$ and $a\in\mathcal A$, define
\begin{align}
    V_\star^{\mathrm{opt}}(h)
    &:=
    \sup_{\pi\in\Pi_{\mathscr H}}
    V_\star^\pi(h),
    \label{eq:true-optimal-value}
    \\
    Q_\star^{\mathrm{opt}}(h,a)
    &:=
    \sup_{\pi\in\Pi_{\mathscr H}}
    Q_\star^\pi(h,a).
    \label{eq:true-optimal-action-value}
\end{align}
\end{definition}

The fixed-point identification established above gives
\begin{align}
    V_\star^{\mathrm{opt}}(h)
    =
    U_\star(h).
\end{align}
Consequently, the true-world optimal action-value satisfies
\begin{align}
    Q_\star^{\mathrm{opt}}(h,a)
    =
    r_\star(h,a)
    +
    \gamma
    \sum_{y\in\mathcal Y}
    \mathrm{Pr}_\star(y\mid h,a)
    V_\star^{\mathrm{opt}}(hay),
    \label{eq:true-optimal-bellman-relations}
\end{align}
and
\begin{align}
    V_\star^{\mathrm{opt}}(h)
    =
    \max_{a\in\mathcal A}
    Q_\star^{\mathrm{opt}}(h,a).
    \label{eq:true-optimal-value-action-value-relation}
\end{align}

Now we define the optimal values in the world model.

\begin{definition}[Optimal world-model value functions]
\label{def:model-optimal-planning-values}
For $h\in\mathscr H$ and $a\in\mathcal A$, define
\begin{align}
    V_M^{\mathrm{opt}}(h)
    &:=
    \sup_{\pi\in\Pi_{\mathscr H}}
    V_M^\pi(h),
    \label{eq:model-optimal-value}
    \\
    Q_M^{\mathrm{opt}}(h,a)
    &:=
    \sup_{\pi\in\Pi_{\mathscr H}}
    Q_M^\pi(h,a).
    \label{eq:model-optimal-action-value}
\end{align}
\end{definition}

The world-model Bellman operator acts on memory space. We therefore
write its fixed point and associated action-value function as
\begin{align}
    v_M^{\mathrm{opt}}(m)
    &:=
    u_M(m),
    \\
    q_M^{\mathrm{opt}}(m,a)
    &:=
    r_M(m,a)
    +
    \gamma
    \sum_{y\in\mathcal Y}
    \mathrm{Pr}_M(y\mid m,a)
    v_M^{\mathrm{opt}}
    \bigl(T_{M,y}^a(m)\bigr).
    \label{eq:bellman_action_value_optimal}
\end{align}
They satisfy
\begin{align}
    v_M^{\mathrm{opt}}(m)
    =
    \max_{a\in\mathcal A}
    q_M^{\mathrm{opt}}(m,a).
    \label{eq:model-optimal-value-relation}
\end{align}

The history-indexed value functions are related to these memory-state
functions through the encoder
\begin{align}
    V_M^{\mathrm{opt}}(h)
    &=
    v_M^{\mathrm{opt}}(E_M(h)),
    \label{eq:model-optimal-value-bridge}
    \\
    Q_M^{\mathrm{opt}}(h,a)
    &=
    q_M^{\mathrm{opt}}(E_M(h),a).
    \label{eq:model-optimal-q-bridge}
\end{align}
Hence
\begin{align}
    V_M^{\mathrm{opt}}(h)
    =
    \max_{a\in\mathcal A}
    Q_M^{\mathrm{opt}}(h,a).
\end{align}

All suprema above are pointwise in the displayed history $h$ and range
over $\Pi_{\mathscr H}$. For a value $V^{\mathrm{opt}}(h)$, the policy
selects the first and all subsequent actions. For an action-value
$Q^{\mathrm{opt}}(h,a)$, the displayed first action $a$ is fixed and the
policy selects the actions from rollout depth $k=1$ onward. Thus Bellman optimality defines a function of the current state. For \emph{the true world, the state is
the complete history}. For the model, the optimal history-indexed quantities
depend on the supplied history through the memory state $E_M(h)$.

\subsection{Exact models preserve value functions}

The next theorem shows that exactness preserves the complete trajectory law
under every history policy, and therefore preserves both policy-evaluation and
optimal value functions.

\begin{theorem}
\label{thm:value_equivalence_exact_world_models}
Let $\mathsf M$ be a world model and let $E_M$ be an encoder such that
$(\mathsf M,E_M)$ is an exact realization of $\mathsf W_\star$ in the sense
of Definition~\ref{def:predictive_consistency}. Then, for every
$\pi\in\Pi_{\mathscr H}$, every reachable query history
$h\in\mathscr H_\star$, every $n\geq0$, and every continuation $\zeta_n$,
\begin{align}
    \pr_{M,\pi}^{h}(\zeta_n)
    =
    \pr_{\star,\pi}^{h}(\zeta_n).
    \label{eq:exact-model-finite-trajectory-equality}
\end{align}
The analogous equality holds when the first action is forced to $a\in\mathcal A$ in both rollouts.
Consequently,
\begin{align}
    V_M^\pi(h)
    &=
    V_\star^\pi(h),
    &
    Q_M^\pi(h,a)
    &=
    Q_\star^\pi(h,a).
\end{align}
Moreover,
\begin{align}
    V_M^{\mathrm{opt}}(h)
    &=
    v_M^{\mathrm{opt}}(E_M(h))
    =
    V_\star^{\mathrm{opt}}(h),
    \\
    Q_M^{\mathrm{opt}}(h,a)
    &=
    q_M^{\mathrm{opt}}(E_M(h),a)
    =
    Q_\star^{\mathrm{opt}}(h,a).
\end{align}
Consequently, the sets of optimal actions also agree
\begin{align}
    \argmax_{a\in\mathcal A}
    Q_M^{\mathrm{opt}}(h,a)
    =
    \argmax_{a\in\mathcal A}
    Q_\star^{\mathrm{opt}}(h,a).
\end{align}
\end{theorem}

\begin{proof}
Fix $h\in\mathscr H_\star$ and
$\pi\in\Pi_{\mathscr H}$. For a continuation prefix $\zeta_k$, set
\begin{align}
    g_k
    &:=
    h\mathbin{\smallfrown}\zeta_k,
    &
    \mu_k
    &:=
    m_k(h,\zeta_k).
\end{align}
We prove simultaneously that the model and true-world probabilities of every
prefix agree and that every prefix having positive common probability
satisfies
\begin{align}
    g_k
    &\in
    \mathscr H_\star,
    &
    \mu_k
    &=
    E_M(g_k).
    \label{eq:exact-model-memory-synchronization}
\end{align}

At $k=0$, both laws assign probability one to the empty continuation,
$g_0=h\in\mathscr H_\star$, and $\mu_0=E_M(h)$. Suppose the claims hold for
$\zeta_k$. If its common probability is zero, every extension of that prefix
has probability zero under both laws. Otherwise, exactness of the pair $(\mathsf M,E_M)$ gives
\begin{align}
    \mathrm{Pr}_M(y\mid\mu_k,a)
    =
    \mathrm{Pr}_M(y\mid E_M(g_k),a)
    =
    \mathrm{Pr}_\star(y\mid g_k,a)
\end{align}
for every $a\in\mathcal A$ and $y\in\mathcal Y$. Together with the induction hypothesis, multiplying by the common prefix
probability and the common policy factor gives
\begin{align}
    \pr_{M,\pi}^{h}(\zeta_{k+1})
    &=
    \pr_{M,\pi}^{h}(\zeta_k)
    \pi(a\mid g_k)
    \mathrm{Pr}_M(y\mid\mu_k,a)
    \nonumber\\
    &=
    \pr_{\star,\pi}^{h}(\zeta_k)
    \pi(a\mid g_k)
    \mathrm{Pr}_\star(y\mid g_k,a)
    \nonumber\\
    &=
    \pr_{\star,\pi}^{h}(\zeta_{k+1}),
\end{align}

If the extended prefix has positive probability, then
$\mathrm{Pr}_\star(y\mid g_k,a)>0$. Hence $g_kay\in\mathscr H_\star$, and the memory-update condition in~\eqref{eq:predictive-consistency-update} gives
\begin{align}
    T_{M,y}^{a}(\mu_k)
    =
    T_{M,y}^{a}(E_M(g_k))
    =
    E_M(g_kay).
\end{align}
This completes the induction and proves
\eqref{eq:exact-model-finite-trajectory-equality}.

The same induction applies when the first action is fixed to
$a\in\mathcal A$. At rollout depth zero, both the model and true-world
processes receive the same supplied action $a$; from depth one onward, both
use the same policy $\pi$. Therefore
\begin{align}
    V_M^\pi(h)
    &=
    V_\star^\pi(h),
    &
    Q_M^\pi(h,a)
    &=
    Q_\star^\pi(h,a).
\end{align}
Taking suprema over $\Pi_{\mathscr H}$ gives the optimal
value and action-value equalities. The memory-state identities follow from
\eqref{eq:model-optimal-value-bridge} and
\eqref{eq:model-optimal-q-bridge}.
\end{proof}

\subsection{Finite-dimensional classical world models}\label{sec:classical_realization}
Here we formally define classical world models.
We state some notation that we will use for classical world models. For real vectors $u,v \in\mathbb R^d$, write $\langle u,v\rangle$ for their inner product, and write $\mathbf 1 \in\mathbb R^d$ for the all-ones vector. The identity matrix is $\mathbb{I}_{d\times d}\in\mathbb{R}^{d\times d}$. A non-negative matrix $D\in\mathbb{R}^{d\times d}_{\geq 0}$ is column-substochastic if $\langle\mathbf 1,D u \rangle
    \leq
    \langle\mathbf 1,u \rangle$ for every non-negative vector $u\in\mathbb{R}^d_{\geq 0}$. We call $D$ column-stochastic if the
inequality is an equality for every $u\in\mathbb{R}^d_{\geq 0}$.
When the column convention is clear, we simply say substochastic and
stochastic. Now we formally define our notion of classical world model.

A classical world model of memory dimension $N$ uses a stochastic physical
memory with $N$ perfectly distinguishable configurations, labeled by
$i\in[N]$. A general state of this memory is represented by a
probability vector $z\in\mathcal{P}([N])$. In POMDP terminology, $z$ is the
distribution over the $N$ internal configurations; physically, its entries are the
preparation probabilities of one $N$-state memory system. We use $N$ for this physical memory dimension, independently of the continuum
of possible state distributions.

The dynamics of this memory take the standard form of a
finite-state POMDP, equivalently the
non-negative hidden-state subclass of a controlled observable-operator
model~\cite{jaeger2000observable,vidyasagar2011complete}. This representation assigns a
non-negative matrix to each action--outcome branch. It is particularly
convenient here because both branch probabilities and posterior distributions over states are
obtained from products of the same matrices, in direct parallel with the
quantum-instrument representation that we use for quantum world models~\cite{vidyasagar2011complete,Hsu2008,fanizza2024quantum}. Related operator
representations have also been useful in recent statistical analyses of
POMDPs~\cite{jin2020sample,liu2022partially}.

\begin{definition}
\label{def:classical-world-model}
An $N$-dimensional \emph{classical world model} is specified by
\begin{align}
\mathsf C_N
&=
\left(
\mathcal A,
\mathcal Y,
[N],
z^{\mathrm{in}},
\mathbf D
\right),
&
\mathbf D
&:=
\left\{
D_y^{(a)}
\right\}_{(a,y)\in\mathcal A\times\mathcal Y},
\end{align}
where $[N]=\{1,\ldots,N\}$ labels the perfectly distinguishable physical
memory configurations,
$z^{\mathrm{in}}\in\mathcal{P}([N])$ is the initial memory state, $\mathcal{A}$ is the set of actions, $\mathcal Y\subseteq\mathcal O\times\mathbb R$ is the set of observations-rewards, and each
$D_y^{(a)}\in\mathbb R_{\geq0}^{N\times N}$ is column-substochastic, and the family satisfies
\begin{align}
\mathbf 1^{\mathsf T}
\left(
\sum_{y\in\mathcal Y}D_y^{(a)}
\right)
=
\mathbf 1^{\mathsf T}
\qquad
\text{for every }a\in\mathcal A.
\end{align}
The entry $(D_y^{(a)})_{ji}$ is the joint probability that the model
generates $y$ and moves from memory configuration $i$ to configuration $j$
when supplied with action $a$.
\end{definition}

The corresponding abstract memory-state space is
$\mathcal Z_C=\mathcal{P}([N])$. The branch matrices induce the outcome law
\begin{align}
\mathrm{Pr}_C(y\mid z,a)
&:=
\left\langle
\mathbf 1,D_y^{(a)}z
\right\rangle
\label{eq:classical-world-model-probability}
\end{align}
and, whenever this probability is nonzero, the conditional memory update
\begin{align}
T_{C,y}^{a}(z)
&:=
\frac{D_y^{(a)}z}
{\left\langle\mathbf 1,D_y^{(a)}z\right\rangle}.
\label{eq:classical-world-model-update}
\end{align}
After fixing $z_{\mathrm{ref}}$, $\mathsf C_N$ induces the abstract recurrent
tuple $(\mathcal A,\mathcal Y,\mathcal Z_C,z^{\mathrm{in}},
\mathrm{Pr}_C,T_C)$ required by Definition~\ref{def:world_model}.
The maps $\mathrm{Pr}_C$ and $T_C$ are derived from $\mathbf D$ and the
fixed zero-probability convention.

For history-indexed queries, write
$z_0:=z^{\mathrm{in}}$ and, for
$h=(a_0,y_1,\ldots,a_{t-1},y_t)$, define
\begin{align}
D_h
&:=
D_{y_t}^{(a_{t-1})}\cdots D_{y_1}^{(a_0)},
&
D_{\emptyset}
&:=
\mathbb I_{N\times N}.
\end{align}
The encoder associated with the model is
\begin{align}
E_C(h)
=
\frac{D_hz_0}
{\left\langle\mathbf 1,D_hz_0\right\rangle}
\label{eq:encoder_classical}
\end{align}
whenever the denominator is nonzero, and
$E_C(h)=z_{\mathrm{ref}}$ otherwise. This encoder is derived from the
model's branch dynamics and is the standard forward-filtering preparation obtained by
conditioning the initial state $z^{\mathrm{in}}$ through the same branch
matrices $D_y^{(a)}$ that govern the subsequent simulation
~\cite{rabiner1989tutorial,KAELBLING199899,vidyasagar2011complete}. Moreover, for strictly positive probability histories, its form is unique in order to fulfill relation~\eqref{eq:predictive-consistency-update} associated to the dynamics of the matrices $\mathbf{D}$. The encoder remains a
separate interface and is not a component of $\mathsf C_N$, but it is
not an additional free parameter in the classical results.
Throughout the classical results below, all history-indexed quantities use this canonical encoder.

\subsection{Finite-dimensional quantum world models}
\label{subsec:quantum-instrument-world-models}

A quantum world model of memory dimension $d$ uses a $d$-level physical
memory with Hilbert space $\mathcal H_Q\simeq\mathbb C^d$.
Write $\mathsf L(\mathcal H_Q)$ for the linear operators on $\mathcal H_Q$ and define
\begin{align}
\mathcal S(\mathcal H_Q)
:=
\left\{
\rho\in\mathsf L(\mathcal H_Q):
\rho\succeq0,\ \operatorname{Tr}\rho=1
\right\}.
\end{align} Its states are
density operators $\rho\in\mathcal S(\mathcal H_Q)$. A set of quantum states
can be perfectly distinguished in a single use only when their supports are
mutually orthogonal, so such a set contains at most $d$ states. For each candidate action, a quantum world model replaces non-negative branch matrices by quantum
instrument elements, forming a quantum instrument for each action~\cite{monras2016,fanizza2024quantum}. The same definitions have also been used to describe quantum POMDPs~\cite{barry2014quantum,lumbreras2026reinforcement}.

\begin{definition}
\label{def:quantum-world-model}
A $d$-dimensional \emph{quantum world model} is specified by
\begin{align}
\mathsf Q_d
&=
\left(
\mathcal A,
\mathcal Y,
\mathcal H_Q,
\rho^{\mathrm{in}},
\boldsymbol{\mathcal E}
\right),
&
\boldsymbol{\mathcal E}
&:=
\left\{
\mathcal E_y^{(a)}
\right\}_{(a,y)\in\mathcal A\times\mathcal Y},
\end{align}
where $\mathcal H_Q\simeq\mathbb C^d$ is the Hilbert space of the memory,
$\rho^{\mathrm{in}}\in\mathcal S(\mathcal H_Q)$ is its initial state, and
each
\begin{align}
\mathcal E_y^{(a)}:
\mathsf L(\mathcal H_Q)
\longrightarrow
\mathsf L(\mathcal H_Q)
\end{align}
is completely positive and trace nonincreasing. For every
$a\in\mathcal A$,
$\sum_{y\in\mathcal Y}\mathcal E_y^{(a)}$ is trace preserving, so
$\{\mathcal E_y^{(a)}\}_{y\in\mathcal Y}$ forms a quantum instrument.
\end{definition}

The corresponding abstract memory-state space is
$\mathcal Z_Q=\mathcal S(\mathcal H_Q)$. The instrument operations induce
the outcome law
\begin{align}
\mathrm{Pr}_Q(y\mid\rho,a)
&:=
\operatorname{Tr}\!\left[
\mathcal E_y^{(a)}(\rho)
\right]
\label{eq:quantum-world-model-probability}
\end{align}
and, whenever this probability is nonzero, the conditional memory update
\begin{align}
T_{Q,y}^{a}(\rho)
&:=
\frac{\mathcal E_y^{(a)}(\rho)}
{\operatorname{Tr}[\mathcal E_y^{(a)}(\rho)]}.
\label{eq:quantum-world-model-update}
\end{align}
Choose  a reference state
$\rho_{\mathrm{ref}}\in\mathcal Z_Q$ and set
$T_{Q,y}^{a}(\rho)=\rho_{\mathrm{ref}}$ on zero-probability branches. Thus,
$\mathrm{Pr}_Q$ and $T_Q$ provide the abstract recurrent representation required by
Definition~\ref{def:world_model}; they are induced by
$\boldsymbol{\mathcal E}$.

For history-indexed queries, write
$\rho_0:=\rho^{\mathrm{in}}$ and, for
$h=(a_0,y_1,\ldots,a_{t-1},y_t)$, define
\begin{align}
\mathcal E_h
&:=
\mathcal E_{y_t}^{(a_{t-1})}
\circ\cdots\circ
\mathcal E_{y_1}^{(a_0)},
&
\mathcal E_{\emptyset}
&:=
\operatorname{id}.
\end{align}
The encoder associated with the instrument dynamics is
\begin{align}
E_Q(h)
=
\frac{\mathcal E_h(\rho_0)}
{\operatorname{Tr}[\mathcal E_h(\rho_0)]}
\label{eq:encoder_quantum}
\end{align}
whenever the denominator is nonzero, and
$E_Q(h)=\rho_{\mathrm{ref}}$ otherwise.  The expression~\eqref{eq:encoder_quantum} identifies the memory state
associated with a history. For every history $h$ satisfying
$\operatorname{Tr}[\mathcal E_h(\rho_0)]>0$, it obeys the update
relation~\eqref{eq:predictive-consistency-update} induced by
$\boldsymbol{\mathcal E}$. On zero-probability histories,
$\rho_{\mathrm{ref}}$ is an arbitrary fixed convention.

The classical and quantum specifications above induce operational recurrent
representations of the form in Definition~\ref{def:world_model}. In each
case, the outcome law and conditional memory update are derived from the
outcome-labelled physical dynamics rather than supplied as independent model
data. Their respective memory dimensions are $N$ and $d$, independently of
the number of statistical states in $\mathcal{P}([N])$ or density operators in
$\mathcal S(\mathcal H_Q)$.

\section{The FRDN true world and finite-dimensional classical gap}
\label{sec:fdrn-classical-gap}

This section constructs the true world used to prove our main results. We begin with the
Fox--Rubin--Dharmadhikari--Nadkarni (FRDN) renewal
process~\cite{dharmadhikari1963sufficient,fox1968,
dharmadhikari1970,vidyasagar2011complete}. The original process generates a
stochastic sequence and has neither actions nor rewards. Its output
probabilities admit a finite-dimensional linear realization: they can be
computed through products of fixed finite-dimensional matrices. Nevertheless,
the same probabilities cannot be generated by any finite-state hidden Markov
model. We use them below to define the action-conditioned dynamics of a
controlled world and assign rewards to its possible outcomes. 

\subsection{The FRDN true world}
\label{subsec:fdrn-controlled-world}

We first specify the renewal law that determines the outcome probabilities of
the true world between two resets. We will equip our agent with the actions Wait, Probe and Maintain. Suppose that, after a reset, the agent
repeatedly chooses Wait. The true world draws an auxiliary random lifetime
$L\in\mathbb N_0$. Conditional on $L=\ell$, it produces Tick (an observation) on the first
$\ell$ Wait actions and Break (another observation) on the next one, after which the process resets.
At each reset, a fresh independent copy of $L$ is drawn.

Let
$p_\ell:=\mathrm{Pr}(L=\ell)$ be the probability that a run contains exactly $\ell$
Ticks. For the FRDN process, fix $\lambda\in(0,1/2]$ and
$\alpha\in\mathbb R$ such that $\alpha/\pi\notin\mathbb Q$, and define
\begin{align}
    p_\ell
    &:=
    \lambda^\ell
    \sin^2\!\left(\frac{\ell\alpha}{2}\right),
    \qquad \ell\geq1,
    \\
    p_0
    &:=
    1-\sum_{\ell=1}^{\infty}p_\ell .
    \label{eq:fdrn-lifetime-law}
\end{align}
This is a probability law because
$\sum_{\ell\geq1}p_\ell\leq\lambda/(1-\lambda)\leq1$.

After $t$ consecutive Ticks, the observed history implies that $L\geq t$.
We therefore define the survival probability
\begin{align}
    \Phi(t)
    &:=
    \mathrm{Pr}(L\geq t)
    =
    \sum_{\ell=t}^{\infty}p_\ell,
    \quad t\geq1,
    &
    \Phi(0)&:=1,
    \label{eq:survival_tail}
\end{align}
and the conditional probability of one additional Tick,
\begin{align}
    S(t)
    &:=
    \mathrm{Pr}(L\geq t+1\mid L\geq t)
    =
    \frac{\Phi(t+1)}{\Phi(t)},
    \label{eq:true_tick_prob}
\end{align}
for $t\geq 0$. The irrationality assumption implies $p_\ell>0$ for every $\ell\geq1$.
Consequently, $\Phi(t)>0$ for every $t$, so the conditional probability
$S(t)$ in~\eqref{eq:true_tick_prob} is well defined.

These quantities are probabilities: $p_\ell$ is a
probability mass, $\Phi(t)$ is a survival probability, and $S(t)$ is a
conditional probability. They determine the true-world kernel below: after
$t$ consecutive Wait--Tick outcomes since the last reset, the next Wait
produces Tick with probability $S(t)$ and Break with probability $1-S(t)$.

For $t\geq1$, summing the geometric series gives
\begin{align}
    \Phi(t)
    &={}
    \lambda^t\bigl[A-B\cos(t\alpha+\varphi)\bigr],
    \label{eq:fdrn-tail-closed-form}
    \\
    A
    &:=
    \frac{1}{2(1-\lambda)}, \,\,
    B
    :={}
    \frac{1}{2|1-\lambda e^{i\alpha}|},
\end{align}
where
$(1-\lambda e^{i\alpha})^{-1}
 =|1-\lambda e^{i\alpha}|^{-1}e^{i\varphi}$.
Since $A>B>0$,
\begin{align}
    S(t)
    &={}
    f(t\alpha),
    \qquad t\geq1,
    \label{eq:fdrn-survival-closed-form}
    \\
    f(x)
    &:={}
    \lambda
    \frac{A-B\cos(x+\alpha+\varphi)}
         {A-B\cos(x+\varphi)} .
    \label{eq:periodic_funtion_tick}
\end{align}
The function $f$ is continuous, $2\pi$-periodic, and nonconstant. Moreover,
$f(x)\in[0,1]$: the irrational orbit
$\{t\alpha\bmod 2\pi\}$ is dense since $\alpha / \pi \notin \mathbb Q$, $f(t\alpha)=S(t)\in[0,1]$, and $f$ is
continuous. For $t\ge1$, $S(t)=f(t\alpha)$. The value of $S(0)$ is defined
separately and plays no role in the asymptotic lower bound.

We now use this probability to define the true controlled world. For an
outcome label $y_o^a$, the superscript records the action $a$, while the
subscript records the observation $o$. The action index is mnemonic and is
not an additional component of the observation--reward pair. For special
reward choices, labels associated with different actions may denote the same
element of the outcome alphabet; this is unambiguous because each branch is
indexed by both its action and its outcome.

\begin{definition}[FRDN true world]
\label{def:fdrn-controlled-world}
In the sense of Definition~\ref{def:true-controlled-world}, the
\emph{FRDN true controlled world} is
\begin{align}
    \mathsf W_\star^{\mathrm{FRDN}}
    &:=
    \bigl(
        \mathcal A_{\mathrm{FRDN}},
        \mathcal Y_{\mathrm{FRDN}},
        \mathscr H,
        \mathrm{Pr}_\star
    \bigr).
    \label{eq:true_fdrn}
\end{align}
Its action set is
\begin{align}
    \mathcal A_{\mathrm{FRDN}}
    &:=
    \{W,M,P\},
\end{align}
where $W$, $M$, and $P$ denote Wait, Maintain, and Probe, respectively.

Let $\mathcal O:=\{T,B\}$ be the observation set, where $T$ denotes Tick
and $B$ denotes Break. Fix parameters
$C_W,C_M,C_P>0$ and $R_W,R_P\in\mathbb R$, and define the possible
observation--reward outcomes by
\begin{align}
    y_T^W&:=(T,R_W),
    &
    y_B^W&:=(B,-C_W),
    &
    y_B^M&:=(B,-C_M),
    \\
    y_T^P&:=(T,R_P),
    &
    y_B^P&:=(B,-C_P).
\end{align}
The outcome set is
\begin{align}
    \mathcal Y_{\mathrm{FRDN}}
    &:=
    \{y_T^W,y_B^W,y_B^M,y_T^P,y_B^P\}.
\end{align}

For $h\in\mathscr H$, let $\ell(h)$ be the length of the terminal sequence
of consecutive Wait--Tick action--outcome pairs in $h$. Equivalently,
\begin{align}
    \ell(\emptyset)
    &:=
    0,
    \\
    \ell(hay)
    &:=
    \begin{cases}
        \ell(h)+1,
        & a=W\ \text{and}\ y=y_T^W,\\
        0,
        & \text{otherwise}.
    \end{cases}
    \label{eq:fdrn-history-clock}
\end{align}
With $S(t)$ denoting the conditional Tick probability defined
in~\eqref{eq:true_tick_prob}, the true outcome law is
\begin{align}
    \mathrm{Pr}_\star(y\mid h,W)
    &:=
    \begin{cases}
        S(\ell(h)),
        & y=y_T^W,\\
        1-S(\ell(h)),
        & y=y_B^W,\\
        0,
        & \text{otherwise},
    \end{cases}
    \\
    \mathrm{Pr}_\star(y\mid h,M)
    &:=
    \begin{cases}
        1,
        & y=y_B^M,\\
        0,
        & \text{otherwise},
    \end{cases}
    \\
    \mathrm{Pr}_\star(y\mid h,P)
    &:=
    \begin{cases}
        S(\ell(h)),
        & y=y_T^P,\\
        1-S(\ell(h)),
        & y=y_B^P,\\
        0,
        & \text{otherwise}.
    \end{cases}
    \label{eq:fdrn-true-kernel}
\end{align}
\end{definition}
Here and throughout the FRDN construction, the history space is instantiated
using the FRDN alphabets:
\begin{align}
\mathscr H
:=
\bigcup_{t\geq0}
\left(
\mathcal A_{\mathrm{FRDN}}
\times
\mathcal Y_{\mathrm{FRDN}}
\right)^t.
\end{align}

The lifetime $L$ is an auxiliary construction used to specify the renewal
law, while $\ell(h)$ is a deterministic function of the observed history;
neither introduces an additional state variable into the true-world tuple.
The Wait--Tick branch increments $\ell(h)$, whereas Wait--Break, Maintain,
and either Probe outcome reset it to zero. Using the one-step reward
definition~\eqref{eq:one-step-expected-rewards}, the expected immediate
rewards are
\begin{align}\label{eq:immediate_reward_expected}
    r_\star(h,W)
    &={}
    (R_W+C_W)S(\ell(h))-C_W,
    \\
    r_\star(h,M)
    &={}
    -C_M,
    \\
    r_\star(h,P)
    &={}
    (R_P+C_P)S(\ell(h))-C_P.
\end{align}

For the controlled FRDN world, the reachable history tree
$\mathscr H_\star$ is generated by the kernels above. Equivalently, a finite
string is reachable if and only if every appended outcome has positive conditional
probability given the preceding history and queried action.

\subsection{Convergence of finite-dimensional classical memories via Perron--Frobenius}

The goal of this subsection is to isolate the finite-dimensional classical
memory constraint imposed by Perron--Frobenius theory on non-negative matrices.
Along the all-Tick trajectory, a finite-dimensional classical world model updates
its distributions over states by repeatedly applying the same non-negative Tick branch and
renormalizing. Perron--Frobenius theory implies that such $N$-dimensional non-negative matrix dynamics cannot
track an irrational rotation forever: after passing to finitely many
arithmetic subsequences, its normalized distributions over states converge. More precisely, for some
period $p$, each residue class $r\in\lbrace 0,\ldots,p-1 \rbrace$ collects the times
$t=r+kp$ for $k\in\mathbb{N}_0$, and the classical memory along each such subsequence has a limiting
state. This will imply that the classical Tick predictions become asymptotically
constant on each residue class.

Let $\mathsf C_N$ be an arbitrary $N$-dimensional classical world model in
the sense of Definition~\ref{def:classical-world-model}, with action set
$\mathcal A_{\mathrm{FRDN}}$ and outcome alphabet
$\mathcal Y_{\mathrm{FRDN}}$. Let
\begin{align}
    D_T
    &:=
    D_{y_T^W}^{(W)}
\end{align}
be its Wait--Tick branch matrix, and let
$z_0\in\mathcal{P}([N])$ be its initial distribution over states.

Define the all-Tick action--outcome trajectory by
\begin{align}
\mathscr F_{\mathrm{tick}}
&:=
h_1h_2\cdots,
&
h_t
&:=
(W,y_T^W),
\qquad t\geq1,
\nonumber\\
\ph_0
&:=
\emptyset,
&
\ph_t
&:=
h_1h_2\cdots h_t,
\qquad t\geq1.
\label{eq:fdrn-tick-history}
\end{align}
Thus,
\begin{align}
\ph_{t+1}
&=
\ph_t W y_T^W,
&
\ell(\ph_t)
&=
t.
\end{align}
Moreover,
\begin{align}
\mathrm{Pr}_\star(y_T^W\mid\ph_t,W)
&=
S(t)
=
\frac{\Phi(t+1)}{\Phi(t)}
>
0,
\end{align}
so every prefix $\ph_t$ is reachable.

Each $\ph_t$ can serve as the common root of either a true-world
continuation or a world-model rollout, with
\begin{align}
H_0^\star
=
H_0^M
=
\ph_t.
\end{align}
A general history $h$ satisfying $\ell(h)=t$ need not equal $\ph_t$; only
its terminal Wait--Tick suffix has length $t$.

Along these prefixes, define the model survival probability and conditional
Tick prediction by
\begin{align}
\Phi_C(t)
&:=
\left\langle
\mathbf 1,D_T^t z_0
\right\rangle,
\label{eq:fdrn-classical-tail}
\\
S_C(t)
&:=
\mathrm{Pr}_C
\!\left(
y_T^W\mid E_C(\ph_t),W
\right).
\label{eq:fdrn-classical-tick-probability}
\end{align}
Here, $\Phi_C(t)$ is the probability that the model assigns to the
length-$t$ all-Tick prefix, whereas $S_C(t)$ is its conditional probability
of one further Tick. Whenever $\Phi_C(t)>0$,
\begin{align}
E_C(\ph_t)
&=
\frac{D_T^t z_0}{\Phi_C(t)},
&
S_C(t)
&=
\frac{\Phi_C(t+1)}{\Phi_C(t)}.
\label{eq:fdrn-classical-ratio}
\end{align}
If $\Phi_C(t_0)=0$ for some $t_0$, then $\Phi_C(t)=0$ for every
$t\geq t_0$. By the fixed-reference convention following~\eqref{eq:encoder_classical}, the encoded memory state and its Tick
prediction are then eventually constant, so the claims below are immediate.
We therefore treat the case $\Phi_C(t)>0$ for every $t$.

Our main technical tool is Perron--Frobenius theory, which characterizes the
asymptotic behavior of powers of $D_T$ and hence of the conditional
probability $S_C(t)$. We use the following standard facts
that can be found in~\cite{seneta2006nonnegative,berman1994nonnegative}.

\paragraph{Perron--Frobenius theory for non-negative matrices.}
We use the column-vector convention of
Definition~\ref{def:classical-world-model}. $D_{ji}>0$ means that one
step can carry mass from state $i$ to state $j$. Equivalently, the directed
graph of a non-negative matrix $D\in\mathbb R_{\geq0}^{N\times N}$ has an
edge $i\to j$ whenever $D_{ji}>0$; then $(D^t)_{ji}>0$ precisely when there
is a positive-weight path of length $t$ from $i$ to $j$. A set of states is
\emph{strongly connected} if each state can reach every other, and $D$ is
\emph{irreducible} when its graph is strongly connected.

The strongly connected components can be ordered so that a simultaneous
permutation of rows and columns puts $D$ in \emph{Frobenius normal form},
\begin{align}
    \Pi D\Pi^{\mathsf T}
    =
    \begin{pmatrix}
        B_1 & 0 & \cdots & 0\\
        *   & B_2 & \ddots & \vdots\\
        \vdots & \ddots & \ddots & 0\\
        * & \cdots & * & B_m
    \end{pmatrix}.
    \label{eq:fdrn-frobenius-normal-form}
\end{align}
Each diagonal block $B_j$ is irreducible or a $1\times1$ zero block, while
the off-diagonal blocks describe paths between distinct components. We call
$B_j$ \emph{reachable from} a non-negative vector $z$ if some mass initially
in the support of $z$ can enter that component; equivalently, there exist $\tau\geq0$ and a state $i$ in $B_j$ such that
$(D^\tau z)_i>0$. Blocks that are
not reachable from $z$ never contribute to $D^t z$.

For an irreducible block $B$, let
$\rho_B:=\max\{|\mu|:\mu\in\operatorname{spec}(B)\}$ be its spectral radius.
Its period is the greatest common divisor of the lengths of all closed paths
from any fixed state back to itself,
\begin{align}
    h_B:=\gcd\{t\geq1:(B^t)_{ii}>0\};
\end{align}
the value is independent of $i$. The Perron--Frobenius theorem gives positive
left and right eigenvectors at $\rho_B$, and states that the eigenvalues on
the spectral circle $|\mu|=\rho_B$ are exactly
\begin{align}
    \rho_B e^{2\pi i k/h_B},
    \qquad k=0,1,\ldots,h_B-1.
    \label{eq:fdrn-pf-peripheral-spectrum}
\end{align}
These are the \emph{peripheral eigenvalues}. Since
$h_B\leq\dim B\leq N$, the integer
\begin{align}\label{eq:lcm}
    L_N:=\operatorname{lcm}(1,2,\ldots,N)
\end{align}
is divisible by every possible block period. Hence every peripheral
eigenvalue $\mu$ of $B$ obeys
\begin{align}
    \mu^{L_N}=\rho_B^{L_N}.
    \label{eq:fdrn-period-killing}
\end{align}
Passing to a fixed residue class modulo $L_N$ therefore removes all
Perron--Frobenius phases. 

We now prove the main technical result used to establish the value-function gap.

\begin{lemma}
\label{lem:fdrn-classical-residue-convergence}
Let $D_T\in\mathbb R_{\geq0}^{N\times N}$ be column-substochastic and $ L_N:=\operatorname{lcm}(1,2,\ldots,N)$. Let $z_0\in\mathcal{P}([N])$, and define
\begin{align}
    \Phi_C(t):=\langle\mathbf 1,D_T^t z_0\rangle.
\end{align}
Assume that $\Phi_C(t)>0$ for all sufficiently large $t$. Then, for every
$q\in\mathbb R_{\geq0}^{N}$ with $q\leq\mathbf 1$ and every
$r\in\{0,1,\ldots,L_N-1\}$, there exists $c_{q,r}\in[0,1]$ such that
\begin{align}
    \lim_{n\to\infty}
    \frac{\langle q,D_T^{nL_N+r}z_0\rangle}
         {\langle\mathbf 1,D_T^{nL_N+r}z_0\rangle}
    =c_{q,r}.
    \label{eq:fdrn-readout-residue-limit}
\end{align}
In particular, taking $q=D_T^{\mathsf T}\mathbf 1$ gives constants
$c_r\in[0,1]$ such that
\begin{align}
    S_C(nL_N+r)
    =
    \frac{\Phi_C(nL_N+r+1)}{\Phi_C(nL_N+r)}
    \longrightarrow c_r.
    \label{eq:fdrn-classical-ratio-residue-limit}
\end{align}
\end{lemma}

\begin{proof}
Fix $r$ and abbreviate $D:=D_T$ and $L:=L_N$. Column substochasticity gives
\begin{align}
    \Phi_C(t+1)
    =\langle D^{\mathsf T}\mathbf 1,D^t z_0\rangle
    \leq\langle\mathbf 1,D^t z_0\rangle
    =\Phi_C(t).
\end{align}
Thus, if $\Phi_C(t)$ is positive for all sufficiently large $t$, it is in fact
positive for every $t$. We may therefore define the normalized column vector
\begin{align}
    z_r
    :=
    \frac{D^r z_0}{\langle\mathbf 1,D^r z_0\rangle}
    \in\mathcal{P}([N]).
    \label{eq:fdrn-residue-initial-belief}
\end{align}
Since $D^{nL+r}=D^{nL}D^r$,
\begin{align}
    \frac{\langle q,D^{nL+r}z_0\rangle}
         {\langle\mathbf 1,D^{nL+r}z_0\rangle}
    =
    \frac{\langle q,D^{nL}z_r\rangle}
         {\langle\mathbf 1,D^{nL}z_r\rangle}.
    \label{eq:fdrn-residue-reduction}
\end{align}
It is therefore enough to study the subsequence $nL$ from an arbitrary
initial distribution $z$; below we write $z=z_r$.

Put $D$ in Frobenius normal form~\eqref{eq:fdrn-frobenius-normal-form} and delete all blocks that are not reachable from $z$.
They never receive mass from $z$, so this does not change $D^t z$ or either
scalar in~\eqref{eq:fdrn-residue-reduction}. We henceforth restrict $D$ to the reachable blocks. Let
\begin{align}
    \rho:=\max_j\rho_{B_j}
\end{align}
be the largest spectral radius among its diagonal blocks. The
hypothesis $\Phi_C (t) > 0$ implies $\rho>0$.

The main contribution of $D^t$ in the inner product $\langle\mathbf 1,D^t z\rangle$ has exponential rate $\rho$. Choose a reachable block $B$ with spectral radius $\rho_B=\rho$. By
reachability, there exist $p\geq0$ and a state $a$ in $B$ such that
$(D^p z)_a>0$. Let $e_a$ denote the basis vector of that state
within the block, and let $u>0$ be a left Perron vector of $B$, written as
$B^{\mathsf T}u=\rho u$, and let $\mathbf 1_B$ denote the all-ones vector on
that block. For some $\eta>0$, $\mathbf 1_B\geq\eta u$ componentwise.
Keeping only paths that enter $B$ at $a$ and subsequently remain in $B$ gives,
for $t\geq p$,
\begin{align}
  a(t)
    :=  \langle\mathbf 1,D^t z\rangle
    &\geq
    (D^p z)_a\,\langle\mathbf 1_B,B^{t-p}e_a\rangle
    \nonumber\\
    &\geq
   (D^p z)_a\eta\,\langle u,B^{t-p}e_a\rangle \\
   & =
    (D^p z)_a\eta \langle u,e_a\rangle \rho^{t-p}.
    \label{eq:fdrn-perron-lower-bound}
\end{align}
Thus,
\begin{align}\label{eq:fdrn-perron-lower-boundv2}
    a(t) \geq C_p \rho^t , \quad C_p :=   (D^p z)_a\eta \langle u,e_a\rangle \rho^{-p} > 0.
\end{align}

We also use the standard consequence of Jordan normal form that, for any
finite matrix $D$ and vectors $u,v$, the scalar sequence
$\langle u,D^t v\rangle$ is a finite sum of polynomial--exponential terms
$p_\mu(t)\mu^t$, where $\mu$ ranges over eigenvalues of $D$
~\cite[Sec.~3.1]{Horn_Johnson_1985}. In particular, for all sufficiently large $t$,
\begin{align}
    a(t)
    =
    \langle\mathbf 1,D^t z\rangle
    =
    \sum_{\mu\in\operatorname{spec}(D)\setminus\{0\}}
    p_\mu(t)\mu^t,
    \label{eq:fdrn-jordan-scalar-expansion}
\end{align}
where each $p_\mu$ is a polynomial of $t$. Every eigenvalue of $D$ lies in a
diagonal Frobenius block, and hence has modulus at most $\rho$. Moreover, if
$|\mu|=\rho$, then $\mu$ belongs to a block of spectral radius $\rho$ and is
therefore peripheral for that block. Perron-Frobenius~\eqref{eq:fdrn-period-killing}
then gives $\mu^{nL}=\rho^{nL}$. The lower bound~\eqref{eq:fdrn-perron-lower-bound} ensures that the terms with
$|\mu|=\rho$ do not all cancel after passing to $t=nL$. Indeed, after
substituting $t=nL$ in~\eqref{eq:fdrn-jordan-scalar-expansion}, all terms with
$|\mu|<\rho$ are exponentially smaller, while every term with $|\mu|=\rho$
satisfies $\mu^{nL}=\rho^{nL}$ by~\eqref{eq:fdrn-period-killing}. Hence the
terms with $|\mu|=\rho$ combine into $\rho^{nL}p(n)$ for a real polynomial
$p$. If $p$ were the zero polynomial, then
$a(nL)=o(\rho^{nL})$, contradicting~\eqref{eq:fdrn-perron-lower-boundv2}.
Therefore $p$ is not zero. Let $m$ be its degree and $A$ its leading
coefficient. Then
\begin{align}
    a(nL)
    =
    \rho^{nL}n^m\bigl(A+o(1)\bigr).
    \label{eq:fdrn-denominator-asymptotic}
\end{align}
Since $a(nL)>0$ for all sufficiently large $n$, necessarily $A>0$.

For the numerator in~\eqref{eq:fdrn-residue-reduction}, set
\begin{align}
    b_q(t):=\langle q,D^t z\rangle.
\end{align}
It has a Jordan expansion of the same form, and the componentwise inequality
$0\leq q\leq\mathbf 1$ gives
\begin{align}
    0\leq b_q(t)\leq a(t).
    \label{eq:fdrn-readout-domination}
\end{align}
After substituting $t=nL$, the terms with $|\mu|=\rho$ in the numerator
combine into $\rho^{nL}\widetilde p(n)$ for another real polynomial
$\widetilde p$, while all terms with $|\mu|<\rho$ are exponentially smaller.
The inequality~\eqref{eq:fdrn-readout-domination} implies that
$\widetilde p$ has degree at most $m$: if it had larger degree, then
$b_q(nL)$ would eventually either be negative or larger than $a(nL)$.
Therefore
\begin{align}
    b_q(nL)
    =
    \rho^{nL}n^m\bigl(B_q+o(1)\bigr),
    \label{eq:fdrn-numerator-asymptotic}
\end{align}
where $B_q=0$ when the numerator has strictly smaller order. Dividing
\eqref{eq:fdrn-readout-domination} by $\rho^{nL}n^m$ and taking the limit
gives $0\leq B_q\leq A$. Therefore
\begin{align}
    \frac{\langle q,D^{nL}z\rangle}
         {\langle\mathbf 1,D^{nL}z\rangle}
    \longrightarrow
    \frac{B_q}{A}\in[0,1].
\end{align}
Together with~\eqref{eq:fdrn-residue-reduction}, this proves
\eqref{eq:fdrn-readout-residue-limit}.

Finally, column substochasticity implies
$0\leq D^{\mathsf T}\mathbf 1\leq\mathbf 1$, and
\begin{align}
    \langle D^{\mathsf T}\mathbf 1,D^t z_0\rangle
    =
    \langle\mathbf 1,D^{t+1}z_0\rangle
    =\Phi_C(t+1).
\end{align}
The choice $q=D^{\mathsf T}\mathbf 1$ therefore yields
\eqref{eq:fdrn-classical-ratio-residue-limit}.
\end{proof}

\begin{corollary}
\label{cor:fdrn-classical-belief-residue-convergence}
Under the hypotheses of
Lemma~\ref{lem:fdrn-classical-residue-convergence}, for every residue $r$
there exists $z_{\infty,r}\in\mathcal{P}([N])$ such that
\begin{align}
    E_C(\ph_{nL_N+r})
    =
    \frac{D_T^{nL_N+r}z_0}
         {\langle\mathbf 1,D_T^{nL_N+r}z_0\rangle}
    \longrightarrow z_{\infty,r}.
    \label{eq:fdrn-classical-belief-residue-limit}
\end{align}
Consequently, every continuous function of the classical distribution over states converges to
a constant on each residue class.
\end{corollary}

\begin{proof}
For the $j$th standard basis vector $e_j$,
\begin{align}
    \left[E_C(\ph_{nL_N+r})\right]_j
    =
    \frac{\langle e_j,D_T^{nL_N+r}z_0\rangle}
         {\langle\mathbf 1,D_T^{nL_N+r}z_0\rangle}.
\end{align}
Since $0\leq e_j\leq\mathbf 1$, the lemma gives convergence of every
coordinate. The limiting coordinates are non-negative and sum to one, so
they define $z_{\infty,r}\in\mathcal{P}([N])$.
\end{proof}

\subsection{Gap for the conditional probabilities}\label{subsec:gap_conditional_probability}

We now turn the convergence of finite classical memories established in Lemma~\ref{lem:fdrn-classical-residue-convergence} and Corollary~\ref{cor:fdrn-classical-belief-residue-convergence} into an operational separation for prediction, action-values, and values along the Tick histories.

We start with a standard  tool that converts the problem of approximating the sequence $f(t\alpha)$ along arithmetic subsequences into the simpler problem of approximating the function $f$ by constants in phase average. Recall that $f$ is given in~\eqref{eq:periodic_funtion_tick} and $S(t) = f(t\alpha)$ for $t\geq 1$.

We use Weyl equidistribution to compare the sequence $f(t\alpha)$ with the
phase average of $f$. Since $L_N\alpha/(2\pi)\notin\mathbb Q$, Weyl
equidistribution~\cite[Chap.~1]{kuipers2012uniform} gives, for every residue
$r\in\{0,\ldots,L_N-1\}$ and every Riemann-integrable
$2\pi$-periodic function $g$,
\begin{align}
\lim_{K\to\infty}
\frac{1}{K}
\sum_{n=0}^{K-1}
g\!\left((nL_N+r)\alpha\right)
&=
\frac{1}{2\pi}
\int_0^{2\pi}
g(x)\,dx.
\label{eq:fdrn-weyl-residue-average}
\end{align}
In particular,~\eqref{eq:fdrn-weyl-residue-average} applies both to the
continuous functions used in the value bounds and to the indicator
functions used in~\eqref{eq:fdrn-half-tick-residue}, since the latter have
only finitely many discontinuities.

We shall repeatedly use the following stability consequence of Weyl
equidistribution.

\begin{lemma}
\label{lem:cesaro-weyl-stability}
Fix $r\in\{0,\ldots,L_N-1\}$. Let
$F:\mathbb R\to\mathbb R$ be continuous and $2\pi$-periodic, and let
$c_n\to c$ in $\mathbb R$. Then
\begin{align}
    \lim_{K\to\infty}
    \frac1K
    \sum_{n=0}^{K-1}
    \left|
        F\bigl((nL_N+r)\alpha\bigr)-c_n
    \right| \\
    =
    \frac{1}{2\pi}
    \int_0^{2\pi}
    |F(x)-c|\,dx.
    \label{eq:cesaro-weyl-stability}
\end{align}
\end{lemma}

\begin{proof}
The reverse triangle inequality gives
\begin{align}
    &\left|
        \frac1K
        \sum_{n=0}^{K-1}
        \left|
            F\bigl((nL_N+r)\alpha\bigr)-c_n
        \right|
        \right.
        \nonumber
    \left.
        -
        \frac1K
        \sum_{n=0}^{K-1}
        \left|
            F\bigl((nL_N+r)\alpha\bigr)-c
        \right|
    \right|
    \nonumber\\
    &\qquad\leq
    \frac1K
    \sum_{n=0}^{K-1}|c_n-c|.
    \label{eq:cesaro-weyl-perturbation}
\end{align}
Since $c_n\to c$, convergence implies
\begin{align}
    \lim_{K\to\infty}
    \frac1K
    \sum_{n=0}^{K-1}|c_n-c|
    =
    0.
    \label{eq:cesaro-convergent-perturbation}
\end{align}
On the other hand, Weyl equidistribution~\eqref{eq:fdrn-weyl-residue-average} applied to the continuous
$2\pi$-periodic function $x\mapsto|F(x)-c|$ gives
\begin{align}
    \lim_{K\to\infty}
    \frac1K
    \sum_{n=0}^{K-1}
    \left|
        F\bigl((nL_N+r)\alpha\bigr)-c
    \right|
    =
    \frac{1}{2\pi}
    \int_0^{2\pi}|F(x)-c|\,dx.
\end{align}
Combining the above proves~\eqref{eq:cesaro-weyl-stability}.
\end{proof}

The following constant will be the one appearing in our lower bounds,
\begin{align}
    \kappa_{\mathrm{FRDN}}
    :=
    \min_{c\in[0,1]}
    \frac{1}{2\pi}
    \int_0^{2\pi}|f(x)-c|\,dx.
    \label{eq:kappa-fdrn}
\end{align}
Because $f$ is continuous and nonconstant,
$\kappa_{\mathrm{FRDN}}>0$. 
The following proposition combines this
constant approximation gap with the residue-class convergence established
above.

\begin{proposition}[Finite-dimensional classical prediction gap]
\label{prop:fdrn_tick_probability_uniform_average_gap}
For every finite-dimensional classical world model over
$(\mathcal A_{\mathrm{FRDN}},\mathcal Y_{\mathrm{FRDN}})$, in the sense of
Definition~\ref{def:classical-world-model},
\begin{align}
    \liminf_{T\to\infty}
    \frac1T
    \sum_{t=1}^{T}
    |S(t)-S_C(t)|
    \geq
    \kappa_{\mathrm{FRDN}}.
    \label{eq:fdrn-probability-gap}
\end{align}
\end{proposition}

\begin{proof}
First suppose that $\Phi_C(t_0)=0$ for some $t_0$. Since $D_T$ is
column-substochastic, $\Phi_C(t+1)\le \Phi_C(t)$,
so $\Phi_C(t)=0$ for every $t\ge t_0$. By the fixed-reference convention following~\eqref{eq:encoder_classical}, we then have
$E_C(\ph_t)=z_{\rm ref}$ for all
$t\geq t_0$. Therefore
\begin{align}
S_C(t)
=
\mathrm{Pr}_C(y_T^W\mid z_{\rm ref},W)
=:c_{\rm ref}\in[0,1]
\end{align}
eventually. Since
$\alpha/(2\pi)\notin\mathbb Q$, Weyl equidistribution~\eqref{eq:fdrn-weyl-residue-average} gives
\begin{align}
\lim_{T\to\infty}
\frac1T\sum_{t=1}^T |S(t)-S_C(t)|
&=
\frac{1}{2\pi}\int_0^{2\pi}|f(x)-c_{\rm ref}|\,dx \nonumber \\
&\ge
\kappa_{\mathrm{FRDN}}.
\end{align}
Thus, the claim holds in this case. Hence, from now on, assume
$\Phi_C(t)>0$ for all $t$. Fix a residue class $r\in\{0,\ldots,L_N-1\}$. By
Lemma~\ref{lem:fdrn-classical-residue-convergence},
\begin{align}
    S_C(nL_N+r)
    \longrightarrow
    c_r.
\end{align}
Applying Lemma~\ref{lem:cesaro-weyl-stability} with
\begin{align}
    F(x)&:=f(x),
    &
    c_n&:=S_C(nL_N+r),
    &
    c&:=c_r,
\end{align}
gives
\begin{align}
    \lim_{K\to\infty}
    \frac1K
    \sum_{n=0}^{K-1}
    |S(nL_N+r)-S_C(nL_N+r)| \nonumber\\
    =
    \frac{1}{2\pi}
    \int_0^{2\pi}|f(x)-c_r|\,dx
    \geq
    \kappa_{\mathrm{FRDN}}.
    \label{eq:fdrn-prediction-error-residue-limit}
\end{align}

For $T=KL_N$, decomposition into residue classes gives
\begin{align}
    &\frac{1}{KL_N}
    \sum_{t=0}^{KL_N-1}|S(t)-S_C(t)|
    \nonumber\\
    &\qquad=
    \frac{1}{L_N}
    \sum_{r=0}^{L_N-1}
    \frac1K
    \sum_{n=0}^{K-1}
    |S(nL_N+r)-S_C(nL_N+r)|.
\end{align}
Taking $K\to\infty$ and using
\eqref{eq:fdrn-prediction-error-residue-limit} gives
\begin{align}
\liminf_{K\to\infty}
\frac{1}{KL_N}
\sum_{t=0}^{KL_N-1}|S(t)-S_C(t)|
\geq
\kappa_{\mathrm{FRDN}}.
\end{align}
Since the summands lie in $[0,1]$, shifting the indices from
$0,\ldots,KL_N-1$ to $1,\ldots,KL_N$ is asymptotically irrelevant.
For arbitrary $T$, let $K=\lfloor T/L_N\rfloor$. Then since all terms are nonnegative,
\begin{align}
\frac1T\sum_{t=1}^{T}|S(t)-S_C(t)|
&\geq
\frac{KL_N}{T}
\left(
\frac{1}{KL_N}
\sum_{t=1}^{KL_N}|S(t)-S_C(t)|
\right).
\end{align}
Since $KL_N/T\to1$, taking the lower limit proves
\eqref{eq:fdrn-probability-gap}.
\end{proof}

\subsection{Model-selected actions and decision loss}
\label{subsec:fdrn-model-selected-decisions}

We now study the action selected by an agent that acts greedily with respect
to a finite-dimensional classical world model on the all-Tick query
histories. To make the subsection self-contained, we repeat the decision
quantities introduced in the main text.

For a query history $h$ and world model $\mathsf M$, fix a
\emph{model-greedy action}
\begin{align}
\widehat a_M(h)
\in
\operatorname*{arg\,max}_{a\in\mathcal A_{\mathrm{FRDN}}}
Q_M^{\mathrm{opt}}(h,a).
\label{eq:fdrn-model-greedy-action}
\end{align}
Choose likewise a true-world optimal action
\begin{align}
a_\star(h)
\in
\operatorname*{arg\,max}_{a\in\mathcal A_{\mathrm{FRDN}}}
Q_\star^{\mathrm{opt}}(h,a).
\end{align}
The corresponding model and true-world decision margins are
\begin{align}
g_M(h)
&:=
Q_M^{\mathrm{opt}}
\!\left(
h,\widehat a_M(h)
\right)
-
\max_{a\neq\widehat a_M(h)}
Q_M^{\mathrm{opt}}(h,a),
\nonumber\\
g_\star(h)
&:=
Q_\star^{\mathrm{opt}}
\!\left(
h,a_\star(h)
\right)
-
\max_{a\neq a_\star(h)}
Q_\star^{\mathrm{opt}}(h,a).
\label{eq:fdrn-decision-margins}
\end{align}
Each margin is the difference between the largest and second-largest
action-values, and therefore vanishes when the two largest values are tied.

The true-world loss incurred by deploying the model-greedy action is
\begin{align}
\ell_M(h)
&:=
V_\star^{\mathrm{opt}}(h)
-
Q_\star^{\mathrm{opt}}
\!\left(
h,\widehat a_M(h)
\right).
\label{eq:fdrn-decision-loss}
\end{align}
This is the return sacrificed by the current model-selected action when all
subsequent actions are chosen optimally in the true world.

Along a reachable action--outcome trajectory
$\mathscr F=h_1h_2\cdots$, with prefixes
$\ph_t=h_1\cdots h_t$, define
\begin{align}
\overline{\ell}_M(\mathscr F)
&:=
\liminf_{T\to\infty}
\frac{1}{T}
\sum_{t=1}^{T}
\ell_M(\ph_t).
\label{eq:fdrn-mean-decision-loss}
\end{align}

\begin{definition}[Loss of decision resolution]
\label{def:fdrn-loss-decision-resolution}
A world model $\mathsf M$ loses decision resolution along a reachable
trajectory $\mathscr F=h_1h_2\cdots$ if there exists
$\Delta>0$ such that, for every $\delta>0$ and every
$T\in\mathbb N$, there is a $t\geq T$ satisfying
\begin{align}
g_\star(\ph_t)
&\geq
\Delta,
&
g_M(\ph_t)
&<
\delta.
\label{eq:fdrn-loss-decision-resolution}
\end{align}
\end{definition}

For a finite-dimensional classical world model along
$\mathscr F_{\mathrm{tick}}$, write
\begin{align}
\widehat a_C(t)
&:=
\widehat a_C(\ph_t),
&
\ell_C(t)
&:=
\ell_C(\ph_t),
\nonumber\\
g_C(t)
&:=
g_C(\ph_t),
&
g_\star(t)
&:=
g_\star(\ph_t).
\label{eq:fdrn-decision-shorthand}
\end{align}
At each prefix $\ph_t$, all three candidate actions are evaluated,
independently of the Wait action that extends the all-Tick trajectory.

Fix $0<\eta<1-\lambda$ and specialize the reward parameters to
\begin{align}
R_W
&:=
-(1+\eta),
&
C_W
&:=
1+\eta,
\nonumber\\
C_M
&:=
\eta,
&
R_P
&:=
1-\lambda-\eta,
\nonumber\\
C_P
&:=
\lambda+\eta.
\label{eq:fdrn-decision-rewards}
\end{align}
Thus, Wait gives reward $-(1+\eta)$ after either outcome, Maintain gives the
deterministic reward $-\eta$, and Probe gives reward
$1-\lambda-\eta$ after Tick and $-(\lambda+\eta)$ after Break. Maintain and
both Probe outcomes reset the clock to age zero. Consequently, Maintain and
Probe have the same discounted continuation term; their comparison depends
only on their expected immediate rewards. The next lemma computes the
resulting true-world optimal actions and decision margins.

\begin{lemma}
\label{lem:fdrn-true-decisions-all-gamma}
For the rewards in~\eqref{eq:fdrn-decision-rewards}, every
$\gamma\in[0,1)$ and $t\geq0$ satisfy
\begin{align}
&Q_\star^{\mathrm{opt}}(\ph_t,P)
-Q_\star^{\mathrm{opt}}(\ph_t,M)
=
S(t)-\lambda,
\nonumber\\
&Q_\star^{\mathrm{opt}}(\ph_t,W)
-Q_\star^{\mathrm{opt}}(\ph_t,M)
<
-\lambda.
\label{eq:fdrn-decision-wait-dominated}
\end{align}
Consequently, the true optimal-action set is
\begin{align}
\mathcal A_\star(t)
&:=
\operatorname*{arg\,max}_{a\in\mathcal A_{\mathrm{FRDN}}}
Q_\star^{\mathrm{opt}}(\ph_t,a)
\nonumber\\
&=
\begin{cases}
\{P\}, & S(t)>\lambda,\\
\{P,M\}, & S(t)=\lambda,\\
\{M\}, & S(t)<\lambda,
\end{cases}
\label{eq:fdrn-decision-true-action}
\end{align}
and
\begin{align}
g_\star(t)
=
|S(t)-\lambda|.
\label{eq:fdrn-true-decision-margin}
\end{align}
\end{lemma}

\begin{proof}
For $t\geq0$, define
\begin{align}
H_t
&:=
\bigl(S(t)-\lambda\bigr)_+ .
\label{eq:fdrn-Ht-definition}
\end{align}
We construct a candidate Bellman fixed point whose value depends on a history
only through its clock age. At age zero, we define
\begin{align}
u_0
&:=
\frac{-\eta+H_0}{1-\gamma},
\label{eq:fdrn-u0-definition}
\end{align}
and for $t\geq1$, define
\begin{align}
u_t
&:=
-\eta+\gamma u_0+H_t.
\label{eq:fdrn-ut-definition}
\end{align}
By the defining equation for $u_0$, the same formula also holds at $t=0$.

Given the above quantities, we now define our candidate function that will be the fixed-point equation of the Bellman value function equation~\eqref{eq:optimal-bellman-operator} with the optimal Bellman operator $\mathcal{T}_\star$ defined in~\eqref{eq:true-optimal-bellman-operator}. For a history $h$, let
\begin{align}
U(h)
&:=
u_{\ell(h)}.
\label{eq:fdrn-U-definition}
\end{align}
Since $0\leq H_t\leq1-\lambda$, the function $U$ is bounded. To verify that $U$ is the optimal value, it is useful to separate the
action-specific terms entering the Bellman operator. For a bounded
 function $F$ and an action $a$, define
\begin{align}
\bigl(\mathcal B_{\star,a}F\bigr)(h)
&:=
\sum_{y\in\mathcal Y_{\mathrm{FRDN}}}
\mathrm{Pr}_\star(y\mid h,a)
\left[
r_y+\gamma F(hay)
\right].
\label{eq:fdrn-action-specific-bellman-backup}
\end{align}
Then
\begin{align}
(\mathcal T_\star F)(h)
=
\max_{a\in\mathcal A_{\mathrm{FRDN}}}
\bigl(\mathcal B_{\star,a}F\bigr)(h).
\end{align}
Later we will identify $\mathcal B_{\star,a}U$ with
$Q_\star^{\mathrm{opt}}$, only after we prove that
$U=V_\star^{\mathrm{opt}}$.

Fix a history $h$ with $\ell(h)=t$. Maintain gives the deterministic reward
$-\eta$ and resets the clock to age zero. Its continuation value under $U$
is therefore $u_0$, so
\begin{align}
\bigl(\mathcal B_{\star,M}U\bigr)(h)
&=
-\eta+\gamma u_0.
\label{eq:fdrn-maintain-value-in-proof}
\end{align}

For Probe, Tick occurs with probability $S(t)$ and gives reward
$1-\lambda-\eta$, whereas Break occurs with probability $1-S(t)$ and gives
reward $-(\lambda+\eta)$. Both outcomes reset the clock to age zero. Hence
\begin{align}
\bigl(\mathcal B_{\star,P}U\bigr)(h)
&=
S(t)
\left[
1-\lambda-\eta+\gamma u_0
\right]
\nonumber\\
&\quad+
\bigl(1-S(t)\bigr)
\left[
-\lambda-\eta+\gamma u_0
\right]
\nonumber\\
&=
-\eta+\gamma u_0+S(t)-\lambda.
\label{eq:fdrn-probe-value-in-proof}
\end{align}
Subtracting the Maintain~\eqref{eq:fdrn-maintain-value-in-proof} gives
\begin{align}
\bigl(\mathcal B_{\star,P}U\bigr)(h)
-
\bigl(\mathcal B_{\star,M}U\bigr)(h)
&=
S(t)-\lambda.
\label{eq:fdrn-probe-minus-maintain}
\end{align}
Therefore, the larger of the Maintain and Probe is
\begin{align}
\max\left\{
\bigl(\mathcal B_{\star,M}U\bigr)(h),
\bigl(\mathcal B_{\star,P}U\bigr)(h)
\right\}
&=
-\eta+\gamma u_0
+
\bigl(S(t)-\lambda\bigr)_+
\nonumber\\
&=
-\eta+\gamma u_0+H_t
\nonumber\\
&=
u_t.
\label{eq:fdrn-max-maintain-probe}
\end{align}

It remains to show that Wait never exceeds this value. Wait gives reward
$-(1+\eta)$ after either outcome. A Tick, occurring with probability $S(t)$,
increases the clock age to $t+1$, whereas a Break resets it to zero. Thus,
\begin{align}
\bigl(\mathcal B_{\star,W}U\bigr)(h)
&=
-(1+\eta)
+
\gamma
\left[
S(t)u_{t+1}
+
\bigl(1-S(t)\bigr)u_0
\right].
\label{eq:fdrn-wait-value-in-proof}
\end{align}
Subtracting the Maintain and using
$u_{t+1}-u_0=H_{t+1}-H_0$ gives
\begin{align}
&\bigl(\mathcal B_{\star,W}U\bigr)(h)
-
\bigl(\mathcal B_{\star,M}U\bigr)(h)
=
-1+\gamma S(t)(u_{t+1}-u_0)
\nonumber\\
&\qquad=
-1+\gamma S(t)(H_{t+1}-H_0).
\end{align}
Since $0\leq S(t)\leq1$, $H_0\geq0$, and
$H_{t+1}\leq1-\lambda$,
\begin{align}
\bigl(\mathcal B_{\star,W}U\bigr)(h)
-
\bigl(\mathcal B_{\star,M}U\bigr)(h)
&\leq
-1+\gamma(1-\lambda)
\nonumber\\
&<
-1+(1-\lambda)
\nonumber\\
&=
-\lambda.
\label{eq:fdrn-wait-uniform-dominance}
\end{align}
Thus, Wait is strictly below Maintain. Combining this with
\eqref{eq:fdrn-max-maintain-probe}, we obtain
\begin{align}
(\mathcal T_\star U)(h)
&=
\max_{a\in\mathcal A_{\mathrm{FRDN}}}
\bigl(\mathcal B_{\star,a}U\bigr)(h)
\nonumber\\
&=
u_t
=
U(h).
\end{align}
Hence $U$ is a bounded fixed point of the true-world Bellman optimality
operator. By uniqueness of the bounded fixed point,
\begin{align}
U(h)
=
V_\star^{\mathrm{opt}}(h)
\qquad
\text{for every }h\in\mathscr H.
\label{eq:fdrn-U-equals-optimal-value}
\end{align}

We may now identify the action-specific operators $\mathcal B_{\star,a}$ with the optimal
action-values through
\begin{align}
Q_\star^{\mathrm{opt}}(h,a)
&=
\bigl(\mathcal B_{\star,a}U\bigr)(h),
\end{align}
for which we use the expression in~\eqref{eq:true-optimal-action-value}. Equations~\eqref{eq:fdrn-probe-minus-maintain} and
\eqref{eq:fdrn-wait-uniform-dominance} therefore give
\begin{align}\label{eq:action_difference_frdn}
Q_\star^{\mathrm{opt}}(\ph_t,P)
-
Q_\star^{\mathrm{opt}}(\ph_t,M)
&=
S(t)-\lambda,
\nonumber\\
Q_\star^{\mathrm{opt}}(\ph_t,W)
-
Q_\star^{\mathrm{opt}}(\ph_t,M)
&<
-\lambda.
\end{align}
Moreover, because $S(t)\geq0$,
\begin{align}
Q_\star^{\mathrm{opt}}(\ph_t,W)
-
Q_\star^{\mathrm{opt}}(\ph_t,M)
&<
-\lambda \leq
S(t)-\lambda.
\end{align}
Thus using~\eqref{eq:action_difference_frdn} we determine that Wait is also strictly below Probe, and the two largest action-values are always those of Probe and Maintain.

Consequently using~\eqref{eq:action_difference_frdn}, Probe is uniquely optimal when $S(t)>\lambda$, Maintain is
uniquely optimal when $S(t)<\lambda$, and they are tied when
$S(t)=\lambda$. This proves
\eqref{eq:fdrn-decision-true-action}. Since the top two action-values are
those of Probe and Maintain, their separation is
\begin{align}
g_\star(t)
&=
\left|
Q_\star^{\mathrm{opt}}(\ph_t,P)
-
Q_\star^{\mathrm{opt}}(\ph_t,M)
\right|
\nonumber\\
&=
|S(t)-\lambda|,
\end{align}
which proves~\eqref{eq:fdrn-true-decision-margin}.
\end{proof}

We next determine how often the true optimal action switches. From
~\eqref{eq:periodic_funtion_tick},
\begin{align}
f(x)-\lambda
&=
\frac{
2\lambda B\sin(\alpha/2)
\sin(x+\varphi+\alpha/2)
}{
A-B\cos(x+\varphi)
}.
\label{eq:fdrn-threshold-sine}
\end{align}
The denominator is strictly positive. Moreover,
$\sin(\alpha/2)\neq0$ because $\alpha/\pi$ is irrational. Thus,
$f(x)-\lambda$ is a nonzero multiple of a shifted sine divided by a positive
function. The sets on which it is positive and negative each occupy one
half of a period. Since $S(t)=f(t\alpha)$ for $t\geq1$, Weyl
equidistribution~\eqref{eq:fdrn-weyl-residue-average} gives, for every
$r\in\{0,\ldots,L_N-1\}$,
\begin{align}
\lim_{K\to\infty}
\frac{1}{K}
\sum_{n=0}^{K-1}
\mathbbm 1\!\left\{S(nL_N+r)>\lambda\right\}
&=
\frac{1}{2},
\nonumber\\
\lim_{K\to\infty}
\frac{1}{K}
\sum_{n=0}^{K-1}
\mathbbm 1\!\left\{S(nL_N+r)<\lambda\right\}
&=
\frac{1}{2}.
\label{eq:fdrn-half-tick-residue}
\end{align}
The possible term with $nL_N+r=0$ does not affect either limit. Irrationality
also implies that $S(t)=\lambda$ for at most one integer $t\geq1$, so ties
do not affect the asymptotic frequencies.

A small result we will need for our proof is the continuity of the action-value and value functions with respect to the distributions of classical states.

\begin{lemma}[Continuity of finite-classical value functions]
\label{lem:classical-planning-values-continuous}
Fix a finite-dimensional classical world model and
$\gamma\in[0,1)$. Then
$v_C^{\mathrm{opt}}$ is continuous on $\mathcal{P}([N])$, and
$q_C^{\mathrm{opt}}(\cdot,a)$ is continuous for every
$a\in\mathcal A$.
\end{lemma}

\begin{proof}
Let $v$ be a bounded continuous function on $\mathcal{P}([N])$. For each
action $a$ and outcome $y$, consider the weighted branch term
\begin{align}
F_{a,y}^{v}(z)
:=
\mathrm{Pr}_C(y\mid z,a)\,
v\!\left(T_{C,y}^{a}(z)\right).
\end{align}
It is continuous wherever $\mathrm{Pr}_C(y\mid z,a)>0$. If $z_n\to z$ and
$\mathrm{Pr}_C(y\mid z,a)=0$, then
\begin{align}
\left|F_{a,y}^{v}(z_n)\right|
\leq
\|v\|_\infty \mathrm{Pr}_C(y\mid z_n,a)
\longrightarrow
0.
\end{align}
Hence every weighted branch term extends continuously through
zero-probability points.

The expected one-step reward $r_C(\cdot,a)$ is linear and therefore
continuous. It follows that the Bellman optimality operator maps continuous
functions to continuous functions. Starting from the zero function, its
iterates are continuous and converge uniformly to the unique fixed point
$v_C^{\mathrm{opt}}$. Thus $v_C^{\mathrm{opt}}$ is continuous.
Equation~\eqref{eq:bellman_action_value_optimal} then implies that
$q_C^{\mathrm{opt}}(\cdot,a)$ is continuous for every action $a$.
\end{proof}

We now combine continuity with the residue-class convergence of finite
classical memories. On each residue class, the model's action-value vector
converges. Its limiting vector either has a tie at the top, in which case the
predicted decision margin vanishes, or has a unique maximizer, in which case
the model eventually selects one fixed action. The latter cannot track the
true switching between Probe and Maintain.

\begin{theorem}[Limits of finite classical decisions]
\label{thm:fdrn-greedy-action-obstruction}
For every discount factor $\gamma\in[0,1)$, every $N\in\mathbb N$, every
$N$-dimensional classical world model over
$(\mathcal A_{\mathrm{FRDN}},\mathcal Y_{\mathrm{FRDN}})$, and every
model-greedy selection in~\eqref{eq:fdrn-model-greedy-action}, at least one
of the following alternatives holds:

\begin{itemize}

\item The decision margin of the classical model becomes arbitrarily small,
\begin{align}
\liminf_{t\to\infty}g_C(t)
=
0.
\label{eq:fdrn-vanishing-greedy-gap}
\end{align}

\item The model-selected action is truly suboptimal on at least half of the
prefixes of $\mathscr F_{\mathrm{tick}}$ in the long run,
\begin{align}
\liminf_{T\to\infty}
\frac1T
\sum_{t=1}^{T}
\mathbbm 1
\left\{
\widehat a_C(t)\notin\mathcal A_\star(t)
\right\}
\geq
\frac12.
\label{eq:fdrn-wrong-action-frequency}
\end{align}

\end{itemize}

Moreover, the first alternative holds if and only if there is a residue
$r\in\{0,\ldots,L_N-1\}$ for which
\begin{align}
g_C(nL_N+r)
\longrightarrow
0.
\label{eq:fdrn-greedy-gap-residue-zero}
\end{align}
\end{theorem}

\begin{proof}
If the model eventually assigns zero probability to the all-Tick branch, the
fixed-reference convention makes $E_C(\ph_t)$ eventually equal to
$z_{\mathrm{ref}}$. Otherwise,
Corollary~\ref{cor:fdrn-classical-belief-residue-convergence} applies.
Hence, in either case, for every residue
$r\in\{0,\ldots,L_N-1\}$ there is a distribution over states $z_{\infty,r}\in\mathcal{P}([N])$ such that
\begin{align}
E_C(\ph_{nL_N+r})
\longrightarrow
z_{\infty,r}.
\end{align}

Using the bridge relation~\eqref{eq:model-optimal-q-bridge} and
Lemma~\ref{lem:classical-planning-values-continuous}, the corresponding
action-value vectors satisfy
\begin{align}
\mathbf q_{n,r}
&:=
\left(
Q_C^{\mathrm{opt}}(\ph_{nL_N+r},a)
\right)_{a\in\mathcal A_{\mathrm{FRDN}}}
\nonumber\\
&\longrightarrow
\left(
q_C^{\mathrm{opt}}(z_{\infty,r},a)
\right)_{a\in\mathcal A_{\mathrm{FRDN}}}
=:
\mathbf q_r.
\label{eq:fdrn-decision-classical-q-residue-limit}
\end{align}
Let $g_r^\infty$ be the difference between the largest and second-largest
components of $\mathbf q_r$. The top two gap is a continuous function of a
finite vector, so
\begin{align}
g_C(nL_N+r)
\longrightarrow
g_r^\infty.
\label{eq:fdrn-greedy-gap-residue-limit}
\end{align}
Since the number of residue classes is finite,
\begin{align}
\liminf_{t\to\infty}g_C(t)
=
\min_{0\leq r<L_N}g_r^\infty.
\label{eq:fdrn-greedy-gap-liminf}
\end{align}
Equations~\eqref{eq:fdrn-greedy-gap-residue-limit} and
~\eqref{eq:fdrn-greedy-gap-liminf} prove the final equivalence in the
theorem. In particular, if $g_r^\infty=0$ for some $r$, the first
alternative holds.

Suppose instead that $g_r^\infty>0$ for every residue. Then $\mathbf q_r$ has a unique maximizing action, denoted by $a_r$, which is the model-greedy action~\eqref{eq:fdrn-model-greedy-action}. Thus, convergence implies
\begin{align}
\widehat a_C(nL_N+r)
=
a_r
\end{align}
for all sufficiently large $n$.

If $a_r=P$, Lemma~\ref{lem:fdrn-true-decisions-all-gamma} shows that the
selected action is suboptimal whenever
$S(nL_N+r)<\lambda$. If $a_r=M$, it is suboptimal whenever
$S(nL_N+r)>\lambda$. If $a_r=W$, it is suboptimal for every $n$.
Equation~\eqref{eq:fdrn-half-tick-residue} therefore shows that the
selected action is suboptimal on at least half of the histories in every
residue class. Averaging over the residue classes first for
$T=KL_N$, and then observing that an incomplete final block is negligible,
proves~\eqref{eq:fdrn-wrong-action-frequency}.
\end{proof}

The first alternative becomes an operational failure only if the predicted
margin vanishes away from true-world ties. The following corollary shows
that this is precisely what happens.

\begin{corollary}[Loss of decision resolution]
\label{cor:fdrn-loss-decision-resolution}
Under the hypotheses of
Theorem~\ref{thm:fdrn-greedy-action-obstruction}, suppose that
$\liminf_{t\to\infty}g_C(t)=0$. Then the classical world model loses decision resolution along
$\mathscr F_{\mathrm{tick}}$ in the sense of
Definition~\ref{def:fdrn-loss-decision-resolution}. More precisely, there is a
constant $\Delta_{\mathrm{res}}>0$ and an increasing sequence
$t_j\to\infty$ such that
\begin{align}
g_\star(t_j)
&\geq
\Delta_{\mathrm{res}},
&
g_C(t_j)
&\longrightarrow
0.
\label{eq:fdrn-classical-loss-decision-resolution}
\end{align}
The constant $\Delta_{\mathrm{res}}$ depends only on the true-world
function $f$ and the threshold $\lambda$, and is therefore independent of
$\gamma$, the classical memory dimension, and the chosen classical model.
\end{corollary}

\begin{proof}
By the final assertion of
Theorem~\ref{thm:fdrn-greedy-action-obstruction}, choose a residue
$r\in\{0,\ldots,L_N-1\}$ such that
\begin{align}
g_C(nL_N+r)
\longrightarrow
0.
\end{align}
It remains to choose a subsequence in this residue class on which the true
decision margin remains positive.

Equation~\eqref{eq:fdrn-threshold-sine} shows that
$f-\lambda$ is not identically zero. Define
\begin{align}
\Delta_{\mathrm{res}}
&:=
\frac{1}{2}
\max_{x\in[0,2\pi]}
|f(x)-\lambda|.
\end{align}
Then $\Delta_{\mathrm{res}}>0$. By continuity, there exists a nonempty
open interval $I_{\mathrm{res}}\subset[0,2\pi)$ such that
\begin{align}
|f(x)-\lambda|
&\geq
\Delta_{\mathrm{res}}
\qquad
\text{for every }x\in I_{\mathrm{res}}.
\end{align}
Define the $2\pi$-periodic indicator
\begin{align}
\chi_{I_{\mathrm{res}}}(x)
&:=
\mathbbm 1
\left\{
x\bmod 2\pi\in I_{\mathrm{res}}
\right\}.
\end{align}
This function is Riemann integrable, since it has discontinuities only at
the endpoints of $I_{\mathrm{res}}$. Because
$L_N\alpha/(2\pi)\notin\mathbb Q$, Weyl equidistribution~\eqref{eq:fdrn-weyl-residue-average} gives
\begin{align}
&\lim_{K\to\infty}
\frac{1}{K}
\sum_{n=0}^{K-1}
\chi_{I_{\mathrm{res}}}
\bigl((nL_N+r)\alpha\bigr)=
\frac{|I_{\mathrm{res}}|}{2\pi}
>
0.
\end{align}
Hence the phases $(nL_N+r)\alpha\bmod2\pi$ enter
$I_{\mathrm{res}}$ infinitely often. Choose an increasing sequence of such
visits $n_j$ and set $t_j
:=
n_jL_N+r$,
discarding a possible initial term with $t_j=0$. Then
\begin{align}
|S(t_j)-\lambda|
&=
\left|
f(t_j\alpha)-\lambda
\right|
\geq
\Delta_{\mathrm{res}}.
\end{align}
Using~\eqref{eq:fdrn-true-decision-margin}, we obtain $g_\star(t_j)
\geq
\Delta_{\mathrm{res}}$.
At the same time, $g_C(t_j)\to0$ by the choice of the residue class $r$.
This is precisely the condition in
\eqref{eq:fdrn-loss-decision-resolution}.
\end{proof}

If the predicted margin stays positive, each residue class instead has an
eventually fixed selected action. We now lower-bound the mean decision loss
of these actions on $\mathscr F_{\mathrm{tick}}$.

\begin{corollary}[Decision loss]
\label{cor:fdrn-decision-regret-gap}
Under the hypotheses of
Theorem~\ref{thm:fdrn-greedy-action-obstruction}, suppose that
\begin{align}
\liminf_{t\to\infty}
g_C(t)
&>
0.
\label{eq:fdrn-positive-greedy-margin}
\end{align}
Then the mean decision loss defined in
\eqref{eq:fdrn-mean-decision-loss} satisfies
\begin{align}
\overline{\ell}_C
\!\left(\mathscr F_{\mathrm{tick}}\right)
&\geq
\Delta_{\mathrm{dec}}
>
0,
\label{eq:fdrn-mean-decision-loss-bound}
\end{align}
where
\begin{align}
\Delta_{\mathrm{dec}}
:=
\frac{1}{2\pi}\min\Bigg\{
\int_0^{2\pi}
\bigl(\lambda-f(x)\bigr)_+\,dx, \nonumber\\
\int_0^{2\pi}
\bigl(f(x)-\lambda\bigr)_+\,dx
\Bigg\}
>
0.
\label{eq:fdrn-decision-regret-constant}
\end{align}
The constant $\Delta_{\mathrm{dec}}$ depends only on the true world and
is independent of both the classical memory dimension and the chosen model.
\end{corollary}

\begin{proof}
By~\eqref{eq:fdrn-greedy-gap-liminf}, the assumption
\eqref{eq:fdrn-positive-greedy-margin} implies
$g_r^\infty>0$ for every residue class. The proof of
Theorem~\ref{thm:fdrn-greedy-action-obstruction} therefore gives, for each
$r\in\{0,\ldots,L_N-1\}$, an action $a_r$ such that
\begin{align}
\widehat a_C(nL_N+r)
=
a_r
\end{align}
for all sufficiently large $n$.

Fix a residue $r$ and write $t=nL_N+r$. Lemma
\ref{lem:fdrn-true-decisions-all-gamma} shows that Wait is strictly below
both Probe and Maintain. Hence
\begin{align}
V_\star^{\mathrm{opt}}(\ph_t)
=
\max\left\{
Q_\star^{\mathrm{opt}}(\ph_t,P),
Q_\star^{\mathrm{opt}}(\ph_t,M)
\right\}.
\end{align}

If $a_r=P$, then, for all sufficiently large $n$,
\begin{align}
\ell_C(t)
&=
V_\star^{\mathrm{opt}}(\ph_t)
-
Q_\star^{\mathrm{opt}}(\ph_t,P)
\nonumber\\
&=
\left(
Q_\star^{\mathrm{opt}}(\ph_t,M)
-
Q_\star^{\mathrm{opt}}(\ph_t,P)
\right)_+
\nonumber\\
&=
\bigl(\lambda-S(t)\bigr)_+.
\end{align}
If $a_r=M$, then
\begin{align}
\ell_C(t)
&=
V_\star^{\mathrm{opt}}(\ph_t)
-
Q_\star^{\mathrm{opt}}(\ph_t,M)
\nonumber\\
&=
\left(
Q_\star^{\mathrm{opt}}(\ph_t,P)
-
Q_\star^{\mathrm{opt}}(\ph_t,M)
\right)_+
\nonumber\\
&=
\bigl(S(t)-\lambda\bigr)_+.
\label{eq:fdrn-probe-maintain-decision-loss}
\end{align}
Finally, if $a_r=W$, then
\begin{align}
\ell_C(t)
&\geq
Q_\star^{\mathrm{opt}}(\ph_t,M)
-
Q_\star^{\mathrm{opt}}(\ph_t,W)
>
\lambda.
\label{eq:fdrn-wait-decision-loss}
\end{align}

For the Probe case, Weyl equidistribution~\eqref{eq:fdrn-weyl-residue-average} gives
\begin{align}
&\lim_{K\to\infty}
\frac{1}{K}
\sum_{n=0}^{K-1}
\bigl(
\lambda-S(nL_N+r)
\bigr)_+
\nonumber\\
&\qquad=
\frac{1}{2\pi}
\int_0^{2\pi}
\bigl(\lambda-f(x)\bigr)_+
\,dx.
\end{align}
For the Maintain case, it gives
\begin{align}
&\lim_{K\to\infty}
\frac{1}{K}
\sum_{n=0}^{K-1}
\bigl(
S(nL_N+r)-\lambda
\bigr)_+
\nonumber\\
&\qquad=
\frac{1}{2\pi}
\int_0^{2\pi}
\bigl(f(x)-\lambda\bigr)_+
\,dx.
\end{align}
Both phase averages are strictly positive by
\eqref{eq:fdrn-threshold-sine}. By the definition
\eqref{eq:fdrn-decision-regret-constant}, each is at least
$\Delta_{\mathrm{dec}}$.

Moreover,
\begin{align}
0
\leq
\bigl(\lambda-f(x)\bigr)_+
\leq
\lambda,
\end{align}
and therefore
\begin{align}
\Delta_{\mathrm{dec}}
\leq
\frac{1}{2\pi}
\int_0^{2\pi}
\bigl(\lambda-f(x)\bigr)_+
\,dx
\leq
\lambda.
\end{align}
Thus, the uniform Wait loss in
\eqref{eq:fdrn-wait-decision-loss} is also at least
$\Delta_{\mathrm{dec}}$. Consequently, every residue class has
asymptotic mean decision loss at least $\Delta_{\mathrm{dec}}$.
Averaging over the finitely many residue classes gives
\begin{align}
\overline{\ell}_C
\!\left(\mathscr F_{\mathrm{tick}}\right)
\geq
\Delta_{\mathrm{dec}}.
\end{align}
\end{proof}

With
$\Delta_{\mathrm{tr}}
:=
\min\{\Delta_{\mathrm{res}},\Delta_{\mathrm{dec}}\}>0$,
Theorem~\ref{thm:fdrn-greedy-action-obstruction} and
Corollaries~\ref{cor:fdrn-loss-decision-resolution}
and~\ref{cor:fdrn-decision-regret-gap} show that, along the fixed reachable
trajectory $\mathscr F_{\mathrm{tick}}$, every finite-dimensional classical
world model either loses decision resolution at level
$\Delta_{\mathrm{tr}}$ or selects a true-world-suboptimal action with
lower asymptotic frequency at least $1/2$ and has mean decision loss at
least $\Delta_{\mathrm{tr}}$. Both the trajectory and the constant are
independent of the classical memory dimension and the chosen model.

\subsubsection{Other actions can be optimal}
\label{subsubsec:fdrn-wm-revised}

The preceding reward assignment makes Probe and Maintain the
competing optimal actions. We focused on this choice because it
gives particularly simple analytical bounds. Here we show that
other reward choices can make different actions optimal by
considering a setting in which Wait and Maintain compete.
The same arguments yield the same qualitative separation results,
although the numerical bounds change with the true world parameters
and rewards. We therefore only sketch how the previous arguments
apply, without reproducing the proofs.

\begin{proposition}[Wait and Maintain can both be optimal]
\label{prop:fdrn-wm-revised}
For every true world in Definition~\ref{def:fdrn-controlled-world}
with $0<\lambda\leq1/2$ and $\alpha/\pi\notin\mathbb Q$,
there exists $\eta\in(0,1)$ such that the reward assignment
\begin{align}
R_W=C_M&=\eta/2, & C_W&=1-\eta/2,\nonumber\\
R_P&=-\eta, & C_P&=\eta
\label{eq:wm-revised-rewards}
\end{align}
has the following properties simultaneously for every
$\gamma\in[0,1)$:
\begin{itemize}
\item There are sets of positive integer indices
$\mathcal T_W$ and $\mathcal T_M$, independent of $\gamma$,
such that Wait is uniquely optimal at $\ph_t$ whenever
$t\in\mathcal T_W$, and Maintain is uniquely optimal whenever
$t\in\mathcal T_M$.

\item Both sets are infinite. More precisely, each contains
a positive limiting fraction of the indices in every arithmetic
subsequence: there are constants $p_W,p_M>0$, independent of
$\gamma$, such that for every integer $L\geq1$, every
$r\in\{0,\ldots,L-1\}$, and each $a\in\{W,M\}$,
\begin{align}
\lim_{N\to\infty}\frac{1}{N}
\sum_{n=0}^{N-1}
\mathbbm 1\!\left\{nL+r\in\mathcal T_a\right\}
=p_a>0.
\end{align}

\item Probe is never optimal.
\end{itemize}
\end{proposition}

\begin{proof}
We use the rewards in \eqref{eq:wm-revised-rewards}, with
$\eta\in(0,1)$ to be chosen below.
Maintain and both Probe outcomes reset the clock. Therefore, for every $t\geq0$,
\begin{align}
&Q_\star^{\mathrm{opt}}(\ph_t,M)
=-\eta/2+\gamma V_\star^{\mathrm{opt}}(\ph_0),
\label{eq:wm-revised-maintain}\\
&Q_\star^{\mathrm{opt}}(\ph_t,P)
-Q_\star^{\mathrm{opt}}(\ph_t,M)
=-\eta/2.
\label{eq:wm-revised-probe}
\end{align}
Probe is thus always strictly below Maintain. Wait instead advances the
clock after a Tick and resets it after a Break. The difference of action values between Wait and Maintain is 
\begin{align}
&Q_\star^{\mathrm{opt}}(\ph_t,W)
-Q_\star^{\mathrm{opt}}(\ph_t,M)\nonumber\\
&\quad=S(t)-1+\eta
+\gamma S(t)\bigl[
 V_\star^{\mathrm{opt}}(\ph_{t+1})
 -V_\star^{\mathrm{opt}}(\ph_0)\bigr].
\label{eq:wm-revised-comparison-exact}
\end{align}

We first bound how much future decisions can improve on immediately
choosing Maintain, without assuming which action is optimal at reset.
By~\eqref{eq:wm-revised-maintain} Maintain action-value is independent of age, thus
\begin{align}
V_\star^{\mathrm{opt}}(\ph_0)
&\geq Q_\star^{\mathrm{opt}}(\ph_0,M)
=Q_\star^{\mathrm{opt}}(\ph_{t+1},M).
\end{align}
Using this in \eqref{eq:wm-revised-comparison-exact}, together with
$S(t)-1+\eta\leq\eta$ and Probe's strict suboptimality~\eqref{eq:wm-revised-probe}, gives
\begin{align}
0&\leq V_\star^{\mathrm{opt}}(\ph_t)
-Q_\star^{\mathrm{opt}}(\ph_t,M)\nonumber\\
&\leq\eta+\gamma S(t)\bigl[
V_\star^{\mathrm{opt}}(\ph_{t+1})
-Q_\star^{\mathrm{opt}}(\ph_{t+1},M)\bigr],
\label{eq:wm-revised-one-step-bound}
\end{align}
and now we can iterate this bound by bounding the terms with $t+1$ with $t+2$.
To iterate this bound, note that from the definition of $S(t)$~\eqref{eq:true_tick_prob} and the closed form~\eqref{eq:fdrn-tail-closed-form} imply, for $t,n\geq1$,
\begin{align}
\prod_{j=0}^{n-1}S(t+j)
&=\frac{\Phi(t+n)}{\Phi(t)}=\lambda^n
\frac{A-B\cos((t+n)\alpha+\varphi)}
     {A-B\cos(t\alpha+\varphi)} \nonumber \\
&\leq\frac{A+B}{A-B}\lambda^n.
\label{eq:wm-revised-product}
\end{align}
Thus, for every $t\geq1$, iterating
\eqref{eq:wm-revised-one-step-bound} yields
\begin{align}
V_\star^{\mathrm{opt}}(\ph_t)
&-Q_\star^{\mathrm{opt}}(\ph_t,M) \leq\eta\sum_{n=0}^{\infty}\gamma^n
                 \prod_{j=0}^{n-1}S(t+j)\nonumber\\
&\leq\eta\left[1+\frac{A+B}{A-B}
                         \frac{\lambda}{1-\lambda}\right]
=: \eta C_{\mathrm{surv}}.
\label{eq:wm-revised-value-bound}
\end{align}
The final bound follows from $\gamma^n\leq1$
and the geometric-series formula. The remainder tends to zero by
\eqref{eq:wm-revised-product}, because values are bounded for each fixed
$\gamma<1$. The constant $C_{\mathrm{surv}}$ bounds the expected number
of actions when repeatedly choosing Wait until the next Break,
including the action producing Break, starting at any $t\geq1$.

We now choose $\eta>0$ small enough that
\begin{align}
\eta(1+sC_{\mathrm{surv}})<1-s,
\qquad s\in\{S(0),\lambda^2\}.
\label{eq:wm-revised-choice}
\end{align}
Such a choice exists: $\lambda^2<1$, and the lifetime law
\eqref{eq:fdrn-lifetime-law} gives
$S(0)<\lambda/(1-\lambda)\leq1$, with strictness following already from
$\sin^2(\alpha/2)<1$.

Now we can bound the best action at reset. Using~\eqref{eq:wm-revised-value-bound} at age one in
\eqref{eq:wm-revised-comparison-exact} gives
\begin{align}
&Q_\star^{\mathrm{opt}}(\ph_0,W)
-Q_\star^{\mathrm{opt}}(\ph_0,M)\nonumber\\
&\quad\leq(1+\eta C_{\mathrm{surv}})S(0)-1+\eta<0,
\label{eq:wm-revised-reset}
\end{align}
for which we used the choice of $\eta$ with $s = S(0)$.
Recall again that the expression of $Q_\star^{\mathrm{opt}}(\ph_t,M)$ is independent of $t$ by~\eqref{eq:wm-revised-maintain}. Therefore, at reset the unique optimal action is Maintain, so
\begin{align}\label{eq:reset_optimal}
V_\star^{\mathrm{opt}}(\ph_0)=Q_\star^{\mathrm{opt}}(\ph_t,M),\quad  \text{for every } t.
\end{align}

Now we are gonna determine for times $t\geq 1$ the conditions that make Wait and Mantain optimal. Using~\eqref{eq:reset_optimal} together with~\eqref{eq:wm-revised-comparison-exact} and using
\eqref{eq:wm-revised-value-bound} now gives
\begin{align}
S(t)-1+\eta
&\leq Q_\star^{\mathrm{opt}}(\ph_t,W)
       -Q_\star^{\mathrm{opt}}(\ph_t,M)\nonumber\\
&\leq(1+\eta C_{\mathrm{surv}})S(t)-1+\eta.
\label{eq:wm-revised-comparison}
\end{align}
Consequently, define the following sets of positive integer indices for Wait and Maintain
\begin{align}
\mathcal T_W
&=\left\{t\geq1:S(t)>1-\eta\right\},\\
\mathcal T_M
&=\left\{t\geq1:
S(t)<\frac{1-\eta}{1+\eta C_{\mathrm{surv}}}\right\}.
\end{align}
The key point is that by~\eqref{eq:wm-revised-comparison} Wait is uniquely optimal at $\ph_t$ whenever $t\in\mathcal T_W$,
and Maintain whenever $t\in\mathcal T_M$. 

It remains to prove that both of these sets are infinite. Recall that $S(t)=f(t\alpha)$ for $t\geq1$, where
\begin{align}
f(x)
=\lambda\frac{A-B\cos(x+\alpha+\varphi)}
                 {A-B\cos(x+\varphi)}.
\end{align}
The denominator is strictly positive because $A>B$.
The definitions of $A$, $B$, and $\varphi$ give
\begin{align}
A(1-\lambda)&=\frac12,
&
B e^{i\varphi}&=\frac{1}{2(1-\lambda e^{i\alpha})}.
\end{align}
Substituting these identities into $1-f(x)$ and
$f(x)-\lambda^2$ yields
\begin{align}
1-f(x)
&=\frac{\sin^2(x/2)}
        {A-B\cos(x+\varphi)},\\
f(x)-\lambda^2
&=\frac{\lambda\sin^2((x+\alpha+2\varphi)/2)}
        {A-B\cos(x+\varphi)}.
\end{align}
Thus $\lambda^2\leq f(x)\leq1$. The upper bound is attained
at $x=0$, and the lower bound at
$x=-\alpha-2\varphi$ modulo $2\pi$.
Since $f$ is continuous and $2\pi$-periodic, its exact range is
\begin{align}
f([0,2\pi))=[\lambda^2,1].
\label{eq:wm-revised-range}
\end{align}
By the choice of $\eta$ in~\eqref{eq:wm-revised-choice} we have
\begin{align}
\lambda^2
<\frac{1-\eta}{1+\eta C_{\mathrm{surv}}}
<1-\eta<1.
\end{align}
Therefore, each strict inequality defining $\mathcal T_W$ and
$\mathcal T_M$ holds on an interval of values of $x$ of positive
length. In particular, the constants
\begin{align}
p_W
&=\frac{1}{2\pi}\int_0^{2\pi}
\mathbbm 1\!\left\{f(x)>1-\eta\right\}\,dx,\\
p_M
&=\frac{1}{2\pi}\int_0^{2\pi}
\mathbbm 1\!\left\{
f(x)<\frac{1-\eta}{1+\eta C_{\mathrm{surv}}}
\right\}\,dx
\end{align}
are both strictly positive.

To apply Weyl equidistribution
\eqref{eq:fdrn-weyl-residue-average}, we check that the two
indicator functions are Riemann integrable.
Since $A>B$, the function $f$ is continuous. Each indicator
can therefore be discontinuous only where $f(x)=q$, for
$q\in\{1-\eta,(1-\eta)/(1+\eta C_{\mathrm{surv}})\}$.
This equation is equivalent to
\begin{align}
A(\lambda-q)
=B\bigl[&
(\lambda\cos\alpha-q)\cos(x+\varphi)\nonumber\\
&-\lambda\sin\alpha\sin(x+\varphi)\bigr].
\end{align}
Since $\lambda\sin\alpha\neq0$, this defines a straight line
in the coordinates $(\cos(x+\varphi),\sin(x+\varphi))$.
These coordinates trace the unit circle exactly once as
$x$ ranges over $[0,2\pi)$, and a line intersects a circle
at at most two points. Thus each indicator has at most two
discontinuities per period and is Riemann integrable,
allowing us to apply Weyl equidistribution  \eqref{eq:fdrn-weyl-residue-average}. 

For every integer $L\geq1$, the number $L\alpha/(2\pi)$
is irrational. Hence \eqref{eq:fdrn-weyl-residue-average} gives,
for every $r\in\{0,\ldots,L-1\}$ and $a\in\{W,M\}$,
\begin{align}
\lim_{N\to\infty}\frac{1}{N}
\sum_{n=0}^{N-1}
\mathbbm 1\!\left\{nL+r\in\mathcal T_a\right\}
=p_a>0.
\end{align}
The possible index $nL+r=0$ changes at most one term and does
not affect the limit. Thus, among the histories
$\ph_r,\ph_{L+r},\ph_{2L+r},\ldots$, the fraction whose indices
belong to $\mathcal T_a$ tends to $p_a$, and action $a$ is
uniquely optimal at every such history. The sets and these
limiting fractions are independent of $\gamma$.
\end{proof}

Since $\mathcal T_W$ and $\mathcal T_M$ have positive frequencies
on every arithmetic subsequence, we can use the same classical
arguments as for Probe versus Maintain. Choosing closed intervals
of positive length inside the two strict optimality regions gives
a uniform positive true decision margin, by
\eqref{eq:wm-revised-probe} and \eqref{eq:wm-revised-comparison}.
The proofs of Theorem~\ref{thm:fdrn-greedy-action-obstruction}
and Corollaries~\ref{cor:fdrn-loss-decision-resolution}
and~\ref{cor:fdrn-decision-regret-gap} therefore give the same
separation: every finite-dimensional classical world model either
loses decision resolution or recommends suboptimal actions on a
nonzero fraction of the queried histories, with a positive lower
bound on mean decision loss in the latter case.

More precisely, if the model does not lose decision resolution,
the fraction of its model-greedy recommendations that are
suboptimal along $\mathscr F_{\mathrm{tick}}$ is bounded below
in the limit by $\min\{p_W,p_M\}>0$. This bound and the mean decision-loss bound
depend on the numerical values of the true world parameters and rewards,
but are independent of the classical model, its memory size,
and $\gamma$.

\subsection{Estimation errors for value functions}
\label{subsec:fdrn-planning-value-errors}

The preceding subsection concerned the ordering of the action-values. We now
study their numerical accuracy along the all-Tick trajectory and focus on the reward choice that makes optimal Probe and Mantain. 

For a world model $\mathsf M$ and history $h$, define
\begin{align}
e_{\mathrm Q}^{M}(h)
&:=
\max_{a\in\mathcal A_{\mathrm{FRDN}}}
\left|
Q_\star^{\mathrm{opt}}(h,a)
-
Q_M^{\mathrm{opt}}(h,a)
\right|,
\label{eq:fdrn-action-value-error}
\\
e_V^{M}(h)
&:=
\left|
V_\star^{\mathrm{opt}}(h)
-
V_M^{\mathrm{opt}}(h)
\right|.
\label{eq:fdrn-optimal-value-error}
\end{align}
Along a reachable trajectory $\mathscr F=h_1h_2\cdots$, with prefixes
$\ph_t=h_1\cdots h_t$, define
\begin{align}
\overline e_{\mathrm Q}^{M}(\mathscr F)
&:=
\liminf_{T\to\infty}
\frac{1}{T}
\sum_{t=1}^{T}
e_{\mathrm Q}^{M}(\ph_t),
\label{eq:fdrn-mean-action-value-error}
\\
\overline e_V^{M}(\mathscr F)
&:=
\liminf_{T\to\infty}
\frac{1}{T}
\sum_{t=1}^{T}
e_V^{M}(\ph_t).
\label{eq:fdrn-mean-optimal-value-error}
\end{align}
We apply these quantities below to finite-dimensional classical world
models along $\mathscr F_{\mathrm{tick}}$.

Lemma~\ref{lem:fdrn-true-decisions-all-gamma} already determines the
relative true-world action-values: Probe differs from Maintain by
$S(t)-\lambda$, while Wait is strictly below Maintain. To obtain the
absolute value functions needed below, it therefore remains only to compute
the common baseline supplied by Maintain. Because Maintain gives the
immediate reward $-\eta$ and resets the clock, this requires a single
Bellman step.

For convenience, define the continuous $2\pi$-periodic function
\begin{align}
H(x)
&:=
\bigl(f(x)-\lambda\bigr)_+,
\label{eq:fdrn-optimal-value-profile}
\end{align}
and the age-independent reset baseline
\begin{align}
b_\gamma
&:=
-\eta
+\gamma V_\star^{\mathrm{opt}}(\ph_0).
\label{eq:fdrn-reset-baseline}
\end{align}

\begin{lemma}[True FRDN value functions]
\label{lem:fdrn-all-gamma-true-value}
For the rewards in~\eqref{eq:fdrn-decision-rewards}, every
$\gamma\in[0,1)$, and every $t\geq0$, the optimal value and Probe
action-value at the all-Tick prefix $\ph_t$ satisfy
\begin{align}
V_\star^{\mathrm{opt}}(\ph_t)
&=
-\eta
+\gamma V_\star^{\mathrm{opt}}(\ph_0)
+\bigl(S(t)-\lambda\bigr)_+,
\label{eq:fdrn-all-gamma-true-value}
\\
Q_\star^{\mathrm{opt}}(\ph_t,P)
&=
-\eta
+\gamma V_\star^{\mathrm{opt}}(\ph_0)
+S(t)-\lambda .
\label{eq:fdrn-all-gamma-true-probe-value}
\end{align}
In particular, for $t\geq1$,
\begin{align}
V_\star^{\mathrm{opt}}(\ph_t)
&=
-\eta
+\gamma V_\star^{\mathrm{opt}}(\ph_0)
+H(t\alpha),
\label{eq:fdrn-true-planning-phase-value}
\\
Q_\star^{\mathrm{opt}}(\ph_t,P)
&=
-\eta
+\gamma V_\star^{\mathrm{opt}}(\ph_0)
+f(t\alpha)-\lambda .
\label{eq:fdrn-true-planning-phase-probe}
\end{align}
Moreover, the reset value is
\begin{align}
V_\star^{\mathrm{opt}}(\ph_0)
&=
\frac{
-\eta+\bigl(S(0)-\lambda\bigr)_+
}{
1-\gamma
}.
\label{eq:fdrn-reset-value}
\end{align}
\end{lemma}

\begin{proof}
Fix $t\geq0$. Maintain gives the deterministic immediate reward $-\eta$
and resets the clock to age zero. Hence, by the Bellman action-value
relation~\eqref{eq:true-optimal-bellman-relations},
\begin{align}
Q_\star^{\mathrm{opt}}(\ph_t,M)
&=
-\eta
+\gamma V_\star^{\mathrm{opt}}(\ph_0)=
b_\gamma.
\label{eq:fdrn-maintain-reset-baseline}
\end{align}
Thus, $b_\gamma$ is the age-independent action-value of Maintain.

The first relation in
\eqref{eq:fdrn-decision-wait-dominated} gives
\begin{align}
Q_\star^{\mathrm{opt}}(\ph_t,P)
&=
Q_\star^{\mathrm{opt}}(\ph_t,M)
+
S(t)-\lambda
\nonumber\\
&=
b_\gamma+S(t)-\lambda
\nonumber\\
&=
-\eta
+\gamma V_\star^{\mathrm{opt}}(\ph_0)
+S(t)-\lambda.
\end{align}
This proves~\eqref{eq:fdrn-all-gamma-true-probe-value}.

The second relation in
\eqref{eq:fdrn-decision-wait-dominated} shows that Wait is strictly below
Maintain. Therefore, the value--action-value relation
\eqref{eq:true-optimal-value-action-value-relation} reduces to
\begin{align}
V_\star^{\mathrm{opt}}(\ph_t)
&=
\max\left\{
Q_\star^{\mathrm{opt}}(\ph_t,M),
Q_\star^{\mathrm{opt}}(\ph_t,P)
\right\}
\nonumber\\
&=
\max\left\{
b_\gamma,
b_\gamma+S(t)-\lambda
\right\}
\nonumber\\
&=
b_\gamma
+\bigl(S(t)-\lambda\bigr)_+
\nonumber\\
&=
-\eta
+\gamma V_\star^{\mathrm{opt}}(\ph_0)
+\bigl(S(t)-\lambda\bigr)_+.
\end{align}
This proves~\eqref{eq:fdrn-all-gamma-true-value}.

Taking $t=0$ in this identity gives
\begin{align}
V_\star^{\mathrm{opt}}(\ph_0)
&=
-\eta
+\gamma V_\star^{\mathrm{opt}}(\ph_0)
+\bigl(S(0)-\lambda\bigr)_+.
\end{align}
Solving this scalar fixed-point equation yields
\begin{align}
V_\star^{\mathrm{opt}}(\ph_0)
&=
\frac{
-\eta+\bigl(S(0)-\lambda\bigr)_+
}{
1-\gamma
},
\end{align}
which is~\eqref{eq:fdrn-reset-value}.

Finally, for $t\geq1$,
\eqref{eq:fdrn-survival-closed-form} gives
$S(t)=f(t\alpha)$, and hence
\begin{align}
\bigl(S(t)-\lambda\bigr)_+
&=
\bigl(f(t\alpha)-\lambda\bigr)_+
=
H(t\alpha).
\end{align}
Substitution into
\eqref{eq:fdrn-all-gamma-true-value} and
\eqref{eq:fdrn-all-gamma-true-probe-value} gives
\eqref{eq:fdrn-true-planning-phase-value} and
\eqref{eq:fdrn-true-planning-phase-probe}, respectively.
\end{proof}

The baseline $b_\gamma$ in~\eqref{eq:fdrn-reset-baseline} is the
age-independent action-value of Maintain. The Probe action-value and optimal
value can therefore be written as
\begin{align}
Q_\star^{\mathrm{opt}}(\ph_t,P)
&=
b_\gamma+S(t)-\lambda,
\nonumber\\
V_\star^{\mathrm{opt}}(\ph_t)
&=
b_\gamma+\bigl(S(t)-\lambda\bigr)_+.
\label{eq:fdrn-true-planning-baseline-form}
\end{align}
Thus, Probe adds the signed, age-dependent advantage $S(t)-\lambda$, while
optimization between Probe and Maintain replaces this signed profile by its
positive part. Consequently, for $t\geq1$, the discount factor changes only
the additive baseline $b_\gamma$; the nonconstant phase profiles
$f-\lambda$ and $H$ are independent of $\gamma$. 

\begin{lemma}[Residue limits of finite-dimensional classical value functions]
\label{lem:fdrn-classical-planning-residue-limits}
Fix $\gamma\in[0,1)$ and an $N$-dimensional classical world model over
$(\mathcal A_{\mathrm{FRDN}},\mathcal Y_{\mathrm{FRDN}})$. For every
$r\in\{0,\ldots,L_N-1\}$, there exist constants $q_{P,r}$ and $v_r$ such
that
\begin{align}
Q_C^{\mathrm{opt}}(\ph_{nL_N+r},P)
&\longrightarrow
q_{P,r},
&
V_C^{\mathrm{opt}}(\ph_{nL_N+r})
&\longrightarrow
v_r.
\label{eq:fdrn-classical-planning-residue-limits}
\end{align}
\end{lemma}

\begin{proof}
If $\Phi_C(t_0)=0$ for some $t_0$, the fixed-reference convention makes
$E_C(\ph_t)$ equal to $z_{\mathrm{ref}}$ for every $t\geq t_0$. Otherwise,
Corollary~\ref{cor:fdrn-classical-belief-residue-convergence} gives, for
each residue $r$, a distribution over states $z_{\infty,r}$ such that
$E_C(\ph_{nL_N+r})\to z_{\infty,r}$. Hence such a limit $z_{\infty,r}$ exists in either case.
The bridge relations~\eqref{eq:model-optimal-value-bridge}
and~\eqref{eq:model-optimal-q-bridge}, together with
Lemma~\ref{lem:classical-planning-values-continuous}, give
\eqref{eq:fdrn-classical-planning-residue-limits}.
\end{proof}

We first use the Probe action to lower-bound the action-value error.

\begin{theorem}[Action-value error for finite classical models]
\label{thm:fdrn_q_gap}
Fix the rewards in~\eqref{eq:fdrn-decision-rewards}. For every
$\gamma\in[0,1)$, every $N\in\mathbb N$, and every $N$-dimensional
classical world model over
$(\mathcal A_{\mathrm{FRDN}},\mathcal Y_{\mathrm{FRDN}})$,
\begin{align}
\overline e_{\mathrm Q}^{C}
\!\left(
\mathscr F_{\mathrm{tick}}
\right)
&\geq
\kappa_{\mathrm{FRDN}}
>
0.
\label{eq:fdrn-mean-action-value-error-bound}
\end{align}
\end{theorem}

\begin{proof}
Fix the classical model and a residue
$r\in\{0,\ldots,L_N-1\}$. By
Lemma~\ref{lem:fdrn-classical-planning-residue-limits},
\begin{align}
c_{n,r}
&:=
\lambda
+Q_C^{\mathrm{opt}}(\ph_{nL_N+r},P)
-b_\gamma
\nonumber\\
&\longrightarrow
c_r
:=
\lambda+q_{P,r}-b_\gamma , 
\label{eq:fdrn-classical-probe-rescaled-limit}
\end{align}
where $ b_\gamma = -\eta+\gamma V_\star^{\mathrm{opt}}(\ph_0)$.
For $nL_N+r\geq1$, Lemma~\ref{lem:fdrn-all-gamma-true-value} gives
\begin{align}
&\left|
Q_\star^{\mathrm{opt}}(\ph_{nL_N+r},P)
-Q_C^{\mathrm{opt}}(\ph_{nL_N+r},P)
\right|
\nonumber\\
&\qquad=
\left|
f\bigl((nL_N+r)\alpha\bigr)-c_{n,r}
\right|.
\label{eq:fdrn-probe-error-rescaled}
\end{align}
The possible term with $nL_N+r=0$ does not affect the average. Applying
Lemma~\ref{lem:cesaro-weyl-stability} to $f$ gives
\begin{align}
&\lim_{K\to\infty}
\frac1K
\sum_{n=0}^{K-1}
\left|
Q_\star^{\mathrm{opt}}(\ph_{nL_N+r},P)
-Q_C^{\mathrm{opt}}(\ph_{nL_N+r},P)
\right|
\nonumber\\
&\qquad=
\frac1{2\pi}
\int_0^{2\pi}|f(x)-c_r|\,dx
\geq
\kappa_{\mathrm{FRDN}}.
\label{eq:fdrn-probe-error-residue-average}
\end{align}
For the final inequality, if $c_r\notin[0,1]$, project it onto $[0,1]$.
Since $f(x)\in[0,1]$, this projection cannot increase the integrand, and
the claim follows from~\eqref{eq:kappa-fdrn}.

By~\eqref{eq:fdrn-action-value-error}, the error $e_{\mathrm Q}^{C}$ dominates the Probe action-value error.
Decomposing an average of length $KL_N$ into the $L_N$ residue classes and
using~\eqref{eq:fdrn-probe-error-residue-average} therefore gives
\begin{align}
\liminf_{K\to\infty}
\frac1{KL_N}
\sum_{t=0}^{KL_N-1}e_{\mathrm Q}^{C}(\ph_t)
&\geq
\kappa_{\mathrm{FRDN}}.
\end{align}
For fixed $\gamma$, all value functions are bounded, so shifting from the
indices $0,\ldots,KL_N-1$ to $1,\ldots,KL_N$ is asymptotically
irrelevant. For arbitrary $T$, let $K=\lfloor T/L_N\rfloor$.
Nonnegativity gives
\begin{align}
\frac1T\sum_{t=1}^T e_{\mathrm Q}^{C}(\ph_t)
&\geq
\frac{KL_N}{T}
\left(
\frac1{KL_N}\sum_{t=1}^{KL_N}e_{\mathrm Q}^{C}(\ph_t)
\right),
\end{align}
and $KL_N/T\to1$. This proves
\eqref{eq:fdrn-mean-action-value-error-bound}.
\end{proof}

The previous theorem uses a single action. For the optimal value, the
relevant quantity that gives the lower bound is
\begin{align}
\kappa_{\mathrm{val}}
&:=
\min_{c\in[0,1-\lambda]}
\frac1{2\pi}
\int_0^{2\pi}|H(x)-c|\,dx.
\label{eq:fdrn-optimal-value-gap-constant}
\end{align}
Equation~\eqref{eq:fdrn-threshold-sine} shows that $H$ vanishes on one
nonempty open set and is positive on another. Hence $H$ is continuous and
nonconstant, so $\kappa_{\mathrm{val}}>0$.

\begin{theorem}[Optimal-value error for finite classical models]
\label{thm:fdrn_value_gap}
Fix the rewards in~\eqref{eq:fdrn-decision-rewards}. For every
$\gamma\in[0,1)$, every $N\in\mathbb N$, and every $N$-dimensional
classical world model over
$(\mathcal A_{\mathrm{FRDN}},\mathcal Y_{\mathrm{FRDN}})$,
\begin{align}
\overline e_V^{C}
\!\left(
\mathscr F_{\mathrm{tick}}
\right)
&\geq
\kappa_{\mathrm{val}}
>
0.
\label{eq:fdrn-mean-optimal-value-error-bound}
\end{align}
\end{theorem}

\begin{proof}
Fix the classical model and a residue
$r\in\{0,\ldots,L_N-1\}$. By
Lemma~\ref{lem:fdrn-classical-planning-residue-limits},
\begin{align}
d_{n,r}
&:=
V_C^{\mathrm{opt}}(\ph_{nL_N+r})-b_\gamma
\longrightarrow
d_r
:=
v_r-b_\gamma,
\label{eq:fdrn-classical-value-shifted-limit}
\end{align}
where $ b_\gamma = -\eta+\gamma V_\star^{\mathrm{opt}}(\ph_0)$. For $nL_N+r\geq1$, Lemma~\ref{lem:fdrn-all-gamma-true-value} gives
\begin{align}
&\left|
V_\star^{\mathrm{opt}}(\ph_{nL_N+r})
-V_C^{\mathrm{opt}}(\ph_{nL_N+r})
\right|
\nonumber\\
&\qquad=
\left|
H\bigl((nL_N+r)\alpha\bigr)-d_{n,r}
\right|.
\end{align}
Lemma~\ref{lem:cesaro-weyl-stability} applied to $H$ therefore gives
\begin{align}
&\lim_{K\to\infty}
\frac1K
\sum_{n=0}^{K-1}
\left|
V_\star^{\mathrm{opt}}(\ph_{nL_N+r})
-V_C^{\mathrm{opt}}(\ph_{nL_N+r})
\right|
\nonumber\\
&\qquad=
\frac1{2\pi}
\int_0^{2\pi}|H(x)-d_r|\,dx
\geq
\kappa_{\mathrm{val}}.
\label{eq:fdrn-value-error-residue-average}
\end{align}
As before, the possible index $nL_N+r=0$ is irrelevant. If
$d_r\notin[0,1-\lambda]$, projection onto this interval cannot increase
the distance from $H(x)\in[0,1-\lambda]$, which proves the final
inequality.

Decomposing into residue classes as in the proof of
Theorem~\ref{thm:fdrn_q_gap}, and then accounting for the incomplete final
block, proves~\eqref{eq:fdrn-mean-optimal-value-error-bound}.
\end{proof}

Taking
$\Delta_{\mathrm{est}}
:=
\min\{\kappa_{\mathrm{FRDN}},\kappa_{\mathrm{val}}\}>0$,
Theorems~\ref{thm:fdrn_q_gap} and~\ref{thm:fdrn_value_gap} show that both
mean errors are at least $\Delta_{\mathrm{est}}$ along the same fixed
trajectory $\mathscr F_{\mathrm{tick}}$, uniformly over all finite
classical memory dimensions and all $\gamma\in[0,1)$.

\section{Exact qutrit realization}
\label{sec:fdrn-quantum-realization}

This appendix constructs a quantum world model
$\mathsf Q_{\mathrm{FRDN}}$ with a three-dimensional memory, together with
a separate encoder $E_Q$, such that the pair
$(\mathsf Q_{\mathrm{FRDN}},E_Q)$ exactly realizes the FRDN true world of
Definition~\ref{def:fdrn-controlled-world}.

The construction adapts the qutrit realization of the FRDN renewal process
in~\cite{fanizza2024quantum} to the quantum-instrument world models of
Definition~\ref{def:quantum-world-model}. We first specify the instrument
associated with each action, then identify the memory states used to
initialize history-dependent queries, and finally verify the probability and
memory-update conditions in~\eqref{eq:predictive-consistency-probability} and~\eqref{eq:predictive-consistency-update}.

\subsection{Qutrit instruments for Wait, Maintain, and Probe}\label{subsec:fdrn-quantum-realization}

Let $\mathcal H_Q\simeq\mathbb C^3$, with computational basis
$\{|0\rangle,|1\rangle,|2\rangle\}$, and let
\begin{align}
    \Pi_{01}
    :=
    |0\rangle\langle0|
    +
    |1\rangle\langle1| .
\end{align}
We write $X$ and $Z$ for the Pauli matrices on
$\Pi_{01}\mathcal H_Q$, and write $\mathbb I_3$ for the identity operator on
$\mathcal H_Q$. Since $\alpha/\pi\notin\mathbb Q$, one has
\begin{align}
    0
    <
    \frac{1-\lambda}{|1-\lambda e^{i\alpha}|}
    <
    1.
\end{align}
Because $\tanh(2r)$ increases continuously from $0$ to $1$ for $r>0$,
there is a unique $r>0$ satisfying
\begin{align}
    \tanh(2r)
    =
    \frac{1-\lambda}{|1-\lambda e^{i\alpha}|}.
    \label{eq:fdrn-qutrit-r}
\end{align}

We start by defining the operators that will form the Wait branch. Define the Tick operator
\begin{align}
    A_T
    &:=
    K_{r,\alpha}\Pi_{01},
    &
    K_{r,\alpha}
    &:=
    \sqrt{\lambda}\,
    e^{-rX}e^{i\alpha Z/2}e^{rX}.
    \label{eq:fdrn-qutrit-tick-operator}
\end{align}
Thus $A_T|2\rangle=0$. For the choice~\eqref{eq:fdrn-qutrit-r}, a direct
calculation gives
\begin{align}
    \det(K_{r,\alpha}^{\dagger}K_{r,\alpha})
    =
    \lambda^2,
    \qquad
    \operatorname{Tr}(K_{r,\alpha}^{\dagger}K_{r,\alpha})
    =
    1+\lambda^2 .
\end{align}
The eigenvalues of $K_{r,\alpha}^{\dagger}K_{r,\alpha}$ are therefore
$1$ and $\lambda^2$. Hence
$A_T^\dagger A_T\leq\mathbb I_3$, so
$\mathbb I_3-A_T^\dagger A_T$ is positive. The two effects
$A_T^\dagger A_T$ and
$\mathbb I_3-A_T^\dagger A_T$ will represent Tick and Break,
respectively.

We now specify the reset state. Choose $\theta_\alpha$ by
\begin{align}
    \tan\theta_\alpha
    =
    e^{2r}
    \tan\!\left[
        \frac12
        \arctan\!\left(
            \frac{\lambda\sin\alpha}{1-\lambda\cos\alpha}
        \right)
    \right],
\end{align}
and define
\begin{align}
    |\xi\rangle
    :=
    \frac{e^{i\theta_\alpha}}{\sqrt2}|0\rangle
    +
    \frac{e^{-i\theta_\alpha}}{\sqrt2}|1\rangle .
\end{align}
Finally set
\begin{align}
    \omega
    &:=
    \frac{
        \lambda(1+\lambda)\sin^2(\alpha/2)
    }{
        (1-\lambda)(1+\lambda^2-2\lambda\cos\alpha)
    },
    \\
    \rho_0
    &:=
    \omega|\xi\rangle\langle\xi|
    +(1-\omega)|2\rangle\langle2| .
    \label{eq:fdrn-qutrit-reset-state}
\end{align}
For $0<\lambda\leq1/2$, one has $0<\omega\leq1$.

For each action, the following are the only nonzero branches of the
corresponding instrument; all other branches indexed by
$y\in\mathcal Y_{\mathrm{FRDN}}$ are the zero map.

For Wait,
\begin{align}
    \mathcal E_{y_T^W}^{(W)}(\rho)
    &:=
    A_T\rho A_T^\dagger,
    &
    \mathcal E_{y_B^W}^{(W)}(\rho)
    &:=
    \operatorname{Tr}\!\left[
        (\mathbb I_3-A_T^\dagger A_T)\rho
    \right]\rho_0 .
    \label{eq:fdrn-qutrit-wait-branches}
\end{align}
For Maintain,
\begin{align}
    \mathcal E_{y_B^M}^{(M)}(\rho)
    &:=
    \operatorname{Tr}(\rho)\rho_0 .
    \label{eq:fdrn-qutrit-maintain-branches}
\end{align}
For Probe,
\begin{align}
    \mathcal E_{y_T^P}^{(P)}(\rho)
    &:=
    \operatorname{Tr}\!\left[
        A_T^\dagger A_T\rho
    \right]\rho_0, \\
    \mathcal E_{y_B^P}^{(P)}(\rho)
    &:=
    \operatorname{Tr}\!\left[
        (\mathbb I_3-A_T^\dagger A_T)\rho
    \right]\rho_0 .
    \label{eq:fdrn-qutrit-probe-branches}
\end{align}
These maps define a valid quantum instrument for each action. The Wait--Tick
branch is a Kraus map, while the remaining nonzero branches are
measure-and-prepare maps associated with positive effects. Moreover,
\begin{align}
    \operatorname{Tr}(A_T^\dagger A_T\rho)
    +
    \operatorname{Tr}\!\left[
        (\mathbb I_3-A_T^\dagger A_T)\rho
    \right]
    =
    \operatorname{Tr}\rho .
\end{align}
Thus the Wait and Probe instruments are trace preserving when their branches
are summed. The Maintain map is trace preserving because
$\operatorname{Tr}\rho_0=1$.

\subsection{Clock memories and query initialization}

To derive the normalized clock memories, it is useful first to retain the
probability of the all-Tick branch in the trace of a subnormalized operator.
Define
\begin{align}
    \widetilde\rho_t
    :=
    \left(\mathcal E_{y_T^W}^{(W)}\right)^t(\rho_0),
    \qquad
    t\geq0 .
\end{align}
For $t\geq1$, the first Tick branch removes the $|2\rangle$ component, hence
\begin{align}
    \widetilde\rho_t
    & =
    \omega\,K_{r,\alpha}^t
    |\xi\rangle\langle\xi|
    (K_{r,\alpha}^\dagger)^t,
    \\
    K_{r,\alpha}^t
    & =
    \lambda^{t/2}
    e^{-rX}e^{it\alpha Z/2}e^{rX}.
    \label{eq:fdrn-qutrit-tick-power}
\end{align}
Using~\eqref{eq:fdrn-qutrit-r}--\eqref{eq:fdrn-qutrit-reset-state}, one obtains
\begin{align}
    \operatorname{Tr}\widetilde\rho_t
    &=
    \lambda^t
    \left(
        \frac{1}{2(1-\lambda)}
        -
        \frac{1}{4}
        \frac{e^{it\alpha}}{1-\lambda e^{i\alpha}}
        -
        \frac{1}{4}
        \frac{e^{-it\alpha}}{1-\lambda e^{-i\alpha}}
    \right)
    \nonumber\\
    &=
    \sum_{\ell=t}^{\infty}
    \lambda^\ell
    \sin^2\!\left(\frac{\ell\alpha}{2}\right)
    =
    \Phi(t).
    \label{eq:fdrn-qutrit-survival-tail}
\end{align}
The case $t=0$ gives $\operatorname{Tr}\rho_0=1=\Phi(0)$. Define the normalized
clock memories
\begin{align}
    \rho_t
    :=
    \frac{\widetilde\rho_t}{\Phi(t)},
    \qquad
    t\geq0 .
    \label{eq:fdrn-normalized-qutrit-clock-states}
\end{align}
For $t\geq1$, these states are pure
\begin{align}\label{eq:memory_states_fdrn_pure}
\rho_t
=
|\eta_t\rangle\langle\eta_t|,
\qquad
|\eta_t\rangle
=
\frac{
e^{-rX}e^{it\alpha Z/2}e^{rX}|\xi\rangle
}{
\|
e^{-rX}e^{it\alpha Z/2}e^{rX}|\xi\rangle
\|
} .
\end{align}
A conditional rollout after $t\geq1$ consecutive Wait--Tick outcomes may
therefore be initialized by preparing $|\eta_t\rangle$ in the qubit
subspace, while a reset initializes the mixed state $\rho_0$. Moreover, for every $t\geq1$, there exists a qutrit unitary $U_t$ such that
\begin{align}
U_t|0\rangle
&=
|\eta_t\rangle,
&
\rho_t
&=
U_t|0\rangle\langle0|U_t^\dagger.
\end{align}
Thus the clock state required after $t$ consecutive Wait--Tick outcomes can
be prepared directly. The mixed reset state $\rho_0$ is initialized by its
fixed preparation procedure.

Let $E_Q$ be the canonical encoder associated with the instrument products
as in~\eqref{eq:encoder_quantum}, choosing
$\rho_{\mathrm{ref}}=\rho_0$ on zero-probability histories. Since every
nonzero instrument operation other than Wait--Tick resets the memory, the
encoder satisfies
\begin{align}
E_Q(h)
&=
\rho_{\ell(h)},
\qquad
h\in\mathscr H_\star .
\label{eq:fdrn-qutrit-encoding}
\end{align}
In particular,
$E_Q(\ph_t)=\rho_t$. 

\subsection{Exactness of the qutrit world model}\label{sec:exact_qutrit_worldmodel}

The reset state~\eqref{eq:fdrn-qutrit-reset-state} and the instrument
operations
\eqref{eq:fdrn-qutrit-wait-branches}--%
\eqref{eq:fdrn-qutrit-probe-branches}
define the three-dimensional quantum world model
\begin{align}
\mathsf Q_{\mathrm{FRDN}}
&:=
\left(
\mathcal A_{\mathrm{FRDN}},
\mathcal Y_{\mathrm{FRDN}},
\mathcal H_Q,
\rho_0,
\boldsymbol{\mathcal E}_{\mathrm{FRDN}}
\right),
\nonumber\\
\boldsymbol{\mathcal E}_{\mathrm{FRDN}}
&:=
\left\{
\mathcal E_y^{(a)}
\right\}_{(a,y)\in
\mathcal A_{\mathrm{FRDN}}\times\mathcal Y_{\mathrm{FRDN}}},
\end{align}
where all unlisted instrument operations are the zero map. The encoder
$E_Q$ in~\eqref{eq:fdrn-qutrit-encoding} is the separate initialization
interface used for history-indexed queries. We use $\mathrm{Pr}_Q$ and $T_Q$ below as
shorthand for the outcome law and conditional updates induced by
$\boldsymbol{\mathcal E}_{\mathrm{FRDN}}$.

Using
$\widetilde\rho_{t+1}
 =\mathcal E_{y_T^W}^{(W)}(\widetilde\rho_t)$ and
$\operatorname{Tr}\widetilde\rho_t=\Phi(t)$, the nonzero branch identities
are, for every $t\geq0$,
\begin{align}
    \mathrm{Pr}_Q(y_T^W\mid\rho_t,W)
    &=
    \frac{\operatorname{Tr}\widetilde\rho_{t+1}}{\Phi(t)}
   =
    \frac{\Phi(t+1)}{\Phi(t)}
    =
    S(t),
    \nonumber\\
    T_{Q,y_T^W}^{W}(\rho_t)
    &=
    \rho_{t+1},
    \nonumber\\
    \mathrm{Pr}_Q(y_B^W\mid\rho_t,W)
    &=
    1-S(t),
    \nonumber\\
    T_{Q,y_B^W}^{W}(\rho_t)
    &=
    \rho_0,
    \nonumber\\
    \mathrm{Pr}_Q(y_B^M\mid\rho_t,M)
    &=
    1,
    \nonumber\\
    T_{Q,y_B^M}^{M}(\rho_t)
    &=
    \rho_0,
    \nonumber\\
    \mathrm{Pr}_Q(y_T^P\mid\rho_t,P)
    &=
    S(t),
    \nonumber\\
    T_{Q,y_T^P}^{P}(\rho_t)
    &=
    \rho_0,
    \nonumber\\
    \mathrm{Pr}_Q(y_B^P\mid\rho_t,P)
    &=
    1-S(t),
    \nonumber\\
    T_{Q,y_B^P}^{P}(\rho_t)
    &=
    \rho_0 .
    \label{eq:fdrn-qutrit-branch-identities}
\end{align}

We write $V_Q^\pi(h)$ and $Q_Q^\pi(h,a)$ for the history-indexed policy
values of $\mathsf Q_{\mathrm{FRDN}}$ defined in
Definition~\ref{def:model_policy_value_functions}, and
$V_Q^{\mathrm{opt}}(h)$ and $Q_Q^{\mathrm{opt}}(h,a)$ for its
history-indexed optimal values defined in
Definition~\ref{def:model-optimal-planning-values}. The optimal
quantities have the equivalent memory-state representations given by
\eqref{eq:model-optimal-value-bridge} and
\eqref{eq:model-optimal-q-bridge}.

\begin{theorem}[Exact qutrit realization]
\label{thm:qutrit_realization_exactplanning}
The tuple $\mathsf Q_{\mathrm{FRDN}}$ defined above is a quantum world model
of memory dimension $3$ in the sense of
Definition~\ref{def:quantum-world-model}. Together with the encoder
$E_Q$ in~\eqref{eq:fdrn-qutrit-encoding}, it forms an exact
model--encoder pair for the FRDN true world
$\mathsf W_\star^{\mathrm{FRDN}}$ of
Definition~\ref{def:fdrn-controlled-world}, in the sense of
Definition~\ref{def:predictive_consistency}.

Consequently, for every discount factor $\gamma\in[0,1)$, policy
$\pi\in\Pi_{\mathscr H}$, reachable history
$h\in\mathscr H_\star$, and action
$a\in\mathcal A_{\mathrm{FRDN}}$, the policy values defined in
Definition~\ref{def:model_policy_value_functions} satisfy
\begin{align}
V_Q^\pi(h)
&=
V_\star^\pi(h),
&
Q_Q^\pi(h,a)
&=
Q_\star^\pi(h,a).
\end{align}
The optimal quantities defined in
Definitions~\ref{def:true-optimal-planning-values}
and~\ref{def:model-optimal-planning-values} likewise satisfy
\begin{align}
V_Q^{\mathrm{opt}}(h)
&=
V_\star^{\mathrm{opt}}(h),
&
Q_Q^{\mathrm{opt}}(h,a)
&=
Q_\star^{\mathrm{opt}}(h,a).
\end{align}
Hence, by
\eqref{eq:fdrn-action-value-error}
and~\eqref{eq:fdrn-optimal-value-error},
\begin{align}
e_{\mathrm Q}^{Q}(h)
&=
e_V^{Q}(h)
=
0
\end{align}
after every reachable history.

For the rewards~\eqref{eq:fdrn-decision-rewards}, every
$\gamma\in[0,1)$, and every reachable history
$h\in\mathscr H_\star$, each model-greedy action
$\widehat a_Q(h)$ satisfying
\eqref{eq:fdrn-model-greedy-action} is true-world optimal
\begin{align}
\widehat a_Q(h)
\in
\operatorname*{arg\,max}_{a\in\mathcal A_{\mathrm{FRDN}}}
Q_\star^{\mathrm{opt}}(h,a).
\end{align}
The corresponding decision margin and decision loss satisfy
\begin{align}
g_Q(h)
&=
g_\star(h),
&
\ell_Q(h)
&=
0.
\end{align}
Consequently, for every reachable action--outcome trajectory $\mathscr F$,
\begin{align}
\overline e_{\mathrm Q}^{Q}(\mathscr F)
&=
\overline e_V^{Q}(\mathscr F)
=
\overline\ell_Q(\mathscr F)
=
0.
\end{align}
\end{theorem}

\begin{proof}
Fix a reachable history $h\in\mathscr H_\star$ and set $t=\ell(h)$. By the
encoder construction~\eqref{eq:fdrn-qutrit-encoding},
\begin{align}
E_Q(h)
&=
\rho_t.
\end{align}
Using the outcome and update rules
\eqref{eq:quantum-world-model-probability}
and~\eqref{eq:quantum-world-model-update}, the branch identities
\eqref{eq:fdrn-qutrit-branch-identities}, and the zero maps for all
unlisted outcomes, the qutrit model reproduces the true-world
kernel~\eqref{eq:fdrn-true-kernel}. Therefore, for every
$a\in\mathcal A_{\mathrm{FRDN}}$ and
$y\in\mathcal Y_{\mathrm{FRDN}}$,
\begin{align}
\mathrm{Pr}_Q\!\left(y\mid E_Q(h),a\right)
&=
\mathrm{Pr}_\star(y\mid h,a).
\end{align}
This is precisely the probability condition
~\eqref{eq:predictive-consistency-probability} in
Definition~\ref{def:predictive_consistency}.

The same branch identities reproduce the clock recursion
~\eqref{eq:fdrn-history-clock}: Wait--Tick maps $\rho_t$ to
$\rho_{t+1}$, whereas Wait--Break, Maintain, and either Probe outcome prepare
$\rho_0$. Hence, whenever $\mathrm{Pr}_\star(y\mid h,a)>0$,
\begin{align}
T_{Q,y}^{a}\!\left(E_Q(h)\right)
&=
E_Q(hay).
\end{align}
This is the memory-update condition
~\eqref{eq:predictive-consistency-update}. The pair $(\mathsf Q_{\mathrm{FRDN}},E_Q)$ is therefore exact on the complete
reachable history tree according to
Definition~\ref{def:predictive_consistency}.

The policy-value equalities now follow from
Theorem~\ref{thm:value_equivalence_exact_world_models}. Taking the suprema
over the same class of history policies gives the corresponding equalities
for the optimal values.

Equality of the quantum and true-world optimal action-values implies that
every action satisfying the model-greedy rule
\eqref{eq:fdrn-model-greedy-action} is true-world optimal. Substitution into
\eqref{eq:fdrn-decision-margins}
and~\eqref{eq:fdrn-decision-loss} gives
\begin{align}
g_Q(h)
&=
g_\star(h),
&
\ell_Q(h)
&=
0.
\end{align}
Since these identities hold after every reachable history, the corresponding
trajectory means vanish along every reachable action--outcome trajectory.
\end{proof}

\section{Classical baseline }

\subsection{Numerical procedures for fitted HMM}\label{app:FRDN-classical-numerical}

We numerically optimize the finite-dimensional classical world models introduced above.
Specifically, we specialize the classical world model definition in Definition~\ref{def:classical-world-model} to the Wait--Tick branch and focus on
$D_T := D_{y_T^W}^{(W)}$,
where $W$ denotes the Wait action and $y_T^W$ denotes the Tick outcome under Wait. The initial hidden distribution over states is
$z_0\in\mathcal{P}([N])$, and $D_T\in\mathbb{R}_{\geq 0}^{N\times N}$ is column-substochastic:
$\mathbf{1}^{\top}D_T \leq \mathbf{1}^{\top}.$
For the consecutive Wait--Tick history $\ph_t$, the classical
survival probability is
\begin{align}
\Phi_C(t)
:=
\left\langle\mathbf 1,D_T^t z_0\right\rangle .
\end{align}
Whenever $\Phi_C(t)>0$, its conditional probability of one further Tick is
\begin{align}
S_C(t)
&=
\mathrm{Pr}_C\!\left(
    y_T^W
    \mid
    E_C(\ph_t),W
\right)
=
\frac{\Phi_C(t+1)}{\Phi_C(t)}.
\end{align}
This approximates the true conditional probability
\begin{align}
S(t)
&=
\mathrm{Pr}_\star\!\left(
    y_T^W
    \mid
    \ph_t,W
\right)
=
\frac{\Phi(t+1)}{\Phi(t)}.
\end{align}
We optimize $D_T$ and $z_0$ by minimizing
the weighted Bernoulli cross-entropy
\begin{align}
    \mathcal{L}
    =
    -\sum_{t=0}^{T_{\mathrm{fit}}} w_t
    \left[
        S(t)\log S_C(t)
        +
        \bigl(1-S(t)\bigr)
        \log\bigl(1-S_C(t)\bigr)
    \right].
\end{align}

For each classical memory dimension $N$, we fix the fitting horizon to
$T_{\mathrm{fit}}=1000$ and fit the model on the complete age grid
$t=0,\ldots,T_{\mathrm{fit}}$. Thus, each optimization step uses
$1001$ analytically evaluated conditional Tick probabilities. We use
the natural-frequency weighting
$ w_t=\frac{\Phi(t)}
{\sum_{s=0}^{T_{\mathrm{fit}}}\Phi(s)},
$ which weights each age according to its occurrence probability under the uncontrolled renewal process.

We optimize the initial distribution over states $z_0$ and the Tick-transition matrix
$D_T$ using Adam with learning rate $10^{-3}$ for $3000$ optimization
steps. After each Adam update~\cite{kingma2015adam}, we project $z_0$ onto
$\mathcal{P}([N])$ and each augmented transition column onto
$\mathcal{P}([N+1])$. This projected-gradient procedure preserves
non-negativity, normalization of the initial distribution over states, and
column-substochasticity of $D_T$, while allowing exact zero entries.
For each $N$, we perform $10$ independent random restarts and retain
the feasible iterate with the lowest training loss across all
restarts. The resulting parameters provide a best-found classical approximation of memory dimension $N$ to the FRDN conditional Tick sequence.

After fitting, the selected model is evaluated without further optimization or model selection. The evaluation horizon is specified
separately for each figure to display the relevant short- or long-horizon behavior.

\subsection{Additional numerical results}

To benchmark finite classical memory beyond numerically optimized hidden Markov models, we additionally consider two explicit constructive baselines. The first is an $N$-state \emph{age counter}, whose memory states encode the elapsed clock age $0,1,\ldots,N-1$. For $t<N$, it reproduces the exact true-world Tick probability, $S_C(t)=S(t)$. Once the final state is reached, however, the memory saturates and the model continues to reuse the last hazard $S(N-1)$. The age counter is therefore exact over a finite initial horizon but necessarily loses the nonrepeating age dependence thereafter.

The second baseline is a \emph{rational phase cycle}, obtained by replacing the irrational phase advance $\alpha$ with a rational rotation
$
\beta = 2\pi\frac{p}{q},
\qquad \gcd(p,q)=1.
$
The model contains one reset state and a $q$-state cycle, giving total classical memory dimension $N=q+1$, and evaluates the same phase-dependent Tick profile along this rational orbit. Unlike the fitted HMM, this model is not learned but is constructed explicitly to approximate the irrational clock. Choosing $p/q$ from continued-fraction approximants to $\alpha/(2\pi)$ gives a phase error of order $O(q^{-2})$, so the rational cycle can shadow the true dynamics over increasingly long finite horizons, of order $q^2$ at fixed tolerance. It must nevertheless fail asymptotically because the rational orbit is periodic whereas the irrational orbit never closes. The resulting non-monotonicity in the running prediction error arises from phase beating between the two rotations.

For a dimension-matched comparison of the classical baselines, we use
$N=18$, $100$, and $479$. For $\alpha=\pi/\sqrt{2}$, the corresponding
continued-fraction approximants are
$
\frac{p}{q}=\frac{6}{17},\qquad
\frac{35}{99},\qquad
\frac{169}{478},
$
which give $N=q+1=18$, $100$, and $479$, respectively.

We first compare the three classical baselines at a fixed memory dimension $N=100$, with the results shown in Figure.~\ref{fig:different baseline comparsion}. The age-counter model reproduces the true-world probability
exactly for $t<N$ by construction, and hence its curve nearly coincides with
the true curve throughout the time window shown. The Rational phase cycle model also
closely tracks the true dynamics over this finite horizon, since the rational
phase $\beta=2\pi p/q$ provides a close approximation to the irrational phase
advance $\alpha$. Its prediction nevertheless eventually departs from the true
sequence as the phase mismatch accumulates. In contrast, the best-found HMM,
obtained by directly fitting the conditional Tick probabilities, reproduces
the initial variation but gradually loses the nonrepeating dependence on the
clock age and approaches an approximately constant limiting prediction.
Thus, even at the same classical memory dimension, the different constructions
can exhibit markedly different finite-horizon behavior, although none can
retain the irrational nonrepeating structure indefinitely.

To compare the long-horizon predictive performance of the different finite classical models, we consider the mean absolute prediction error $ E_C(T) :=\frac{1}{T}\sum_{t=1}^{T}
\left|S(t)-S_C(t)\right|$,
where $S(t)$ and $S_C(t)$ denote the true and model-predicted conditional Tick probabilities at clock age $t$, respectively.
Figure~\ref{fig:classical-cumulative-error} shows $E_C(T)$ for the three classical baselines at matched memory dimensions $N=18$, $100$, and $479$. The dashed horizontal line denotes the dimension-independent asymptotic lower
bound $\kappa_{\mathrm{FRDN}}$ (see Appendix~\ref{subsec:gap_conditional_probability} Eq.(~\ref{eq:kappa-fdrn}) for detailed information). For the age counter [Fig.~\ref{fig:classical-cumulative-error}(a)], the error is initially zero because the model stores the clock age exactly until its
finite memory saturates. Beyond this point, the prediction becomes fixed, and the cumulative error rises toward its long-horizon value. The rational phase
cycle [Fig.~\ref{fig:classical-cumulative-error}(b)] can shadow the irrational phase evolution for a substantially longer interval, with increasing memory delaying the onset of the discrepancy. Its non-monotonic behavior at later
times results from phase beating between the rational approximation and the
true irrational rotation. Finally, the fitted HMMs
[Fig.~\ref{fig:classical-cumulative-error}(c)] achieve smaller errors over the
fitted and intermediate horizons, but their long-horizon prediction error also
approaches the finite-classical asymptotic floor. Overall, increasing the classical memory can postpone the predictive discrepancy to larger clock ages, but does not eliminate the nonzero asymptotic error.

Additionally, we numerically demonstrate their accumulative wrong action fraction $ F_C(T)= \frac{1}{T} \sum_{t=1}^{T} \mathbf{1}
\!\left\{ \hat a_C(t)\notin A^\star(t) \right\}$,
where $\hat a_C(t)$ denotes the action selected by the classical model at
clock age $t$, and $A^\star(t)$ is the set of optimal actions in the true
world. Thus, $F_C(T)$ is the fraction of queried clock ages up to $T$ for
which the model-selected action is truly suboptimal. Figure~\ref{fig:wrong action fraction} shows the cumulative
wrong-action fraction for the age counter, Rational phase cycle model, and best-found
HMM at memory dimensions $N=18$, $100$, and $479$. In all three cases,
increasing the memory dimension primarily postpones the onset of the decision
mismatch rather than eliminating it. For the age counter, this behavior is
particularly transparent: the model reproduces the true age dependence over
a longer initial interval as $N$ increases, after which its action choice
eventually becomes fixed. The Rational phase cycle construction can shadow the
irrational phase dynamics for a substantially longer horizon, but eventually
develops oscillatory deviations due to the accumulating phase mismatch.
The fitted HMMs exhibit a different transient behavior, including an initial
overshoot in the wrong-action fraction, but likewise approach the same
long-time behavior. For all three baselines and all memory dimensions shown,
the cumulative wrong-action fraction approaches $1/2$, illustrating that
increasing finite classical memory can delay, but does not remove, the
persistent action-selection mismatch.

Figure~\ref{fig:decision loss} shows the corresponding mean decision loss (see Eq.~(\ref{eq:mean-classical-decision-loss-main}) for detailed information).
The qualitative behavior is consistent with that of the wrong-action fraction:
larger classical memories can substantially extend the horizon over which the
decision loss remains small, but a nonzero asymptotic loss persists.
The age-counter curves approach a common limiting loss after a memory-dependent
delay, while the Rational phase cycle models display transient oscillations associated
with the phase mismatch before settling to a nonzero long-time value.
The best-found HMMs again show pronounced transient behavior.

Interestingly, the long-horizon decision loss is not monotonic in the classical memory dimension $N$. This behavior can be understood from the limiting predictions of the fitted HMMs. Let
$ c_N := \lim_{t\rightarrow\infty} S_C(t).$
For the fitted models considered here, we obtain
$c_{18}\simeq0.4076$, $c_{100}\simeq0.3828$, and
$c_{479}\simeq0.4072$. Since the decision threshold is $\lambda=0.4$, the $N=18$ and $N=479$ models eventually remain on the Probe side of the
decision boundary, whereas the $N=100$ model remains on the Maintain side. Consequently, their asymptotic decision losses are determined by different
one-sided deviations of the true conditional probability from the threshold.
An eventually fixed Probe decision incurs the mean loss
$L_P =
\frac{1}{2\pi}
\int_0^{2\pi}
\bigl(\lambda-f(x)\bigr)_+\,dx,$ whereas an eventually fixed Maintain decision incurs
$L_M =
\frac{1}{2\pi}
\int_0^{2\pi}
\bigl(f(x)-\lambda\bigr)_+\,dx.$ For the parameters used here, $L_P\simeq0.0816$ while $L_M\simeq0.1660$. Hence the $N=100$ fitted HMM approaches a larger
decision-loss plateau than the $N=18$ and $N=479$ models. This non-monotonicity does not contradict the improved finite-horizon predictive accuracy obtained with larger memories: the fitted models are optimized for prediction, while their asymptotic decision loss additionally depends on which side of the decision boundary their limiting prediction occupies.

\begin{figure*}
	\centering
\centerline{\includegraphics[width=1\linewidth]{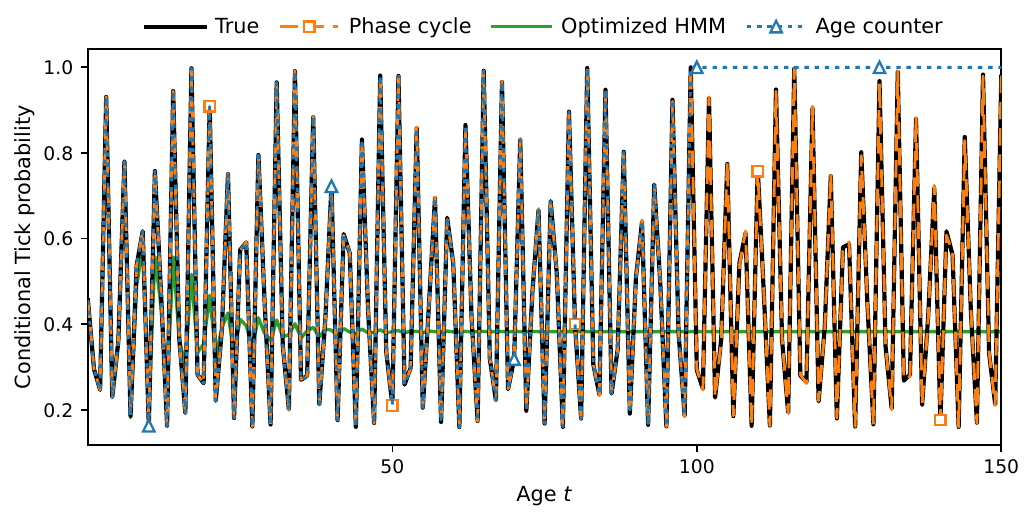}}
	\caption{\textbf{Conditional Tick probabilities of different models at system dimension $N=100$.}
Probability of one additional Tick after $\ph_t$ for
$\lambda=0.4$ and $\alpha=\pi/\sqrt{2}$. }
\label{fig:different baseline comparsion}
\end{figure*}

\begin{figure*}
	\centering
\includegraphics[width=1\linewidth]{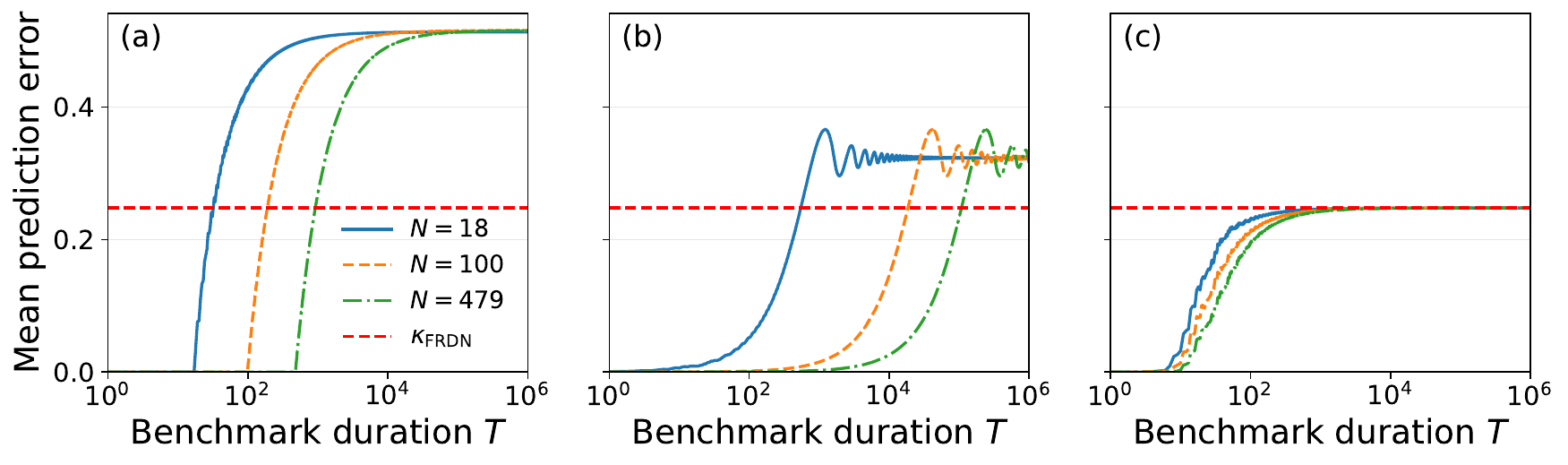}
 \caption{\textbf{Cumulative gaps in predicted conditional probabilities for models of different sizes.}
    (a) Age counter, (b) Rational phase cycle model, and (c) fitted HMM. The red dashed line corresponds to the asymptotic bound $\kappa_{FRDN}$ defined in Eq.(~\ref{eq:kappa-fdrn}). Please refer to  Appendix~\ref{subsec:gap_conditional_probability}) for detailed information. }
\label{fig:classical-cumulative-error}
\end{figure*}

\begin{figure*}
	\centering
\includegraphics[width=1\linewidth]{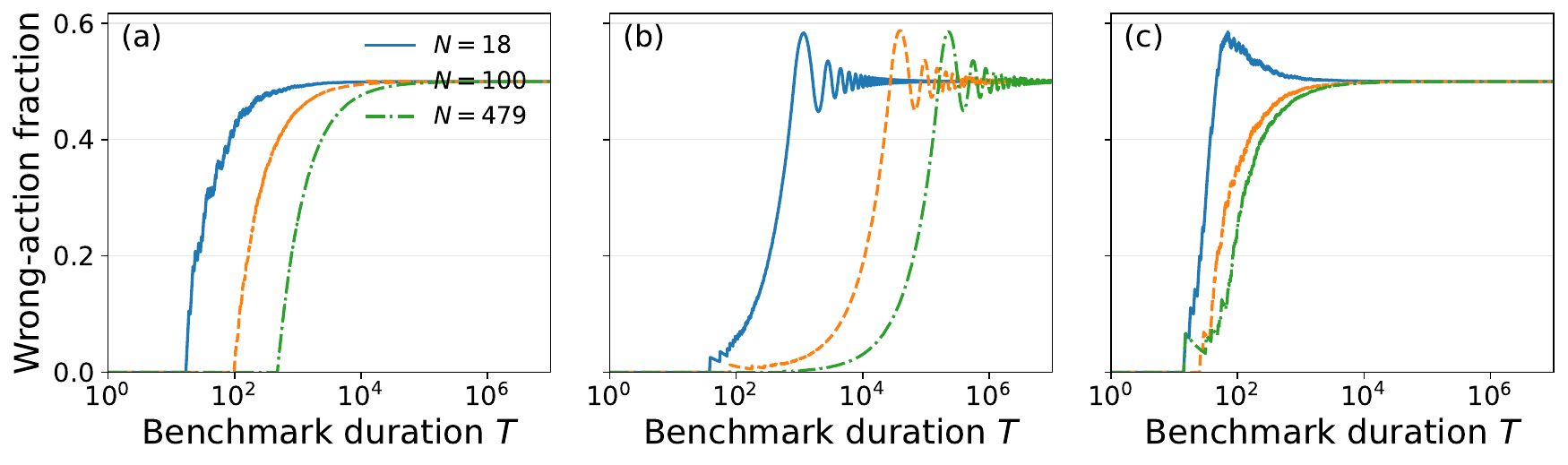}
\caption{\textbf{Cumulative wrong-action fraction for models of different sizes.} (a) Age-counter model, (b) Rational phase cycle model, and (c) fitted HMM.}
\label{fig:wrong action fraction}
\end{figure*}

\begin{figure*}
	\centering
\includegraphics[width=1\linewidth]{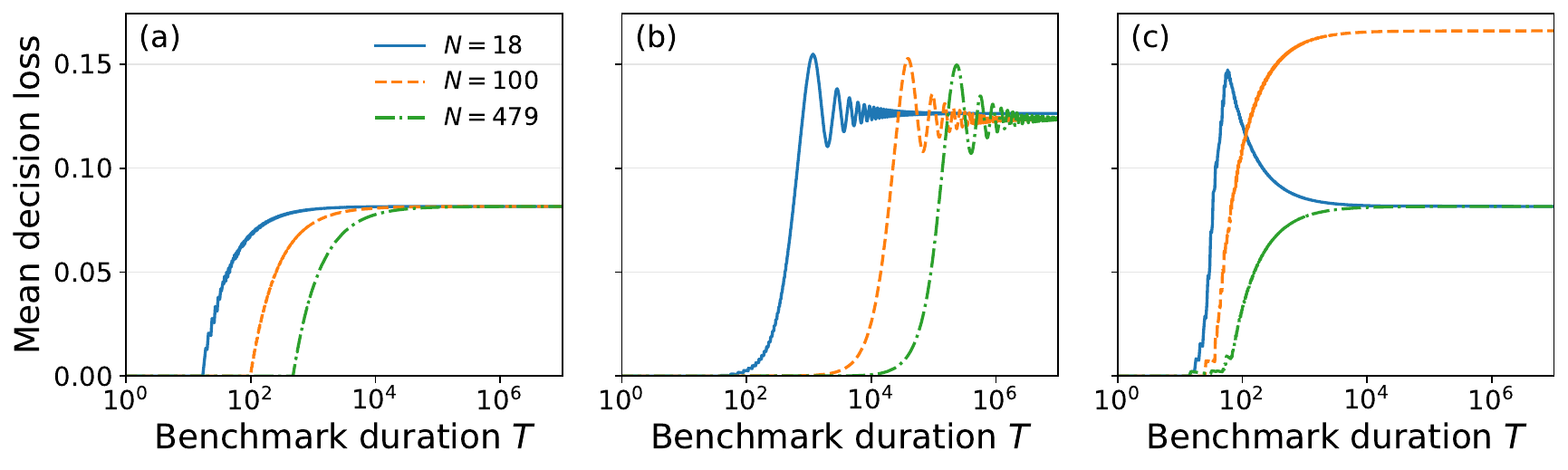}
  \caption{\textbf{Mean decision loss of different models with different sizes.} (a) Age-counter model, (b) Rational phase cycle model, and (c) fitted HMM. }
\label{fig:decision loss}
\end{figure*}

\bibliographystyle{ultimateunsrt}
\bibliography{biblio_qagents}

\end{document}